\documentclass{vldb}
\usepackage{graphicx}
\usepackage{balance}  

\usepackage[linesnumbered,ruled,vlined]{algorithm2e}
\usepackage[normalem]{ulem}
\graphicspath{ {pic/} }

\usepackage{amsthm}
\usepackage{subcaption}
\usepackage{multirow}
\usepackage{xcolor}
\usepackage[nocompress]{cite}

\usepackage{enumitem}
\usepackage[skip=4pt]{caption}
\usepackage[skip=4pt]{subcaption}
\newtheorem{proper}{Theorem}
\newtheorem{prop}[proper]{Property}

\newtheorem{propos}{Theorem}
\newtheorem{proposition}[propos]{Proposition}

\newtheorem{defi}{Theorem}
\newtheorem{definition}[defi]{Definition}
\newtheorem{example}{Theorem}
\newtheorem{ex}[example]{Example}

\newtheorem{theorem}{Theorem}

\vldbTitle{A Sample Proceedings of the VLDB Endowment Paper in LaTeX Format}
\vldbAuthors{Ben Trovato, G. K. M. Tobin, Lars Th{\sf{\o}}rv{$\ddot{\mbox{a}}$}ld, Lawrence P. Leipuner, Sean Fogarty, Charles Palmer, John Smith, Julius P.~Kumquat, and Ahmet Sacan}
\vldbDOI{https://doi.org/10.14778/xxxxxxx.xxxxxxx}
\vldbVolume{12}
\vldbNumber{xxx}
\vldbYear{2019}

\begin{document}


\title{Evaluating Temporal Queries Over Video Feeds}



%
%
%
%

\numberofauthors{3} 

\author{
%
%
\alignauthor
Yueting Chen\\
       \affaddr{York University}\\
       \email{ytchen@eecs.yorku.ca}
\alignauthor
Xiaohui Yu\\
       \affaddr{York University}\\
       \email{xhyu@yorku.ca}
\alignauthor 
Nick Koudas\\
       \affaddr{University of Toronto}\\
       \email{koudas@cs.toronto.edu}
}


\maketitle

\begin{abstract}

Recent advances in Computer Vision and Deep Learning made possible the efficient extraction of a schema from frames of streaming video. As such, a stream of objects and their associated classes along with unique object identifiers derived via object tracking can be generated,  providing unique objects as they are captured across frames.

In this paper we initiate a study of temporal queries involving objects and their co-occurrences in video feeds. For example, queries that identify video segments during which the same two red cars and the same two humans appear jointly for five minutes are of interest to many applications ranging from law enforcement to security and safety. We take the first step and define such queries in a way that they incorporate certain physical aspects of video capture such as object occlusion. We present an architecture consisting of three layers, namely object detection/tracking, intermediate data generation and query evaluation. 
We propose two techniques, Marked Frame Set (MFS) and Strict State Graph (SSG), to organize all detected objects in the intermediate data generation layer, which effectively, given the queries, minimizes the number of objects and frames that have to be considered during query evaluation. We also introduce an algorithm called State Traversal (ST) that processes incoming frames against the SSG and efficiently prunes objects and frames unrelated to query evaluation, while maintaining all states required for succinct query evaluation. 

We present the results of a thorough experimental evaluation utilizing both real and synthetic data  establishing the trade-offs between MFS and SSG. We stress various parameters of interest in our evaluation and demonstrate that the proposed query evaluation methodology coupled with the proposed algorithms is capable to evaluate temporal queries over video feeds efficiently, achieving orders of magnitude performance benefits.
\end{abstract}

\section{Introduction}

Video data abound. As of this writing the equivalent of 500 hours of video is uploaded to Youtube alone, every minute. Since the majority of video feeds will never be inspected by humans, analysis of video, the ability to extract meaningful information from video feeds as well as the ability to express and evaluate different types of video analytics queries has become a pressing concern.

Recent advances in Deep Learning (DL) \cite{DBLP:journals/nature/LeCunBH15,DBLP:books/daglib/0040158} technologies introduced algorithms and models to comprehend challenging data types such as images, video and unstructured text (among many others); recent breakthroughs in applications such as image classification and object detection offer the ability to classify objects in image frames, detect object locations in frames as well as track objects from frame to frame with solid accuracy. 

There has been increasing recent interest in the community \cite{Kang:2017:NON:3137628.3137664,blazeit,DBLP:conf/cidr/KangBZ19, xarchakos, xarchakos1} in the processing of video queries utilizing DL models as primitives for data extraction from videos. We are also concerned with applications that involve video collected from static cameras as it is prevalent in surveillance and security applications. Our focus however is on the efficient execution of {\em temporal queries} over video feeds.  A video feed is an ordered sequence of frames, presented at a specific {\em frame rate}, typically 30 frames per second ({\em fps}) to align with human visual perception. Different types of objects appear in frames and can be easily detected by state-of-the-art object detection algorithms. Similarly, DL models are highly effective to track the same object (e.g., car) from one frame to the next, as objects persist across frames. Given this basic ability to detect and track objects across the temporal evolution of a video feed, several queries of interest arise. For example, we may be interested to discover video segments during which only the exact same two cars (and no other object) appear continuously within a 10 second interval; similarly we may wish to detect video segments during which the same three trucks and the same human object appear for a duration of one minute jointly (possibly along with many other cars/trucks and humans among them). These queries become highly sophisticated when the query objects can be associated in arbitrary Conjunctive Normal Form (CNF) expressions. Moreover, when query specified objects can be connected to external identities (e.g., license plates) such queries become highly powerful. 

\begin{figure}[t!]
    \centering
    \begin{subfigure}[t]{0.25\textwidth}
        \centering
        \includegraphics[height=70pt]{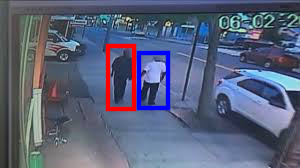}
        \caption{Two Men}
        \label{fig:qex:1}
    \end{subfigure}%
    ~ 
    \begin{subfigure}[t]{0.25\textwidth}
        \centering
        \includegraphics[height=70pt]{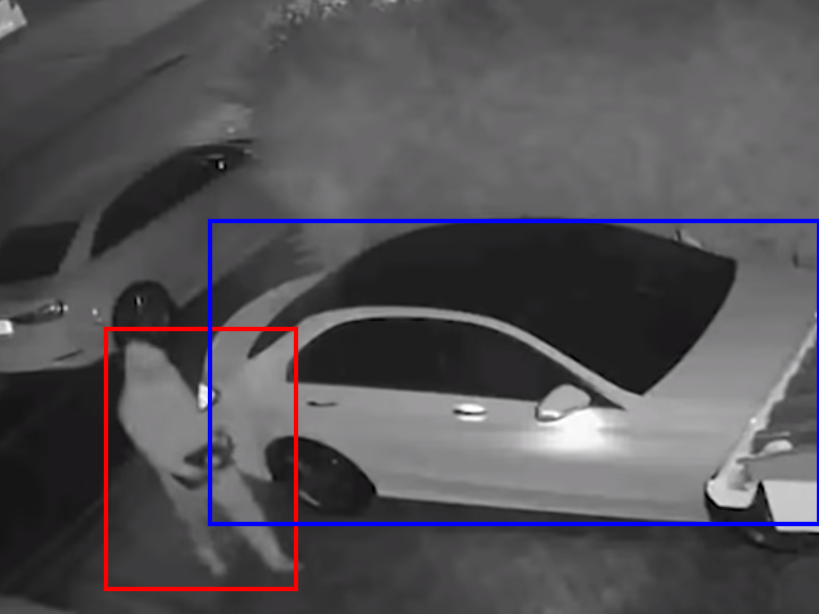}
        \caption{A Car Thief}
        \label{fig:qex:2}
    \end{subfigure}
    \caption{Video Examples}
    \label{fig:qex}
\end{figure}

As a sample application scenario, consider surveillance analysis, where a lot of video footage is examined after an incident for certain events of interest involving joint presence of objects. For example, after an incident witnesses may report a white car and two males on the street. Subsequently, authorities may wish to detect in a large video footage segments where a white car and two humans appear jointly for at most five minutes (e.g., Figure \ref{fig:qex}a). Efficient execution of such queries can automate many of these tasks.

Evaluating queries on video feeds involving multiple objects and temporal constraints presents numerous challenges, however.
First, finding video frames where the above query (a white car and two males appear jointly) is satisfied can be very expensive since both the actual objects and frames that satisfy the query are uncertain, which means that a brute-force approach has to consider all possible object combinations and frames, thus leading to a large computation space. 
Sets of objects that are not promising to be part of any query answer should be pruned immediately from further evaluation.
Second,
{\em Occlusions} occur naturally among objects in videos, where some objects disappear from the visible screen in some frames (because of being blocked by other objects) and then reappear. For example, in Figure \ref{fig:qex}b the human may disappear behind the car for a few seconds or longer. Occlusions can be captured by state-of-art object tracking algorithms \cite{DBLP:conf/itsc/KrebsDF17, Wojke2017simple}, and thus they need to be considered in the query semantics. Similarly, although state-of-the-art object tracking algorithms perform well when tracking an object across frames, they are still error prone. As a result, some objects may disappear or are assigned new identifiers in some frames. Temporal query semantics has to account for pragmatic challenges such as occlusion, and subsequent query evaluation needs to take such semantics into account.

In this paper, we place the problem of evaluating temporal queries over video feeds into perspective, and in particular make the following contributions:
\begin{itemize}[noitemsep]
	\item We define semantics for temporal queries on video feeds accounting for aspects intrinsic to video objects such as object occlusion.
	\item We decouple object detection/classification from query evaluation, introducing an intermediate layer (MCOS Generation) that optimizes the input to the query evaluation module.
	\item We fully develop the MCOS Generation module, detailing a way to organize objects detected through state maintenance. 
	\item We propose {\em Marked Frame Set} (MFS) to prune unnecessary states early. Additionally, to provide even better pruning efficiency in certain scenarios, we also propose to maintain states in a graph structure called {\em Strict State Graph} (SSG). 
	\item We propose an algorithm called {\em State Traversal} ($ST$) that processes incoming frames against the SSG, dynamically adjusting its structure and incrementally and efficiently maintaining all primitives necessary for succinct query evaluation. 
	\item We propose several modifications to a CNF query execution algorithm, and demonstrate how to efficiently couple CNF query execution with algorithm ST, yielding improved performance in our application scenario.
	\item We present a detailed experimental evaluation of the proposed algorithms utilizing both synthetic and real data,
    demonstrating the trade-offs between MFS and SSG. Both approaches for state maintenance yield significant performance improvements over baselines on real datasets, with benefits increasing as query parameters vary.  We also demonstrate that adapting our state maintenance during query evaluation yields an optimized approach that exhibits significant speedup on query evaluation; up to more than 100 times in our experiments.
\end{itemize}

The remainder of this paper is organized as follows. In Section \ref{sect:problem}, we formally introduce the problem of interest in this work,  followed by Section \ref{sect:overview} which presents an overview of our solutions.
Section \ref{sect:state} presents an overview of the proposed MCOS Generation module along with our main technical proposals. In Section \ref{sect:eval}, we discuss query evaluation. Section \ref{sect:exp} presents a thorough experimental evaluation of all introduced approaches. Section \ref{sect:related} reviews related work, and Section \ref{sect:conclusion} concludes the paper. 
All proofs of theorems can be found in the appendix.

\section{Problem Definition} \label{sect:problem}

Consider a video feed containing 5 frames, $(\{B\}, \{ABC\},\\ \{ABDF\}, \{ABCF\}, \{ABD\})$,
where each letter represents an object, and each object set represents a frame in order.
Assume all objects are of the same class. 
Queries such as ``select the object set that appears in all frames" can be easily answered by intersecting the object set of all frames directly. 
However, if we allow occlusions/detection errors and issue queries such as ``Select the video frames where some objects appear jointly for at least 3 frames in a window of 5 frames", the problem becomes nontrivial.
In this case, the object sets $\{B\}$ (appearing in frames $\{0,1,2,3,4\}$) and $\{AB\}$ (in frames $\{1,2,3,4\}$) should be selected. If we further relax the duration threshold from 3 frames to 2 frames, three more object sets ($\{ABC\}$, $\{ABD\}$, $\{ABF\}$) should also be selected. In real-world applications, such queries may further be restricted to a certain number of objects (e.g. at least 2 persons and 1 car to appear jointly).

To formalize the problem, we consider a video feed $V$ as a bounded sequence of frames, i.e., $V = \langle f_0, \ldots,  f_{i}, \ldots, f_{N-1}\rangle$, where $N$ is the total length of the video.
We use $\mathbb{ID}$ to denote the set of possible object identifiers (id) that can appear in the video feed (unique identifiers that are persistent for each detected object across video frames). Typically such identifiers are extracted by the application of object tracking models \cite{DBLP:conf/itsc/KrebsDF17}. We use $\mathbb{L}$ to denote the set of class labels for all objects. Class labels (e.g., human, car) are extracted by the application of object classification/detection models on each frame.

Utilizing object detection and tracking algorithms, we can extract a set of objects $O_i$ from each frame $f_i$, and for each  $o_i \in O_i$, its id $id(o_i)$ and its class label $label(o_i)$. Thus, we can obtain a structured relation $VR$ from the video feed $V$, with schema $(fid, id, class)$, where $fid$ is the frame id $i$ ($i \in [0, N]$), $id\in \mathbb{ID}$ is the object id, and $class\in \mathbb{L}$ is the class label.  Object tracking algorithms ensure that an object has the same $id$ for all frames in which it appears, even under the presence of occlusions \cite{mei2011minimum, Wojke2017simple, wu2013online}. 

Given a window size $w$, we use $VR_{i,w}$ to denote the structured relation obtained from the frames in the current window,  $VR_{i,w} = \langle f_{max(0, i-w)}, \ldots,  f_i\rangle$.
To capture object co-occurrence across frames in $VR_{i,w}$, we introduce a new predicate, $cooc(\mathbb{ID}_q, f)$, to determine whether the objects specified in the set $\mathbb{ID}_q$ are co-occurring in frame $f$.  That is, $cooc(\mathbb{ID}_q, f)$ 
is TRUE iff $\mathbb{ID}_q \subseteq \Pi_{id}$$(\sigma_{fid=f}(VR_{i,w}))$. 
We use $\mathbb{F}$ to denote the set containing all the frame ids in the current window, $\mathbb{F} = \Pi_{fid}(VR_{i,w})$.
We slightly abuse terminology and use the term frame and frame id interchangeably as all frames are unique.

A co-occurrence object set is a set of objects that all appear jointly in certain frames, defined as follows:

\begin{definition}\textbf{Co-occurrence Object Set (COS).}
Given a frame set $\mathbb{F'}$, where $\mathbb{F'} \subseteq \mathbb{F}$, if $\exists \mathbb{ID'} \subseteq  \mathbb{ID}$, $\mathbb{ID'} \ne \emptyset$, such that $\forall f \in \mathbb{F'}$, $cooc(\mathbb{ID'}, f)$ is TRUE, then $\mathbb{ID'}$ is a co-occurrence object set of frame set $\mathbb{F'}$.
\end{definition}

For a given frame set $\mathbb{F'}$, an object set is a maximum co-occurrence object set if and only if none of its supersets is a co-occurrence object set of $\mathbb{F'}$, defined as follows:
 
\begin{definition}\textbf{ Maximum Co-occurrence Object Set\\ (MCOS).}
Suppose $\mathbb{ID'}$ is a co-occurrence object set of frame set $\mathbb{F'}$. $\mathbb{ID'}$ is the maximum co-occurrence object set iff $\forall \mathbb{ID''} \subseteq \mathbb{ID}$, $\mathbb{ID'} \subset \mathbb{ID''}$, such that $\exists f \in \mathbb{F'}$, $cooc(\mathbb{ID''}, f)$ is FALSE.
\end{definition}

In the example mentioned in the beginning of this section, multiple MCOSs exist in the window of 5 frames. For example, object set $\{AB\}$ is an MCOS of frame set $\{1,2,3,4\}$, while $\{ABC\}$ is an MCOS of frame set $\{1,3\}$.

We wish to support arbitrary {\em Conjunctive Normal Form} (CNF) queries over the video feeds. We adopt sliding window query semantics \footnote{Other options are possible, such as tumbling window, and our solution will work equally well with such semantics.}; each query is expressed along with a temporal window context $w$. The window advances every time a new frame is encountered and the query result is evaluated over the most recently encountered $w$ frames.

Each condition $c$ in the CNF query is of the form $x$ $\theta$ $n$, where $x$ is the class label of an object, $n$ is an integer value and $\theta \in \{\le, =, \ge \}$. For example, $'car' \ge n$ specifies the number of cars to be greater or equal to value $n$. Bounded range queries can also be constructed based on conjunctions, e.g., $'person' \le 5 \land 'person' \ge 3$ queries whether the number of people is between 3 and 5. Since objects can appear and disappear in the visible screen due to occlusion, we introduce a {\em duration parameter} $d$ to each query, where $0 \le d \le w$. Without occlusion an object should appear in the visible screen for $w$ frames; with occlusion, however, the number of appearances can be less than $w$. The duration parameter captures this and specifies the minimum number of frames an object (more generally an MCOS) should appear in the current window.  

Given a duration parameter $d$ and a window size $w$, a condition $c$ will be evaluated as TRUE iff there are frames $\mathbb{F}' \subseteq \mathbb{F}$, such that the following conditions are satisfied.

\begin{itemize}[noitemsep, nolistsep]
    \item $|\mathbb{F'}| \ge d$.
    \item There exists some maximum co-occurrence object set $\mathbb{ID'} \subset \mathbb{ID}$, such that $\forall f \in \mathbb{F'}$: 
    
    $G_{count(id)}(\Pi_{id}(\sigma_{class=x \land cooc(\mathbb{ID'}, f)}(VR_{i,w}))) $ $\theta$ $n$. 
\end{itemize}

Since a query $q$ involves multiple conditions, the query result is a set of MCOSs, which includes all the possible frame sets where the overall CNF expression is evaluated to TRUE.

\section{Overview} \label{sect:overview}

\begin{figure}
    \centering
    \includegraphics[width=0.35\textwidth,height=5cm]{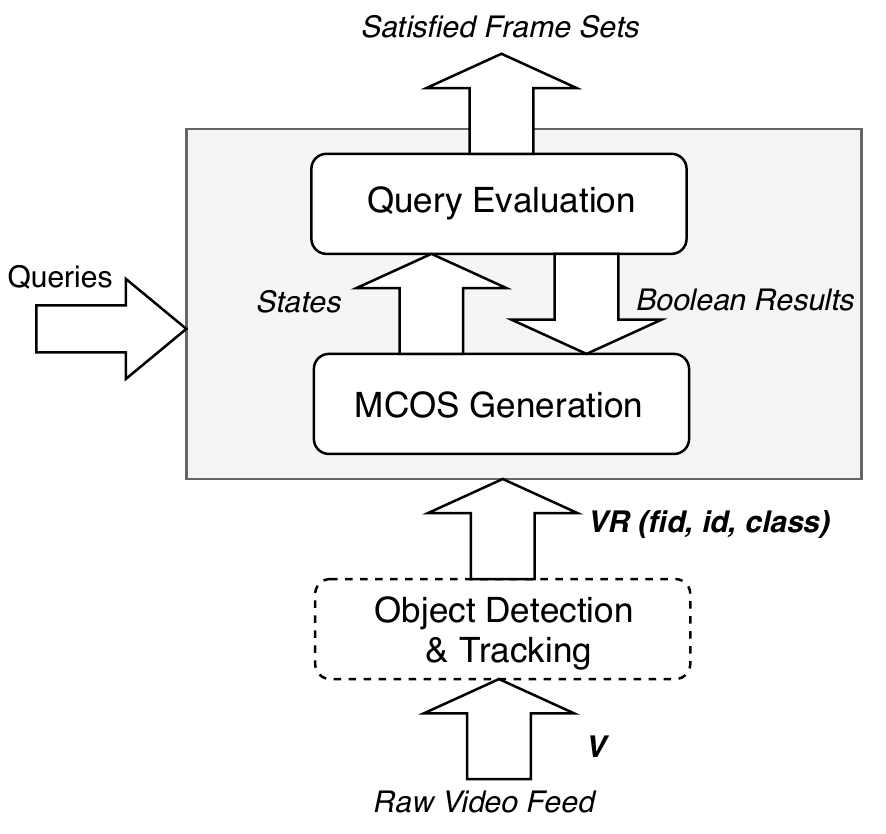}
    \caption{Overview}
    \label{fig:overview}
\end{figure}

The architecture of our overall solution is shown in  Figure~\ref{fig:overview}. Our solution consists of three modules, namely Object Detection and Tracking, Maximum Co-occurrence Object Set (MCOS) Generation, and Query Evaluation.

In the Object Detection/Tracking module, the input video feed, $V$, is processed to extract object ids and associated class labels. This module  consumes the raw video feed (frames) and produces a structured relation $VR$. In this module we deploy state-of-the-art object detection and tracking algorithms. Specifically, we use Faster R-CNN \cite{DBLP:journals/pami/RenHG017} as the object detection algorithm and Deep Sort \cite{Wojke2017simple} as the tracking algorithm.
The module is designed to be `plug and play' such that any algorithm from the computer vision community can be adopted as it becomes available. The output of this layer is a relation of tuples with schema $(fid, id, class)$ indicating that an object with identifier $id$ is of a specific $class$ and is detected at frame $fid$. This relation of tuples is relayed to the MCOS Generation module.

MCOS Generation is responsible for generating and materializing MCOS for query evaluation. 
The module computes MCOS incrementally by reusing intermediate results from the previous window.
It also maintains sufficient information to decide when an MCOS should be removed as the window slides.
To achieve this, the module utilizes information from queries, such as the window size $w$, duration condition $d$, and classes involved in the query. Queries with the same window size will be grouped together since they may share the same MCOS.
Objects with {\em class} not requested by any query may be dropped from $VR$ upon entering the MCOS generation module.
We introduce the notion of a \textit{State} to represent the basic unit that materializes an intermediate set of objects during generation, defined as follows:
\begin{definition}\textbf{State.} Let $\mathbb{F}$ be the set of all frames in the current window and $\mathbb{ID}$ the set of all object identifiers. A state $s=(\mathbb{ID_{\text{s}}}, \mathbb{F_\text{s}})$, where $\mathbb{ID_{\text{s}}} \subseteq \mathbb{ID}$, $\mathbb{F_\text{s}} \subseteq \mathbb{F}$, and $\mathbb{ID}_{s}$ is a COS of frame set $\mathbb{F}_{s}$.
\end{definition}
Each state consists of two types of information: an object set and the frames in which the objects appeared. Note that the object set may not be an MCOS in some states. A state, $s$, is a \textit{valid state} iff $\mathbb{ID}_{s}$ is an MCOS of frame set $\mathbb{F}_{s}$. Otherwise, it is an \textit{invalid state}. 
MCOS Generation should only produce valid states to Query Evaluation since queries are evaluated on MCOS only.

We use $\mathbb{S}$ to represent the set of all possible states in the current window. The object set in a state $s$ is an MCOS as long as $\forall s' \in \mathbb{S}$, if $\mathbb{F}_{s'} = \mathbb{F}_{s}$, then $\mathbb{ID}_{s'} \subset \mathbb{ID}_{s}$.

For each valid state, $s$, processing on the MCOS can be divided into three steps: 
\begin{enumerate}[noitemsep, nolistsep]
    \item Check against the duration parameter $d$. The query could be evaluated as TRUE only if $|\mathbb{F}_{s}| \ge d$. 
    \item Conditions in CNF expressions are aggregates evaluated on class labels (e.g., `car $\ge$ 3'). Thus, objects have to be aggregated after classification based on their class labels.
    \item Once all conditions of an expression have been processed, we can move on to evaluate CNF expressions. For this step we adapt CNF evaluation algorithms \cite{whang2009indexing}.
\end{enumerate}

In MCOS Generation module, simple optimizations can reduce number of states evaluated. 
For example testing the duration parameter $d$ can be conducted when the MCOS are generated (a form of push down optimization).
Namely, a state $s$ containing an MCOS can be pruned if $|\mathbb{F}_s| < d$ ({\em unsatisfied states}). 
States with MCOS that are not pruned ({\em satisfied states}) will be relayed to the Query Evaluation module. 
Since we are evaluating queries over a window, some states that are not satisfied at frame $i$ may be satisfied in the future as the window slides. 
Subsequently state maintenance, requires succinct book-keeping to assume both efficient execution and correctness.

Since queries are defined and evaluated on windows, we maintain states for each window and update these states incrementally.
Queries are analyzed first, objects that are not required for the given query will be discarded immediately.
States can be generated, satisfied, and invalid as new frames arrive. 
We present the MCOS Generation module in the next section. Then we discuss the Query Evaluation module in Section \ref{sect:eval}.

\section{MCOS Generation} \label{sect:state}

\subsection{Objective}
For each new arriving frame, 
the objective of MCOS generation is to maintain a set of states $\mathbb{S}$ and select a subset $\mathbb{S'} \subseteq \mathbb{S}$, where $\forall s' \in \mathbb{S'}$, $s'$ is a valid and satisfied state, while generating $|\mathbb{S'}|$ efficiently.
We refer to this procedure as \textbf{state maintenance}. 

\subsection{Marked Frame Set} \label{sect:mPS}

\subsubsection{Objective}

An MCOS can be obtained by intersecting two or more object sets from the current window. 
Generating MCOS by assessing object intersections across object sets of frames in each window can be very expensive. Since some MCOSs are shared between windows, we use states to keep the information of an MCOS, and derive new states from existing states computed in the previous window.

When the window slides, frames that are removed from the window are denoted as \textbf{expired}. For each existing state, we remove the expired frame ids from their frame sets. A state can be removed from the collection of states in the current window only if the frame set is empty.

\subsubsection{A First Attempt} \label{sect:attempt}

As a  first attempt to state maintenance, consider the following approach. 
Let the set of all states in the current window be $\mathbb{S}$, and assume that the last frame id in the current window is $i$.
When a new frame with id $i'=i+1$ arrives with associated object set $\mathbb{ID}_{o_{i'}}$, 
the following steps are conducted during state maintenance:

\begin{enumerate}[noitemsep]
    \item Append the new frame id: $\forall s \in \mathbb{S}$, if $\exists s' \in \mathbb{S}$, $\mathbb{ID}_{s} \cap \mathbb{ID}_{o_{i'}} = \mathbb{ID}_{s'}$, the new frame id can be simply appended to the existing state $s'$.
    \item Create new states: new states will be created as follows:
    \begin{enumerate}[noitemsep, nolistsep]
        \item For each state $s \in \mathbb{S}$, let $\mathbb{ID'} = \mathbb{ID}_{s} \cap \mathbb{ID}_{o_{i'}}$. If $\forall s' \in \mathbb{S}$, $\mathbb{ID'} \ne \mathbb{ID}_{s'}$, then a new state $ns$ is created with object set $\mathbb{ID'}$. In this case, $ \mathbb{F}_{s} \subset \mathbb{F}_{ns}$.
        \item If $\forall s \in \mathbb{S}$, $\mathbb{ID}_{s} \cap \mathbb{ID}_{o_{i'}} \ne \mathbb{ID}_{o_{i'}}$. This is the case where no state has the same object set as the arriving frame. A new state, $ns$, is created with the new frame id as the only element in $\mathbb{F}_{ns}$.
    \end{enumerate}
    
\end{enumerate}

\begin{ex}
\label{example-1}
Consider generating MCOS for single query with duration parameter $d=3$ frames and a window size of $w=4$ frames, as shown in Table \ref{table:example}. Each letter represents a detected object. The Frame and Object Set columns provide the frame id and associated object set obtained from the video feed. The column States depicts the set of states maintained in the current window, where satisfied states are underlined. Satisfied MCOSs are listed in the EXP(Expected) column. 

If we maintain states by simply keeping the frame set of each object set, the states are shown in the \text{\em Frame Set States} column in the table.

A new state with object set $\{B\}$ is created at frame 0 since there are no existing states. When frame 1 arrives, according to the above steps: we first append the new frame id 1 to state $\{B\}$; then a new state $\{ABC\}$ with frame set $\{1\}$ is created (step 2.b). With the same procedure, when frame 2 arrives, the new frame id is added to state $\{B\}$. Two new states are created: $\{AB\}$ is created by step 2.a, since $\{ABC\} \cap \{ABDF\} = \{AB\}$; $\{ABDF\}$ is created by step 2.b, since no existing states have the same object set as $\{ABDF\}$. By keep doing so, we obtain the complete state set in the current window for each frame.
Note that the object set $\{A\}$ is not generated and maintained in any window. This is because $\{A\}$ is not an MCOS since $A$ always co-occurs with object $B$ (in frames 1-4). 

No object set is satisfied in frame 0 and 1 since the total number of frames is less than 3. At frame 2, frame set $\{0,1,2\}$ generates one object set $\{B\}$, which is an MCOS. When frame 3 arrives, frame set $\{1, 2, 3\}$ produces another MCOS $\{AB\}$.
Note that in frame 4, the state with object set $\{B\}$ is satisfied (size of the frame set $\ge 3$); however, it is not an MCOS of the current frame set. The only MCOS in the current window is $\{AB\}$, which appeared in all frames. 

\end{ex}

\begin{table}[h]
\begin{small}
	\caption{A Video Segment Example}
	\label{table:example}
	\begin{center}
		\begin{tabular}{ |c|c|c|c| } 
			\hline
			FID & Objects & States & EXP \\ 
			\hline
			0& $\{B\}$ & $(\{B\}, \{0\})$ & $\emptyset$\\ 
			\hline
			1& $\{ABC\}$ & $(\{B\}, \{0, 1\})$; $(\{ABC\}, \{1\})$ & $\emptyset$\\
			\hline
			& & \underline{$(\{B\}, \{0, 1, 2\})$}; $(\{ABC\}, \{1\})$ & $\{B\}$ \\
			2& $\{ABDF\}$ &  $(\{AB\}, \{1, 2\})$; $(\{ABDF\}, \{2\})$ &  \\
			\hline
			& &\underline{$(\{B\}, \{0, 1, 2, 3\})$}; $(\{ABC\}, \{1, 3\})$ & $\{B\}$\\
			3&  $\{ABCF\}$ & \underline{$(\{AB\}, \{1, 2, 3\})$}; $(\{ABDF\}, \{2\})$  & $\{AB\}$\\
			&  & $(\{ABF\}, \{2,3\})$;$(\{ABCF\}, \{3\})$ & \\
			\hline
			&  & \underline{$(\{B\}, \{1, 2, 3, 4\})$}; $(\{ABC\}, \{1, 3\})$ & \\ 
			&  & \underline{$(\{AB\}, \{1, 2, 3, 4\})$} & $\{AB\}$\\
			4& $\{ABD\}$ & $(\{ABDF\}, \{2\})$;$(\{ABF\}, \{2,3\})$  &\\
			&  & $(\{ABD\}, \{2,4\})$;$(\{ABCF\}, \{3\})$ & \\ 
			\hline
		\end{tabular}
	\end{center}
\end{small}
\end{table}
In the above example, the only MCOS of frame set $\{1, 2, 3,\\ 4\}$ in the window of frame $4$ is $\{AB\}$. However, we are keeping two states ($\{AB\}$ and $\{B\}$) with the same frame set. $\{B\}$ is not an MCOS of the frame set, thus the corresponding state is an invalid state, and should not be processed further and/or evaluated. If we maintain it will impose unnecessary overhead and it will be processed against new arriving frames.

\subsubsection{Key Frame Set} \label{sect:kmPS}

Such invalid states (shown in Section \ref{sect:attempt}) should be removed from the state set so we introduce a mechanism to do so. 
Comparing frame sets across states can be costly since we don't know which two states have to be compared without any indexes.
To avoid such computations we maintain information regarding when a state becomes invalid. Once a state is invalid, it is removed from the state set immediately.

Consider the state $(\{AB\}, \{1,2,3,4\})$ from frame 4 in Table \ref{table:example}. Removing frame 3 from the frame set, the resulting state $(\{AB\}, \{1,2,4\})$ is still valid since  $\{AB\}$ is still an MCOS of the new frame set. 
However, if we remove both frames 1 and 3, the new state $(\{AB\},\{2, 4\})$ will become invalid since the MCOS for frame set $\{2,4\}$ is $\{ABD\}$.
Thus, for a valid state $s$, a subset of $\mathbb{F}_{s}$ determines the validity of a state. If we identify these frames, we can assert whether a state is valid or not.
Instead of waiting until every frame expires, a state can be pruned as long as all frames in a key frame set expired. This observation can readily be used in Example \ref{example-1} to prune $\{B\}$.
To formally define these frames we introduce the notion of Key Frame Sets:

\begin{definition}
\textbf{Key Frame Set}. Let $s$ be a state with MCOS $\mathbb{ID}_{s}$ of the frame set $\mathbb{F}_{s}$. For some frame set $\mathbb{KF}_{s}$, where $\mathbb{KF}_{s} \subseteq \mathbb{F}_{s}$. $\mathbb{KF}_{s}$ is a key frame set iff:
\begin{enumerate}[noitemsep,nolistsep]
    \item $\mathbb{ID}_{s}$ is not an MCOS of frame set $\mathbb{F}_s \setminus \mathbb{KF}_{s}$.
    \item $\forall kf \in \mathbb{KF}_{s}$, $\mathbb{ID}_{s}$ is an MCOS of frame set $\mathbb{F}_s \setminus \mathbb{KF}_{s} \cup \{kf\}$.
\end{enumerate}

Each frame $kf \in \mathbb{KF}_{s}$ is a \textbf{Key Frame} in the key frame set.
\end{definition}

That is, a Key Frame Set of state, $s$, is a set of frames with the following properties: if we remove all frames in the key frame set from a state $s$, then $\mathbb{ID}_{s}$ will not be an MCOS; if we add one of the key frames back, $\mathbb{ID}_{s}$ will be an MCOS. In example \ref{example-1} if we remove frames 1 and 3 from state $\{AB\}$, another state $(\{ABD\}, \{2,4\})$ has the same frame set. Since $\{AB\} \subset \{ABD\}$, the object set $\{AB\}$ is no longer an MCOS. Thus, $\{1,3\}$ is one of the key frame sets. For the same reason $\{2,4\}$ and $\{1,4\}$ are also key frame sets. Since frames are expired in temporal order, state $\{AB\}$ will be invalid after frame $3$ expires if no other frames are added to the state.

For a state, $s$, that is generated by some frame, $k$, directly, the frame $k$ is a key frame of $s$ (as in case 2 of Section \ref{sect:attempt}).
Even when the window slides, if the frame, $k$, remains in the window, the state, $s$, will always be valid.
To highlight a key frame in a frame set, we add a mark (*) before the frame id. A frame set with such marked frames will be denoted as \textbf{Marked Frame Set}.

Suppose a new frame $i$ arrives, and a state $ns$ is created by frame $i$ (Case 2 of Section \ref{sect:attempt}).
For each state $s \in \mathbb{S}$, frames will be marked according to the {\bf Frame Marking Rules}, as follows:

\begin{enumerate}[noitemsep, nolistsep]
    \item For the state $ns$ that is created by frame $i$ (Step 2.b in Section \ref{sect:attempt}), frame $i$ will always be marked in $\mathbb{F}_{ns}$.
    \item If $\exists s' \in \mathbb{S}$, $\mathbb{ID}_{s'} \cap \mathbb{ID}_{ns} = \mathbb{ID}_{s}$, $\forall f \neq i \in \mathbb{F}_{s'}$, if $f$ is marked in $\mathbb{F}_{s'}$, then $f$ will also be marked in $\mathbb{F}_{s}$.
    
\end{enumerate}

By applying the above rules, the following theorem can be proved:

\begin{theorem}
The set of marked frames in each state is a Key Frame Set.
\label{th:kPS}
\end{theorem}

For a state $s$, if at least one frame is marked, then it has at least one key frame (according to theorem \ref{th:kPS}). According to the definition of Key Frame Set, state $s$ contains an MCOS. Thus, state $s$ is a valid state if at least one frame in $\mathbb{F}_{s}$ is marked. A state is eliminated from the current collection of states once all marked frames in the
frame set expire. 

We show in the following example that
the marked frame set can reduce the maintenance cost significantly while producing the correct result. 

\begin{table}[h]
\begin{small}
	\caption{Example with Marked Frame Set}
	\label{table:mfs}
	\begin{center}
		\begin{tabular}{ |c|c|c| } 
			\hline
			FID & Objects & States \\
			\hline
			0& $\{B\}$ & $(\{B\}, \{*0\})$ \\
			\hline
			1& $\{ABC\}$ & $(\{B\}, \{*0, 1\})$; $(\{ABC\}, \{*1\})$ \\
			\hline
			& & \underline{$(\{B\}, \{*0, 1, 2\})$}; $(\{ABC\}, \{*1\})$ \\
			2& $\{ABDF\}$ &  $(\{AB\}, \{*1,2\})$; $(\{ABDF\}, \{*2\})$ \\
			\hline
			& & \underline{$(\{B\}, \{*0, 1, 2, 3\})$}; $(\{ABC\}, \{*1, 3\})$ \\
			3&  $\{ABCF\}$ & \underline{$(\{AB\}, \{*1, 2, 3\})$}; $(\{ABDF\}, \{*2\})$ \\
			&  & $(\{ABF\}, \{*2,3\})$;$(\{ABCF\}, \{*3\})$ \\
			\hline
			&  & \sout{$(\{B\}, \{1, 2, 3, 4\})$}; $(\{ABC\}, \{*1, 3\})$ \\ 
			&  & \underline{$(\{AB\}, \{*1, 2, *3, 4\})$} \\
			4& $\{ABD\}$ & $(\{ABDF\}, \{*2\})$; $(\{ABF\}, \{*2,3\})$\\ 
			&  & $(\{ABD\},\{*2,*4\})$; $(\{ABCF\}, \{*3\})$ \\ 
			\hline
		\end{tabular}
	\end{center}
\end{small}
\end{table}

\begin{ex}
    Consider the example in Table \ref{table:example} again. State with object set $\{B\}$ will be generated at frame 0, while 0 will be marked in the frame set. From frame 1 to 3, the new frame id will be appended to the existing state. When frame 4 arrives, frame 0 will be removed from the set along with the mark. Since all the marked frames have been removed from the state, the state is invalid and will be removed. Thus, at frame 4, only state $\{AB\}$ will be selected as the result.
\end{ex}

\subsubsection{The MFS Approach} \label{sect:mfs-app}

Following the Frame Marking Rules (Section \ref{sect:kmPS}), we propose the MFS approach, which utilizes a set of states to maintain all possible object sets along with their 
Marked Frame Set. 
In particular, in MFS, we maintain a set of existing states from the previous window, and when a new frame arrives, we compute the object set intersection between each existing state and the new frame and mark the frame set according to the Frame Marking Rules.

\subsection{Strict State Graph}  \label{sect:ssg}

Although MFS provides the ability to remove unnecessary states early, all states still have to be processed for each arriving frame. In this section, we further explore the relationship between states and provide an alternative approach to provide better pruning power.
 
\subsubsection{Principal States}

We first introduce \textbf{Principal State} to denote a subset of the states in the current window, defined as follows:

\begin{definition}
\textbf{Principal State}. 
For a window size $w$, let $V_{i,w} = <f_{max(0, i-w)}, ..., f_i>$ be the corresponding video frames.
A state $s$ in the current window is a principal state iff $\exists f_j \in V_{i,w}$, such that $\mathbb{ID}_{f_j} = \mathbb{ID}_{s}$.
\end{definition}

A state may become a principal state when a new frame, $k$, arrives: a) if $\exists s \in \mathbb{S}$, such that $s$ shares the same object set as the new frame $k$, namely $\mathbb{ID}_{f_k} = \mathbb{ID}_s$ then $s$ will becomes a principal state. 2) if none of the states in $\mathbb{S}$ has the same object set as the new frame, a new state, $s$, will be created with object set that of frame $k$ and will become a principal state. We will say that a principal state, $s$, is created from frame $k$, and that frame $k$ created principal state $s$.
For the principal state, $s$, the frame that created it, $k$, is always a key frame of state $s$.
A principal state will cease being one, once all the frames that created it expire from the current window. 

For example, at frame 0 in Table \ref{table:mfs}, state $\{B\}$ is the only state and is generated by the arriving frame directly, thus it is a principal state. For the same reason, principal state $\{ABC\}$ is created when frame 1 arrives, while another principal state $\{ABDF\}$ is created by frame 2. At frame 3, there are four principal states in total: $\{B\}$, $\{ABC\}$, $\{ABDF\}$ and $\{ABCF\}$. When frame 4 arrives, frame 0 expires from the current window, thus state $\{B\}$ is no longer a principal state.

Principal states carry an important property. Consider a window on the relation; starting with one frame, the only existing state (that was created from the frame) is a principal state. If we add one more frame to the window, a new state may be created based on the existing principal state and the new principal state. By keep on adding more frames, new states can always be created based on the existing states and the new principal state. Thus, every state in the current window is generated by principal states {\em directly or indirectly}.

\subsubsection{Strict State Graph}

To capture the relationship of how the states are generated, we use \textbf{State Graph} to connect those states, defined as follow:

\begin{definition}
\textbf{State Graph(SG)}. A state graph is a directed graph, which can be represented as $G=(\mathbb{S},\mathbb{E})$, where $\mathbb{S}$ is the set of all states and $\mathbb{E}$ is the set of all edges. For each edge $e_{ij} \in \mathbb{E}$, $e_{ij}=(s_i, s_j)$, where $s_i,s_j \in \mathbb{S}$, it means state $s_j$ is generated from state $s_i$.
\end{definition}

For any state graph $G$, the following holds:

\begin{prop}
	$\forall s_i,s_j \in \mathbb{S}$, if $\exists e_{ij}=(s_i,s_j) \in \mathbb{E}$, then $\mathbb{ID}_{s_j} \subset \mathbb{ID}_{s_i}$.
\end{prop}

\begin{figure}[t!]
    \centering
    \begin{subfigure}[t]{0.25\textwidth}
        \centering
        \includegraphics[height=50pt]{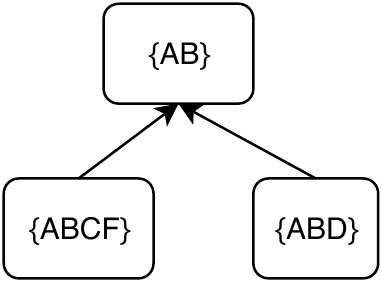}
        \caption{Before}
        \label{fig:edge-ex:before}
    \end{subfigure}%
    ~ 
    \begin{subfigure}[t]{0.25\textwidth}
        \centering
        \includegraphics[height=50pt]{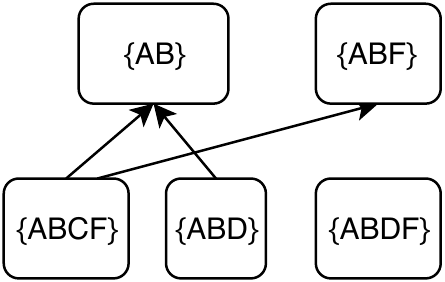}
        \caption{After (1)}
        \label{fig:edge-ex:after1}
    \end{subfigure}
    \begin{subfigure}[t]{0.25\textwidth}
        \centering
        \includegraphics[height=50pt]{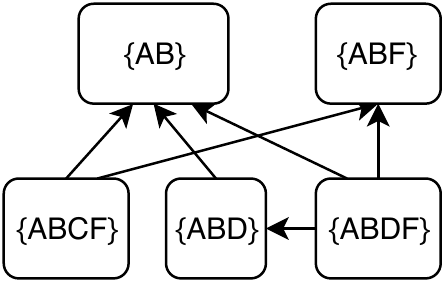}
        \caption{After (2)}
        \label{fig:edge-ex:after2}
    \end{subfigure}%
    ~ 
    \begin{subfigure}[t]{0.25\textwidth}
        \centering
        \includegraphics[height=50pt]{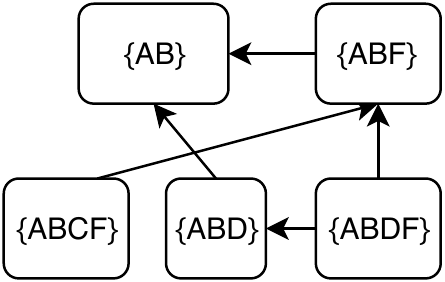}
        \caption{Expected}
        \label{fig:edge-ex:expected}
    \end{subfigure}
    \caption{Adding New Edges to SSG}
    \label{fig:edge-ex}
\end{figure}

Suppose we have two principal states, \{ABD\} and \{ABCF\}, the State Graph generated is shown as Figure \ref{fig:edge-ex:before}. According to Property \ref{prop:prune}, if a new frame with objects \{EG\} arrives, the state with object set \{AB\} can be omitted due to empty object intersection between the new frame and any existing principal states.

Assume a new frame with \{ABDF\} arrives. A new state \{ABF\} can be generated by computing object intersection between \{ABCF\} and the new frame (Figure \ref{fig:edge-ex:after1}). However, if we simply connect states as the way they are generated, we will obtain a graph with unnecessary edges (Figure \ref{fig:edge-ex:after2}). To remove these unnecessary edges, we propose to generate graphs maintaining the following property:

\begin{prop}
	For $\forall s \in \mathbb{S}$, suppose $\mathbb{E}_{s} \subset \mathbb{E}$ is the edge set containing all of the edges starting from $s$. $\forall s_i, s_j \in \mathbb{S}$, where $(s, s_i), (s, s_j) \in \mathbb{E}_s$, we always have $\mathbb{ID}_{s_i} \cap \mathbb{ID}_{s_j} \ne \mathbb{ID}_{s_i}$ and $\mathbb{ID}_{s_i} \cap \mathbb{ID}_{s_j} \ne \mathbb{ID}_{s_j}$.
	\label{prop:ssg}
\end{prop}

We use \textbf{Strict State Graph (SSG)} to denote a State Graph that satisfies the above property.
The graph that conforms to property \ref{prop:ssg} is shown in Figure \ref{fig:edge-ex:expected}.
To leverage the structure of SSG, two crucial issues have to be addressed: 
how to maintain the edges of the graph, and update the marked frame set for each state.
We first focus on edge maintenance in Section \ref{sect:edge-maintain}, and address the state marking issue in Section \ref{sect:state-marking}.
Then we provide the main algorithm in Section \ref{sect:state-algorithm}.
The algorithm analysis will be presented in Section \ref{sect:state-analysis}.

\subsubsection{Edge Maintenance} \label{sect:edge-maintain}
Since other states are generated by principal states directly or indirectly (Section \ref{sect:mPS}), we always maintain the set of principal states in the current window, denoted as $\mathbb{PS}$.

When a new frame $i$ arrives, let $ns$ be the principal state created by $i$. We start traversing the graph from one of the states in $\mathbb{PS}$ (note that $ns$ will be added to $\mathbb{PS}$ after all members of $\mathbb{PS}$ are visited). Assume during traversal of the graph we are visiting state $s$ ($s$ could be a principal state or another existing states on the graph). Utilizing the Strict State Graph, state $s$ will be handled during traversal by the {\bf Graph Maintenance Procedure} as follows:

\begin{enumerate}[noitemsep, nolistsep]
    \item We first compute the object set intersection between $ns$ and state $s$. Let $inter$ denote the intersection result, $inter = \mathbb{ID}_{s} \cap \mathbb{ID}_{ns}$.
    \item If $inter = \emptyset$. That means for any state adjacent to $s$, the object intersection between that state and $ns$ is empty. Thus, all of the adjacent states no longer need to be processed.
    \item If $inter = \mathbb{ID}_{s}$. We add frame $i$ to the frame set of $s$.
    \item Otherwise, we need to further compute object intersections between $ns$ and the adjacent states of $s$. Let $\mathbb{S}_{s} \in \mathbb{S}$ denote the set of adjacent states of $s$:
        \begin{enumerate}[noitemsep,nolistsep]
            \item If $\exists ss \in \mathbb{S}_{s}$, $\mathbb{ID}_{ss} \cap \mathbb{ID}_{ns} = inter$, then the state with object set $inter$ is already on the graph and no new states or edges need to be created.
            \item If $\forall ss \in \mathbb{S}_{s}$, $\mathbb{ID}_{ss} \cap \mathbb{ID}_{ns} \ne inter$, a new state with object set $inter$ is created and added to the graph. We say that the new state with the object set, $inter$, can be \textbf{generated} from $s$.
        \end{enumerate}
    \item After all states are visited according to above 4 steps, the new principal state, $ns$, will be added to the graph.
\end{enumerate}

To satisfy Property \ref{prop:ssg}, we maintain edges in an SSG in two steps:

\begin{enumerate}[noitemsep, nolistsep]
    \item \textbf{Modifying Existing Edges}. 
    When a new state is generated by intersecting objects between the new principal state and existing states (step 4.b of the Graph maintenance Procedure),
    we need to check whether existing edges should be deleted, and add new edges to the graph to satisfy Property \ref{prop:ssg}.
    \item \textbf{Connecting the New Principal State}. 
    When a new principal state is added to the graph, existing states should be carefully chosen to build new edges (step 5 of the Graph maintenance Procedure) while guaranteeing Property \ref{prop:ssg}.  
\end{enumerate}

\subsubsection{Modifying Existing Edges} \label{sect:state-edges}

Suppose a new principal state $ns$ is created by a new frame, and state $s_i$ is being visited ($s_i$ could either be an existing principal state or other states on the graph). $s'$ is another state that can be generated from $s_i$.

Assume after adding the edge, $(s_i, s')$, Property \ref{prop:ssg} is violated. That means, $\exists s_j \in \mathbb{S}_{s_i}$, $\mathbb{ID}_{s'} \cap \mathbb{ID}_{s_j} = \mathbb{ID}_{s_j}$.
In order to satisfy the property, edges have to be removed from the graph. 
To determine which edge should be removed, we need to check the object set intersections between $s'$ and the object sets of all adjacent states of $s_i$. 
In the above case, the intersection between object sets of state $s'$ and $s_j$ will be the same as the object set in $s'$, and thus edge $(s_i, s_j)$ should be removed from the graph, while a new edge $(s', s_j)$ should be added.

Consider Figure \ref{fig:edge-ex:after1} as an example, after the edge $(\{ABCF\},\\ \{ABF\})$ is added, we check the adjacent state of $\{ABCF\}$. $\{AB\} \cap \{ABF\} = \{AB\}$, which violates Property \ref{prop:ssg}. Thus, the existing edge $(\{ABCF\}, \{AB\})$ should be removed. Meanwhile, a new edge $(\{ABF\}, \{AB\})$ will be added.

\subsubsection{Connecting the New Principal State} \label{sect:state-candidate}

To connect a new principal state to the graph,  we first introduce the following theorem:

\begin{theorem}
\label{th4}
Let $\mathbb{PS}$ be the set of principal states and $ns$ be a new principal state. In an SSG, a state $s \in \mathbb{S}$ could be adjacent to $ns$ only if $\exists ps \in \mathbb{PS}$, such that $\mathbb{ID}_{ns} \cap \mathbb{ID}_{ps} = \mathbb{ID}_{s}$.
\end{theorem}

Thus, as per Theorem \ref{th4} one can obtain at most one state to connect $ns$ starting from each state in $\mathbb{PS}$ during graph traversal. We use $C$ to denote a set of states that could be adjacent to the new frame state $ns$. The set $C$ is constructed as follows:
Starting from each principal state, $ps$:
\begin{enumerate}[noitemsep, nolistsep]
    \item If $\mathbb{ID}_{ps} \cap \mathbb{ID}_{ns} = \emptyset$, then no state will be added to $C$.
    \item Otherwise, let $s$ be the state with object set $\mathbb{ID}_{ps} \cap \mathbb{ID}_{ns}$, then $s$ could be an existing states in the graph (Graph Maintenance Procedure step 4.a), or a new state (Graph Maintenance Procedure step 4.b). In either case, state $s$ will be added to $C$.
\end{enumerate}

For a state, $s \in C$, let $CS_{s}$ be the set containing all states that can be obtained by performing a DFS (Depth-First Search) from $s$ on the graph. 
States in $C$ will be selected according to the {\textbf{ State Selection Procedure}} as follows:

\begin{enumerate}[noitemsep, nolistsep]
    \item We sort the states in $C$ according to $|\mathbb{ID}_{s}|$ (object set) in descending order, where $s \in C$. Meanwhile, an empty set, $RS$, is initialized, to materialize all reachable states from $ns$, and is dynamically updated.
    \item State, $s$, is processed and selected as follows:
    \begin{itemize}[noitemsep, nolistsep]
        \item If  $s \notin RS$, then state $s$ will be selected. All states from $CS_{s}$ will be added to $RS$.
        \item If $s \in RS$, $s$ will not be selected.
    \end{itemize}
\end{enumerate}

After selecting the states, we can now add an edge between each selected state and the new principal state. The following theorem can be proved.

\begin{theorem}
The State Selection Procedure will not violate Property \ref{prop:ssg} while all newly generated states can be reached from the new principal state.
\end{theorem}

Now consider the example in Figure \ref{fig:edge-ex:before} again. 
When a new frame $\{ABDF\}$ arrives: From state $\{ABCF\}$, a candidate, $c_1$, with state $\{ABF\}$ will be created, along with the reachable state set $\{ \{ABF\}, \{AB\}\}$; From state $\{ABD\}$, a candidate, $c_2$, with state $\{ABD\}$ will be created, along with the reachable state set $\{ \{ABD\}, \{AB\} \}$. Since they have the same object set size, we can start from either candidate. Assume we select state from $c_1$ first, after that the reachable state set for state $\{ABDF\}$ is $\{ \{ABF\}, \{AB\}\}$. $\{ABD\}$ is not in the set, so we still need to select the state from $c_2$. After that, a new edges will be created between $\{ABD\}$, $\{ABF\}$ and $\{ABDF\}$. The result is the same as Figure \ref{fig:edge-ex:expected}, and Property \ref{prop:ssg} is maintained.

\subsubsection{Marking the Frame Sets of States} \label{sect:state-marking}

During edge maintenance, the intersections between the new arriving frame and existing states are computed on SSG. 
To minimize the overall maintenance cost, we introduce \textbf{State Marking Procedure} to update the Marked Frame Set for each state at the same time.

For each new arriving frame state $ns$ at frame $i$, we then visit the SSG starting from principal states in their arrival order. the marked frame set for each state is updated as follows:

\begin{enumerate}[noitemsep, nolistsep]
    \item Frame $i$ in state $ns$ will always be marked. 
    
    \item For each state $s$:
    \begin{itemize}[noitemsep,nolistsep]
        \item If a new state $s'$ is generated from $s$. We copy the marked frame set from $s$ to $s'$ ($\mathbb{F}_{s'} \gets \mathbb{F}_{s}$). If $s' \neq ns$, we also append the new frame $i$ to its marked frame set.
        \item If $\mathbb{ID}_{s} \cap \mathbb{ID}_{ns} = \mathbb{ID}_{s}$, we simply append the new frame $i$ to the marked frame set of $s$.
    \end{itemize}
    
    \item After state $s$ is processed, it is marked as visited. The adjacent states of $s$ will be further processed by applying the above rule recursively until there are no adjacent states or all adjacent states are already visited.
    
    \item If some principal state, $ps$, (created from frame $i'$) is already visited, and $\mathbb{ID}_{ps} \cap \mathbb{ID}_{ns} \neq \emptyset$. This means from $ps$ and $ns$, a new state $s'$ is generated. Since $ps$ is already visited, the new state $s'$ must already created on the SSG. In this case, we mark the frame id $i'$ in the marked frame set of $s'$.
    
\end{enumerate}

The correctness of the above procedure is guaranteed by the following theorem.

\begin{theorem}
    Using the State Marking Procedure, a state is valid in SSG iff at least one frame in its frame set is marked. 
\end{theorem}

\subsubsection{The Algorithm} \label{sect:state-algorithm}

We now present the algorithm that builds SSGs. The algorithm consists of two main parts: State Traversal (ST), and Connecting the New Principal State (CNPS).

\begin{enumerate}[noitemsep, nolistsep]
    \item We maintain principal states in their order of arrival.
    The State Traversal (ST) algorithm implements the Graph Maintenance Procedure and State Marking Procedure.
    Meanwhile, a set of states, $C$ (described in Section \ref{sect:state-candidate}), will also be generated by the ST algorithm.
    \item Connecting the New Principal State (CNPS) algorithm, implements the state selection procedure to add new edges between existing states and the new principal state, as discussed in Section \ref{sect:state-candidate}.
\end{enumerate}

The ST algorithm is shown as Algorithm \ref{al:st}, which receives 6 parameters: the arriving frame id, $i'$, an existing state to visit, $s$, the parent state $ps$ of state $s$, the new principal state, $ns$,
the previous object intersection $pInter$ (generated by $ps$ and $ns$),
and the current candidate $c$. 

Starting from principal states, $pInter$ is initialized as an empty set, while its parent state $ps$ and the current candidate $c$ is set to $null$.
To avoid visiting the same state multiple times, we include a flag to each state (shown as Lines 1-2). The flag will be set to the current frame id once the state is visited.

For each state, we first remove expired frames from its frame set and prune the state if necessary (function {\em pruneState} Line 3). Then, the intersection between the object sets of states $s$ and $ns$ will be computed (Line 4). 
Lines 5-31 implement the Graph Maintenance Procedure (Steps 1-4), meanwhile, the frame set of each state is updated according to the State Marking Procedure (Steps 1-3).

\begin{algorithm}[h]
\begin{small}
	\SetAlgoLined
	\KwIn{current frame id $i'$; 
		visiting state $s$; parent state $ps$; new principal state $ns$; 
		the previous intersection result $pInter$ (between $ns$ and $ps$); 
		the current candidate $c$ ;
		}
	\lIf{$flag(s) = i'$}{
		return
	}
	flag($s$) $\gets$ $i'$; \textit{// visit each state only once\\}
	pruneState($s$)\textit{// prune the current state}\;
	$inter \gets \mathbb{ID}_{s} \cap \mathbb{ID}_{ns}$ \textit{// compute the intersection} \;
		\uIf{$pInter \ne \emptyset$ \textbf{and } $inter = \emptyset$ }{
			\If{$|pInter| \ne |\mathbb{ID}_{ps}| $ and $|pInter| \ne |\mathbb{ID}_{ns}|$}{
			
				createState($pInter$, $ps$, $ns$, $c$)\;
			}
		} \ElseIf{$inter\ne \emptyset$ } {
			\textit{// check if we need to add states for $pInter$}\\
			\If{$pInter \ne \emptyset$  \textbf{and } $|pInter|>|inter|$}{
				\If{$|pInter| \ne |\mathbb{ID}_{ps}| $  \textbf{and } $|pInter| \ne |\mathbb{ID}_{ns}|$}{
					$s' \gets$ createState($pInter$, $ps$, $ns$, $c$)\;
					\lIf{$|inter| = |\mathbb{ID}_{s}|$} {
					$\mathbb{E} \gets \mathbb{E} \setminus \{(ps, s)\} \cup \{(ns, s)\}$
					}
				}
			}
			\textit{// deal with the current intersection}\\
			\uIf{$|inter| = |\mathbb{ID}_{s}|$}{
			    \lIf{$|pInter| = |inter|$}{
			        $\mathbb{F}_{s} \gets  merge(\mathbb{F}_{s}, \mathbb{F}_{ps})$
			    }
				$\mathbb{F}_{s} \gets \mathbb{F}_{s} \cup \{i'\}$\;
				visitNext($s$, $inter$, $c$)\;
			}
			\uElseIf{$|inter| = |\mathbb{ID}_{ns}|$}{
				$\mathbb{F}_{ns} \gets merge(\mathbb{F}_{s}, \mathbb{F}_{ns})$; 
				$\mathbb{E} \gets \mathbb{E} \cup \{(s, ns)\} $\;
				visitNext($s$, $inter$, $c$)\;
			}
			\Else{
				hasNext $\gets $ visitNext($s$, $inter$, $c$)\;
				\If{hasNext = false}{
					createState($inter$, $s$, $ns$, $c$)\;
				}
			}
	}
	
	\caption{State Traversal (ST)}
	\label{al:st}
\end{small}
\end{algorithm}
In the algorithm, we utilize two other functions, {\em visitNext} and {\em createState}: {\em createState} will create a state (or retrieve an existing one) with object set equal to $inter$ and modify edges according to Section \ref{sect:state-edges}; while {\em visitNext} visits the adjacent vertices of the given state $s$, calling the ST algorithm recursively.

After each principal state is visited by the ST Algorithm, the CNPS Algorithm is responsible for connecting the new principal state to the graph, which is described as the State Selection Procedure in Section \ref{sect:state-candidate}, shown as Algorithm \ref{al:conn}. In the algorithm, $c$ stores a state ($c.s$) and the reachable state set ($c.CS$) from the state, both of which are obtained in the ST algorithm.

A state $s$ is satisfied and valid if $|\mathbb{F}_{s}| \ge d$ and at least one frame in $\mathbb{F}_{s}$ is marked.
Let $\mathbb{SR}_{i}$ be the \textbf{Result State Set}, which contains all satisfied and valid states in the window of frame $i$. The result state set, $\mathbb{SR}_{i}$, is sent to the Query Evaluation module. Meanwhile, we also keep a copy of set $\mathbb{SR}_{i}$.
When a new frame $i'$ arrives, let $G'$ be the new state graph. 
The new result, $\mathbb{SR}_{i'} = \mathbb{SR'}_{i} \cup \mathbb{SR}_{G'}$, where $\mathbb{SR'}_{i}$ is the set of satisfied and valid states from $\mathbb{SR}_{i}$ in the current window,
while $\mathbb{SR}_{G'}$ is the set of satisfied and valid states obtained on the graph.
This is because, for a new principal state, $ns$, we visit state $s \in \mathbb{S}$ on the graph only if $\mathbb{ID}_{s} \cap \mathbb{ID}_{ns} \ne \emptyset$. States $s'$ in the previous result set $\mathbb{SR}_{i}$ may still be valid and satisfied. If the intersection of object sets between state $s'$ and $ns$ is empty, we may leave it out during graph traversal. Keeping only the previous result set is sufficient since for any existing unsatisfied state, $s$, it will become satisfied only if $s$ is visited on the graph in the current window. 
Thus, if $s$ is a satisfied state at frame $i'$, then it is either: a) visited in the current window at frame $i'$, or b) was also satisfied in the window of the previous frame $i$.

\begin{algorithm}
\begin{small}
	\SetAlgoLined
	\KwIn{the state set $C$; the new principal state $ns$; }
	$CL \gets $ sort the $C$ in descending order of $|\mathbb{ID}_{c.s}|$\;
	$RS \gets \emptyset$;\textit{// stores states that can be reached from $ns$}\\
	\For{$c : CL$}{
		\If{ $c.s \notin RS$ }{
		    $\mathbb{E} \gets \mathbb{E} \cup \{(ns, c.s)\}$;
		    $RS \gets RS \cup c.CS$\;
		}
	}
	\caption{CNPS}
	\label{al:conn}
\end{small}
\end{algorithm}
\subsubsection{Analysis} \label{sect:state-analysis}
Assume the number of frames that share the same object set is $\lambda$ on average.
Given a window size $w$, $x = \frac{w}{\lambda}$ number of unique principal states will be generated. In the full version of the paper we prove:

\begin{theorem}
    In SSG, there will be at most $2^{x}$ number of states and $x2^{x}$ number of edges. Thus, the complexity of algorithm ST is $O(x2^{x})$, where $x=\frac{w}{\lambda}$.
\end{theorem}

\section{Query Evaluation} \label{sect:eval}

\subsection{The CNF Algorithm}
For our query evaluation module any CNF evaluation algorithm can be employed. 
We choose to utilize the evaluation algorithm proposed in \cite{whang2009indexing}, which will be referred to as {\em CNFEval}.

The CNFEval algorithm evaluates a set of CNF queries that include $\in$, $\notin$ (set membership) predicates. The algorithm first builds an inverted index for the starting set of queries. The inverted index is dynamically maintained as queries are inserted and/or deleted. 

For example, query $q_1 = age \in \{2, 3\} \land (state \in \{CA\} \lor gender \in \{F\})$, specifies a CNF on name-value pairs. Specifically, this query can be interpreted as follows: for name "age", the value is $2$ or $3$; for name "state", the value is "CA"; and for name "gender", the value is "F". 

An inverted index for $q_1$ containing posting lists is generated as shown in Table \ref{table:cnf-pl}. A posting list is a list of triples that are generated according to query conditions. A triplet is represented as $(qid, p, disjId)$, where $qid$ is the query id, $p$ is the predicate ($\in$ or $\notin$), and $disjId$ is the disjunction id (starting from 0) of the condition. For example, condition $gender \in \{F\}$ is in the second disjunction, and thus a triplet (1, $\in$, 1) is generated.

For a given set of name-value pairs, the algorithm retrieves posting lists from the inverted index first. By ordering and scanning all retrieved posting lists, queries can be answered correctly. For example, given an input $\{(age, 3), (gender, F)\}$, two posting lists $(1, \in, 0)$ and $(1, \in, 1)$ are retrieved. From the two triplets in the lists, we determine that both disjunctions $0$ and $1$ are satisfied. 
Thus, query $q_1$ is evaluated as TRUE. The details of the entire algorithm are available elsewhere \cite{whang2009indexing}.
\begin{table}[]
    \centering
\begin{small}
        \caption{Inverted Index}
        \begin{tabular}{|l|l|}
        \hline
        Key   & Posting List \\ \hline
        (age, 2) & (1, $\in$, 0)     \\ \hline
        (age, 3) & (1, $\in$, 0)     \\ \hline
        (state, CA) & (1, $\in$, 1)     \\ \hline
        (gender, F) & (1, $\in$, 1)     \\ \hline
        \end{tabular}
        \label{table:cnf-pl}
\end{small}
\end{table}

\subsection{CNF Evaluation with Inequality Predicates} \label{sect:eval:eval}

The CNFEval algorithm only handles queries containing set predicates. As a result inequality conditions are not supported naturally. Thus, we propose to build three separate inverted indexes for the query conditions of the form $label$ $\theta$ $n$ ($label$ is used as the key, while $n$ is used as the value), with an extra column, value, associated to each of the three cases, $\leq, \geq, =$. 
Consider a query $q_2=(car \ge 2 \lor person \le 3) \land (car \ge 3 \lor person \ge 2) \land (car \le 5) $ in our framework. Tables \ref{table:pl-ge} and \ref{table:pl-le} depict the inverted indexes for this query. Since this query contains both $\geq$ and $\le$ conditions, two inverted indexes are built. 
Each key in the inverted list is associated with an ordered list; ordering takes place by value in descending (if $\theta$ is $\le$) or ascending (if $\theta$ is $\ge$) order. 
Given an input with a set of name-value pairs, posting lists are retrieved as follows. For each name-value pair, $(k, v)$:

\begin{enumerate}[noitemsep, nolistsep]
    \item Retrieve the ordered list from both tables where key=$k$.
    \item Retrieve posting lists in order where value $\le v$ (if $\theta$ is $\ge$) or value $\ge v$ (if $\theta$ is $\le$).
\end{enumerate}

\begin{table}[!htb]
\begin{small}
    \begin{minipage}[t]{.5\linewidth}
      \caption{$\ge$ Index}
      \centering
        \begin{tabular}{|l|l|l|}
        \hline
        Key                & Value & Posting \\
        &  & List \\\hline
        \multirow{2}{*}{Car} & 2     & (2, 0)     \\ \cline{2-3} 
                           & 3     & (2, 1)     \\ \hline
        Person & 2     & (2, 1)     \\ \hline
        \end{tabular}
        \label{table:pl-ge}
    \end{minipage}%
    \begin{minipage}[t]{.5\linewidth}
      \centering
        \caption{$\le$ Index}
        \begin{tabular}{|l|l|l|}
        \hline
        Key                & Value & Posting \\
        &  & List \\\hline
        Car& 5     & (2, 2)     \\ \hline
        Person   & 3     & (2, 0)     \\ \hline
        \end{tabular}
        \label{table:pl-le}
    \end{minipage} 
\end{small}
\end{table}

We refer to this version of the algorithm, which encompasses our enhancements, as {\em CNFEvalE}. It is utilized to evaluate queries on the Result State Set (Section \ref{sect:state-algorithm}) produced by the MCOS Generation module. The procedure to evaluate queries on video feeds is described as follows:  

\begin{enumerate}[noitemsep]
    \item For the given queries, we first build inverted lists.
    \item For each state, $s$, in the Result State Set (Section \ref{sect:state-algorithm}) generated by the MCOS Generation module:
    \begin{enumerate}[noitemsep,nolistsep]
        \item A set of aggregate values, $\mathbb{A}_{s}$, is computed based on the MCOS of state $s$. Each aggregate value is represented as $(l, v)$, where $l \in \mathbb{L}$ is the class label, while $v$ is the number of objects of type $l$.
        \item The aggregate value set, $\mathbb{A}_{s}$, is utilized as input to CNFEvalE to produce the results.
        \item If some query, $q$, is evaluated as true, the frame set of the state, $\mathbb{F}_{s}$, is produced as the result.
    \end{enumerate}
\end{enumerate}

\subsection{Pruning States by Evaluation Results} \label{sect:eval:prune}

Query evaluation proceeds on the MCOS of a state in the Result State Set (Section \ref{sect:ssg}). We make the following observation: if queries are evaluated as FALSE on the MCOS generated by some state $s$, then they will always be evaluated as FALSE on every possible MCOS generated by state $s'$, where $s'$ is generated from $s$.

This observation is captured by the following proposition:

\begin{proposition}
    Let $s$ be some state in an SSG, and $\mathbb{ID}_{s}$ be the MCOS generated by $s$.
    For each condition $c$ in query $q$, if $c$ is evaluated as FALSE on $\mathbb{ID}_{s}$, then for any state $s'$ with $\mathbb{ID}_{s'}$, where $\mathbb{ID}_{s'} \subset \mathbb{ID}_{s}$, the condition $c$ will always be evaluated as FALSE on $\mathbb{ID}_{s'}$.
    \label{prop:prune}
\end{proposition}
If Proposition \ref{prop:prune} holds, such states can be safely removed from the graph while guaranteeing correctness.
Proposition \ref{prop:prune} can be satisfied only if $\theta$ is $\ge$. To adopt such a pruning strategy, for a given set of queries, we test whether they contain $\ge$ predicates only. If so, when a new state is generated (Graph Maintenance Procedure step 4.b), we evaluate the MCOS of that state using the Query Evaluation module. In this case, we do not wait for the entire Result State Set to be generated. If all queries are evaluated as FALSE, the state will be marked as "terminated". Terminated states will no longer be processed, and thus we can reduce the number of states maintained in the MCOS Generation module. This provides additional pruning flexibility to the MCOS generation module in certain cases.

\section{Experiments} \label{sect:exp}

In this section, we present a thorough experimental evaluation of the two proposed approaches (Marked Frame Set and Strict State Graph) on both synthetic and real datasets and study the tradeoffs between them.
We first evaluate the MCOS generation methods proposed in this paper in Section \ref{sect:exp:state}.
Then, based on MCOS generation , we show the performance of our Query Evaluation in Section \ref{sect:exp:eval}.

\subsection{Settings} \label{sect:exp:setting}
\textbf{Environment}. 
Our experiments are conducted on a machine with Intel i5-9600K CPU, one GeForce GTX 1070 GPU and 16 GB RAM.
All algorithms are implemented in Java except the object detection and tracking algorithms for which we utilize standard implementations.

\textbf{Object Detection and Tracking}.
We adopt Faster R-CNN \cite{DBLP:journals/pami/RenHG017} as the object detection algorithm to identify all objects in videos and Deep SORT \cite{Wojke2017simple} as the object tracking algorithm. 
For both algorithms, we utilize standard implementation from open-source projects \footnote{https://github.com/open-mmlab/mmdetection}\footnote{https://github.com/nwojke/deep\_sort}.

\textbf{Synthetic Datasets}. We use synthetic videos generated by the VisualRoad benchmark \cite{haynes2019visual}, which is proposed to evaluate the performance of video database management systems (VDBMSs). We select two representative videos provided on the project homepage\footnote{https://db.cs.washington.edu/projects/visualroad/}: rain with light traffic (V1), and postpluvial with heavy traffic (V2). 
We also generate videos with different configurations to study the relationship between the performance and the number of objects per frame.

\begin{table}[]
    \centering
    \begin{small}
    \caption{Dataset Statistics}
    \begin{tabular}{|l|c|c|c|c|c|c|}
        \hline
         \textbf{Dataset} & V1 & V2 & D1 & D2 & M1 & M2\\
         \hline
         \textbf{Frames} & 1800 & 1700 & 1150 & 1145 & 1194 & 750 \\
         \hline
         \textbf{Objects} & 173 & 127 & 179 & 158 & 342 & 186 \\
         \hline
         \textbf{Obj/F} & 7.37 & 5.94 & 7.56 & 8.99 & 6.75 & 11.59 \\
         \hline
         \textbf{Occ/Obj} & 3.6 & 6.33 & 5.20 & 7.23 & 3.37 & 3.48 \\
         \hline
         \textbf{F/Obj} & 76.71 & 79.84 & 48.61 & 65.18 & 23.67 & 46.96 \\
         \hline
    \end{tabular}
    \label{tab:dataset-statistics}
    \end{small}
\end{table}

\textbf{Real Datasets}. We also evaluate our approach based on four videos from Detrac \cite{lyu2017ua} and MOT16 datasets \cite{milan2016mot16}. 
From Detrac, we use MVI\_40171 (D1) and MVI\_40751 (D2). From MOT16, we use MOT16-06 (M1) and MOT16-13 (M2). 

\textbf{Dataset Statistics}.  In the rest of this section, we use V1 and V2 to denote two synthetic videos, while using D1, D2, and M1, M2 to denote four real videos. D1 and D2 are captured by static cameras, {\em while M1 and M2 are captured by moving cameras}. In our experiments, we only identify objects belonging to one of the following classes: person, car, truck, and bus. After applying object detection and tracking algorithms, the statistical information of these videos can be summarized as Table \ref{tab:dataset-statistics}.
In the table, Frames and Objects represent the total number of frames and unique objects, respectively; Obj/F represents the average number of objects per frame; Occ/Obj represents the average occlusion times per object; while F/Obj denotes the number of frames that each object appears on average.

\subsection{MCOS Generation} \label{sect:exp:state}

We implement a baseline approach that simply stores the set of frames that each object set appears. We first collect all the object sets that satisfy the given duration threshold, and then check whether they share the same frame set. If that is the case, according to the definition of MCOS, we only keep the object set with the maximum size. The refer to this approach as {\em NAIVE}. 
We also implemented the MFS approach as described in Section \ref{sect:mPS}, as well as SSG as presented in Section \ref{sect:state-algorithm}.
All three methods are memory-based, where a hash table is used to map the object sets with frame sets (NAIVE), or marked frame sets (MFS), or internal states (SSG).
To study the performance of the three methods, we design experiments to measure only the MCOS generation time.

We vary the window size, $w$ (default value is set to $300$ frames), and the duration parameter, $d$, (default value is set to $240$ frames).
With 30 fps, the default setting aims to identify objects that appear at least 8 seconds in the window of 10 seconds.
The number of occlusions per object is also varied during experiments. By default we do not vary occlusion; each video has a number of occlusions per object occurring naturally in each data-set (Table \ref{tab:dataset-statistics}).

We then vary the number of total frames evaluated on datasets, shown as Figure \ref{fig:state-frames}. 
The performance of all methods demonstrates a similar general trend: when the number of frames increases, the total time also increases. 
In all three methods, MCOSs are generated according to the object sets, thus the number of distinct object sets can have a profound impact on the overall performance.
Thus, even under the same window size and duration threshold, their performance on different videos can vary significantly. 
For example, in Figure \ref{fig:state-frame-d1}, processing the first 400 frames is very expensive, the total cost does not increase significantly as the number of frames increases. 
In contrast, processing the {\em last} 145 frames is much more costly on video D2 (Figure \ref{fig:state-frame-d2}).
This is because these frames include more objects per frame than the average objects per frame for the clip. As a result, these frames contribute significantly to the number of states created overall and subsequently to the overall cost. 
MFS performs slightly better than SSG on all videos generated by the visualroad benchmark. These videos have smaller number of objects per frame, while the duration for each object in the visible screen (F/Obj in Table \ref{tab:dataset-statistics}) is relatively longer than other videos.
In such scenarios,
the number of unique states is relatively small, while most states can be generated directly from principal states. 
The pruning power of SSG is not significant compared with MFS since the graph structure does not provide significant advantage. Instead, extra cost is paid to maintain the graph in the case of SSG.
In contrast, on other datasets (D1 to M2),the performance of SSG is better than MFS due to the pruning effectiveness attributed to the graph structure.
For example, in Figure \ref{fig:state-frame-d2}, SSG is around 10\% faster than MFS and 30\% faster in Figure \ref{fig:state-frame-m2}.
In these datasets, the number of objects per frame is much larger (Obj/F in Table \ref{tab:dataset-statistics}) or the duration of objects is the visible screen is lower, thus more states are maintained.

\begin{figure}[h!]
	\centering
	\begin{subfigure}[b]{0.21\textwidth}
		\centering
		\includegraphics[width=\textwidth]{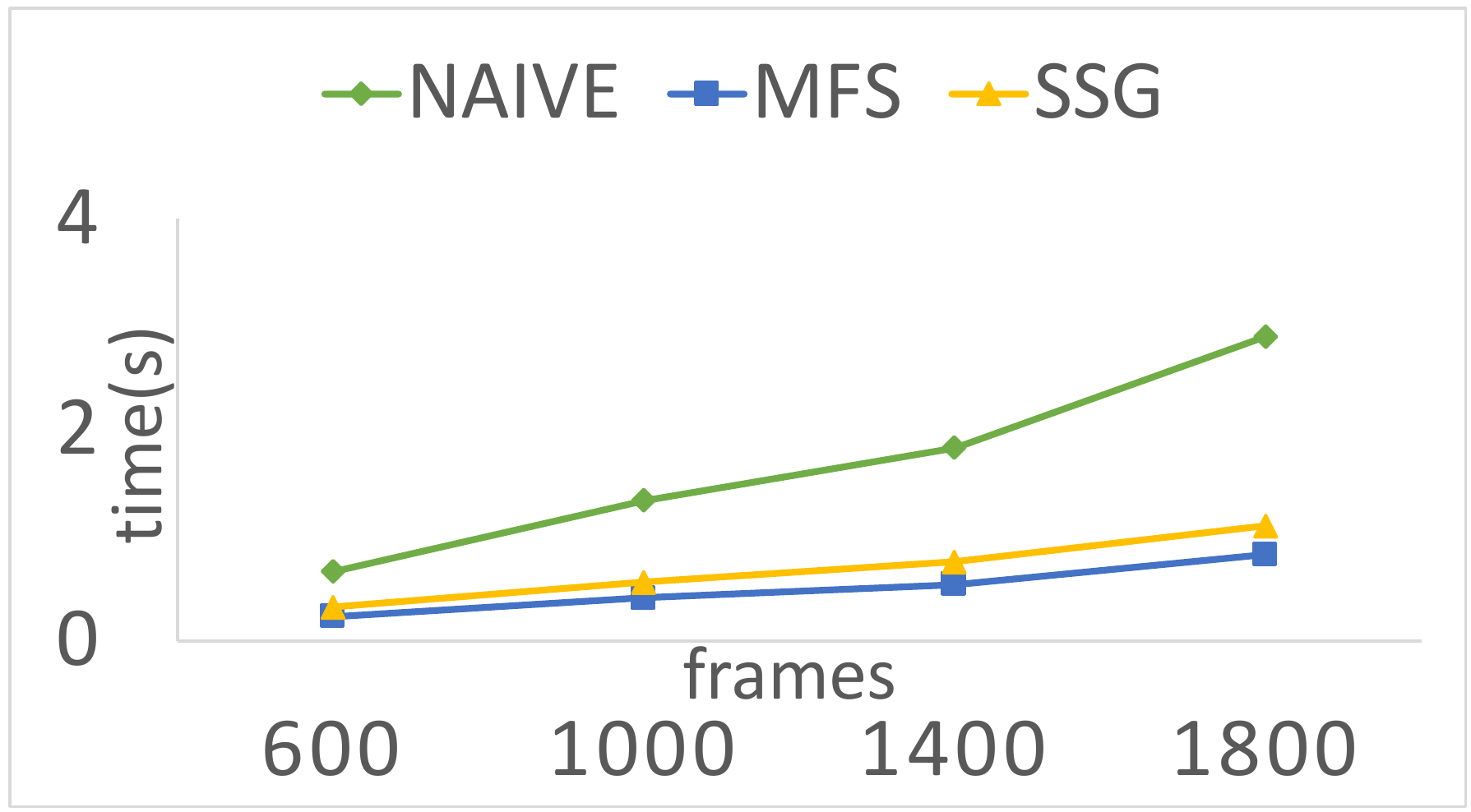}
		\caption{V1}
		\label{fig:state-frame-v1}
	\end{subfigure}
	\hfill
	\begin{subfigure}[b]{0.21\textwidth}
		\centering
		\includegraphics[width=\textwidth]{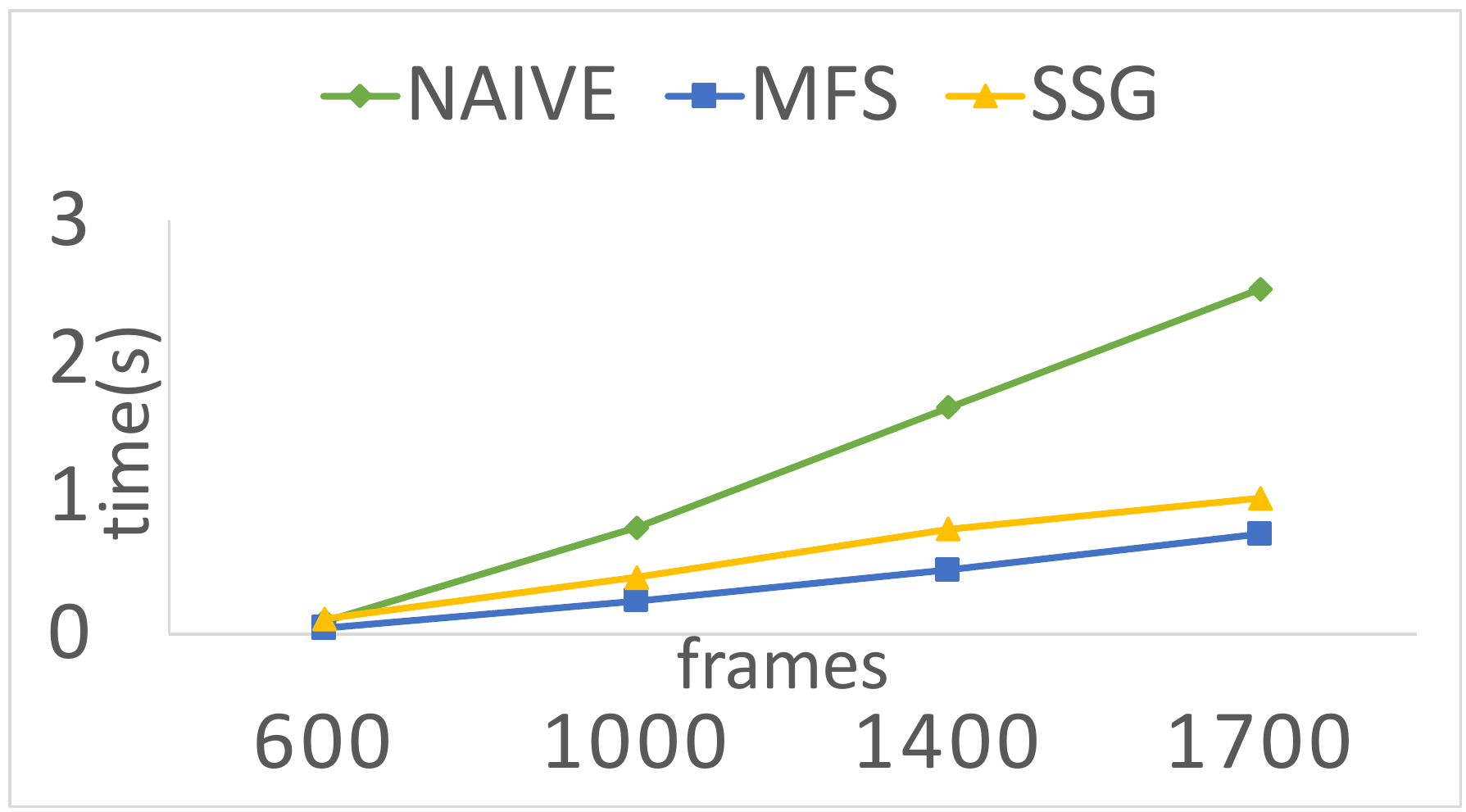}
		\caption{V2}
		\label{fig:state-frame-v2}
	\end{subfigure}
	\begin{subfigure}[b]{0.21\textwidth}
		\centering
		\includegraphics[width=\textwidth]{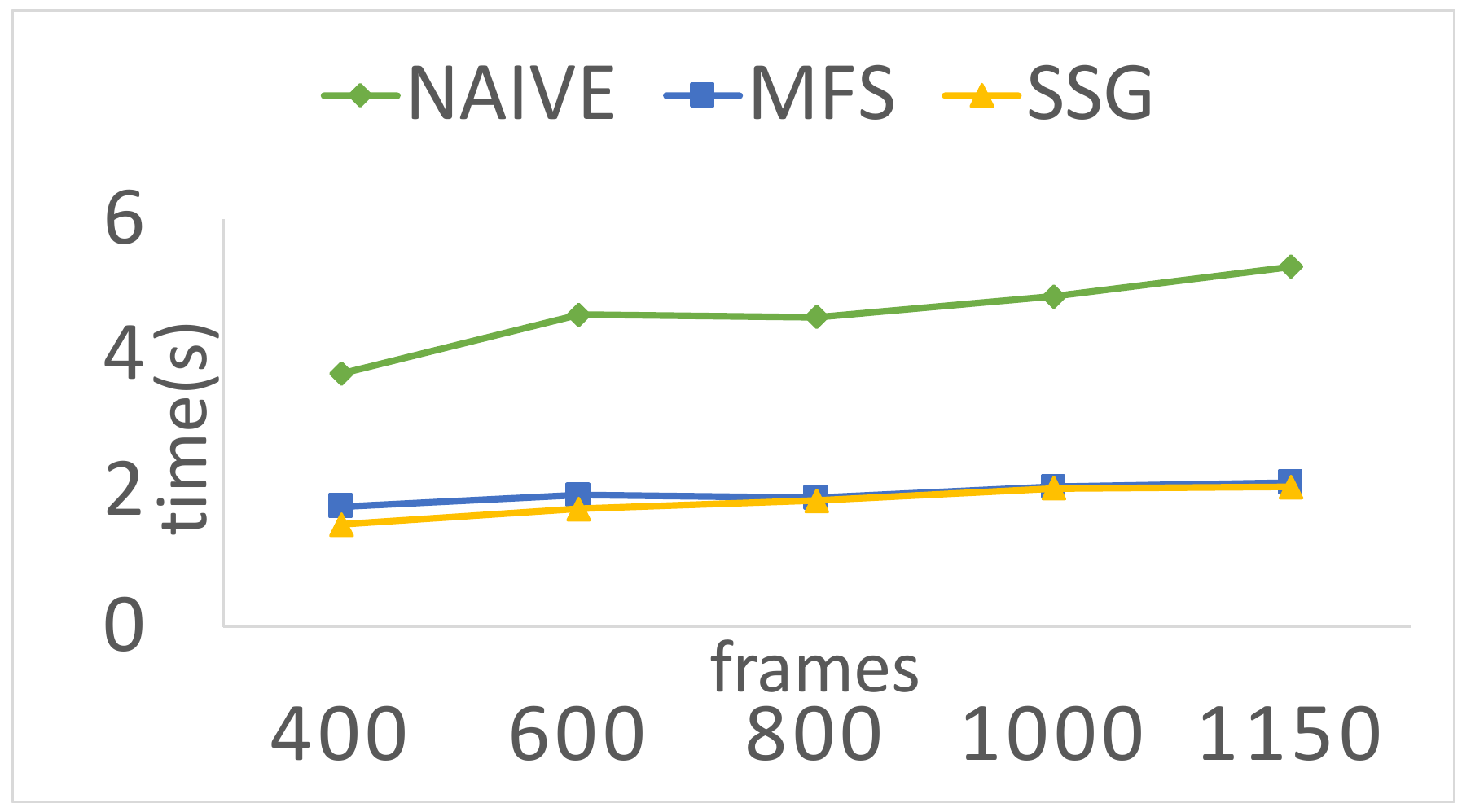}
		\caption{D1}
		\label{fig:state-frame-d1}
	\end{subfigure}
	\hfill
	\begin{subfigure}[b]{0.21\textwidth}
	\centering
	\includegraphics[width=\textwidth]{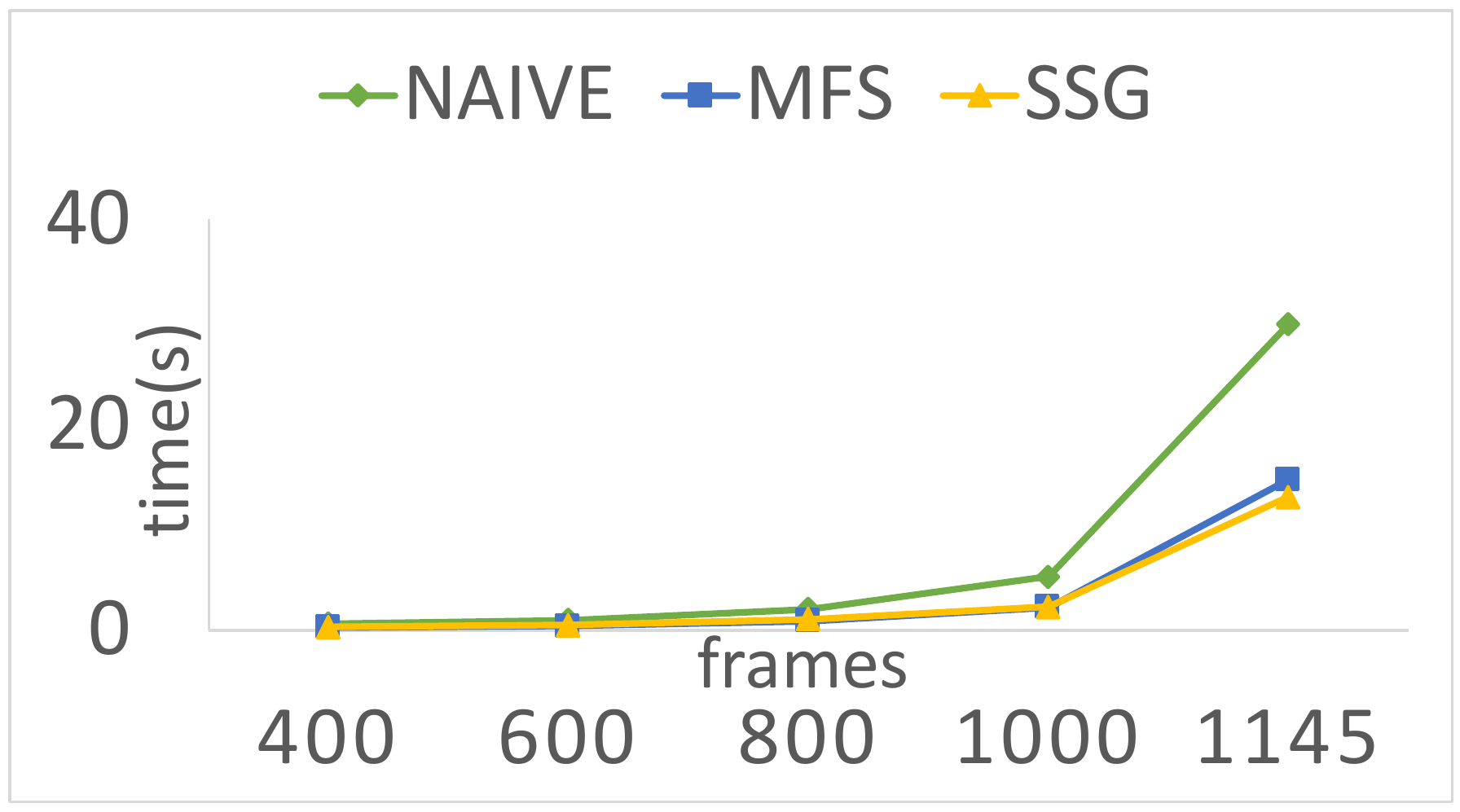}
	\caption{D2}
	\label{fig:state-frame-d2}
	\end{subfigure}
	\begin{subfigure}[b]{0.21\textwidth}
		\centering
		\includegraphics[width=\textwidth]{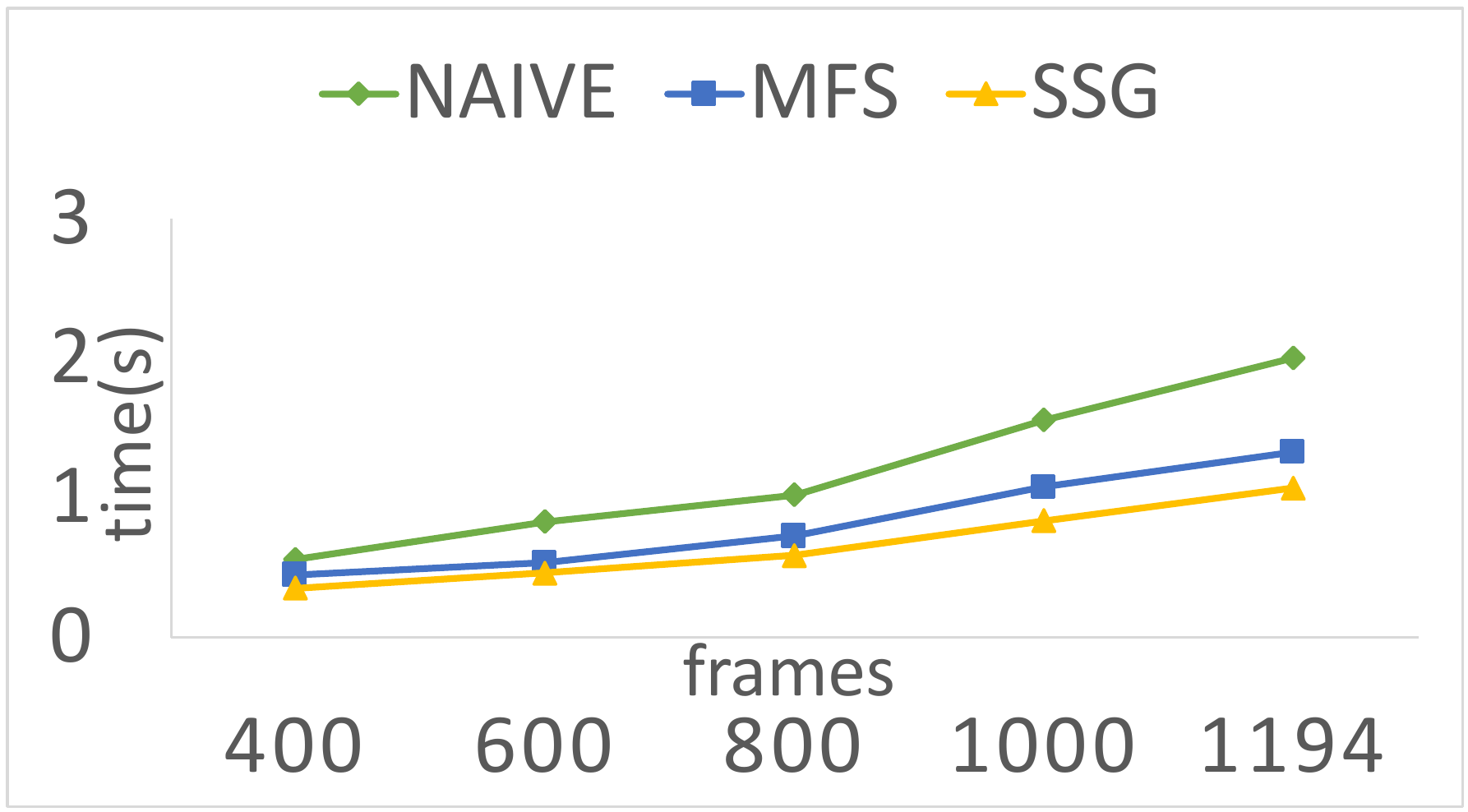}
		\caption{M1}
		\label{fig:state-frame-m1}
	\end{subfigure}
	\hfill
	\begin{subfigure}[b]{0.21\textwidth}
	\centering
	\includegraphics[width=\textwidth]{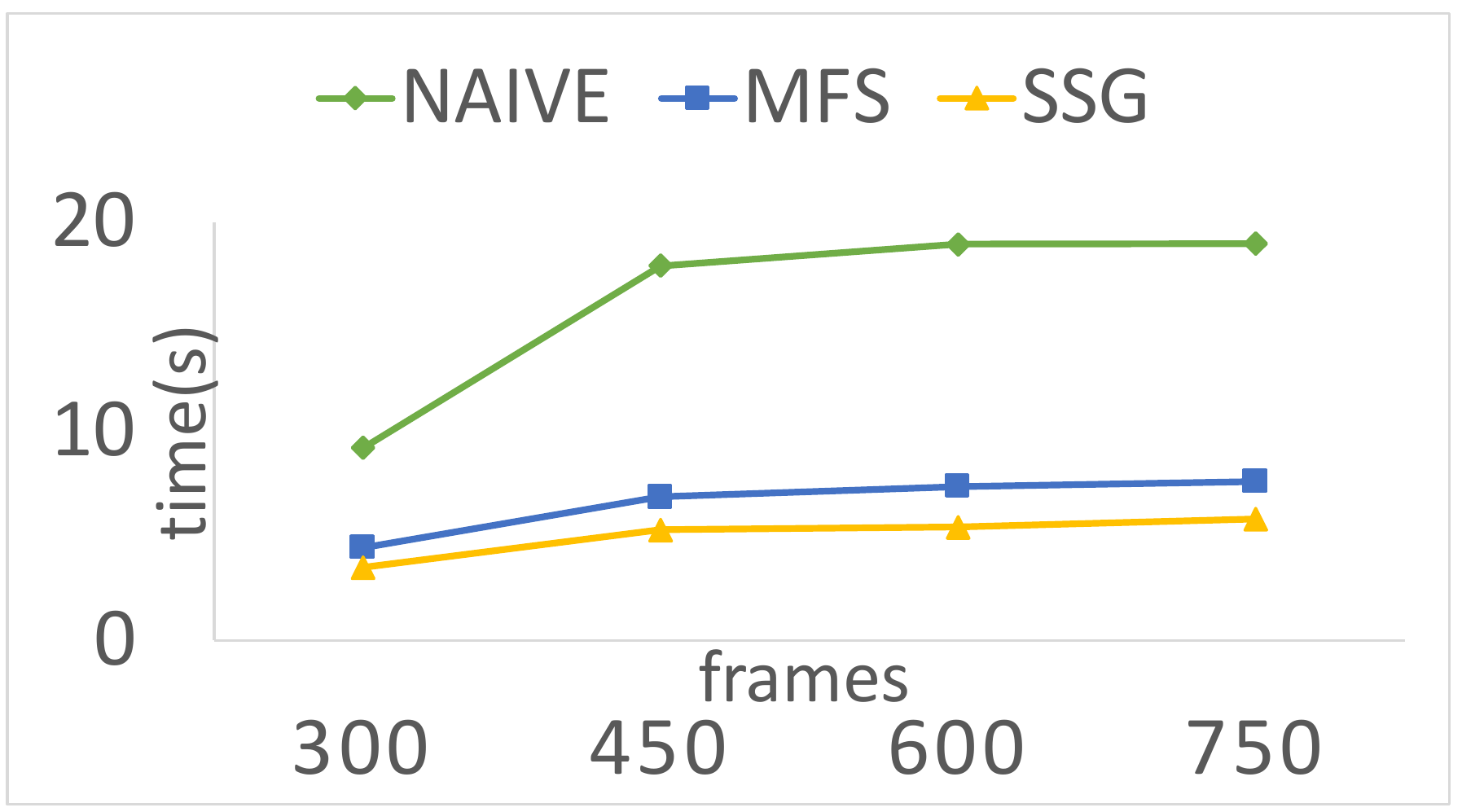}
	\caption{M2}
	\label{fig:state-frame-m2}
	\end{subfigure}
	\caption{Varying the Total Number of Frames}
	\label{fig:state-frames}
\end{figure}

With a window size $w=300$ frames, we vary the duration parameter $d$ from 180 to 270 frames shown as Figure \ref{fig:state-duration}. The performance of all methods is relatively stable since the duration parameter only influences the size of the Result State Set. All possible states still have to be maintained and computed.
Among all datasets, MFS performs best on V2, where the maximum speed up is more than 3 times than NAIVE, while SSG performs best on M2, where the maximum speedup is around 3.5 times than NAIVE.
This observation is also consistent with the previous figures, since V2 has the least number of objects per frame, while M2 has the most.

\begin{figure}[h!]
	\centering
	\begin{subfigure}[b]{0.21\textwidth}
		\centering
		\includegraphics[width=\textwidth]{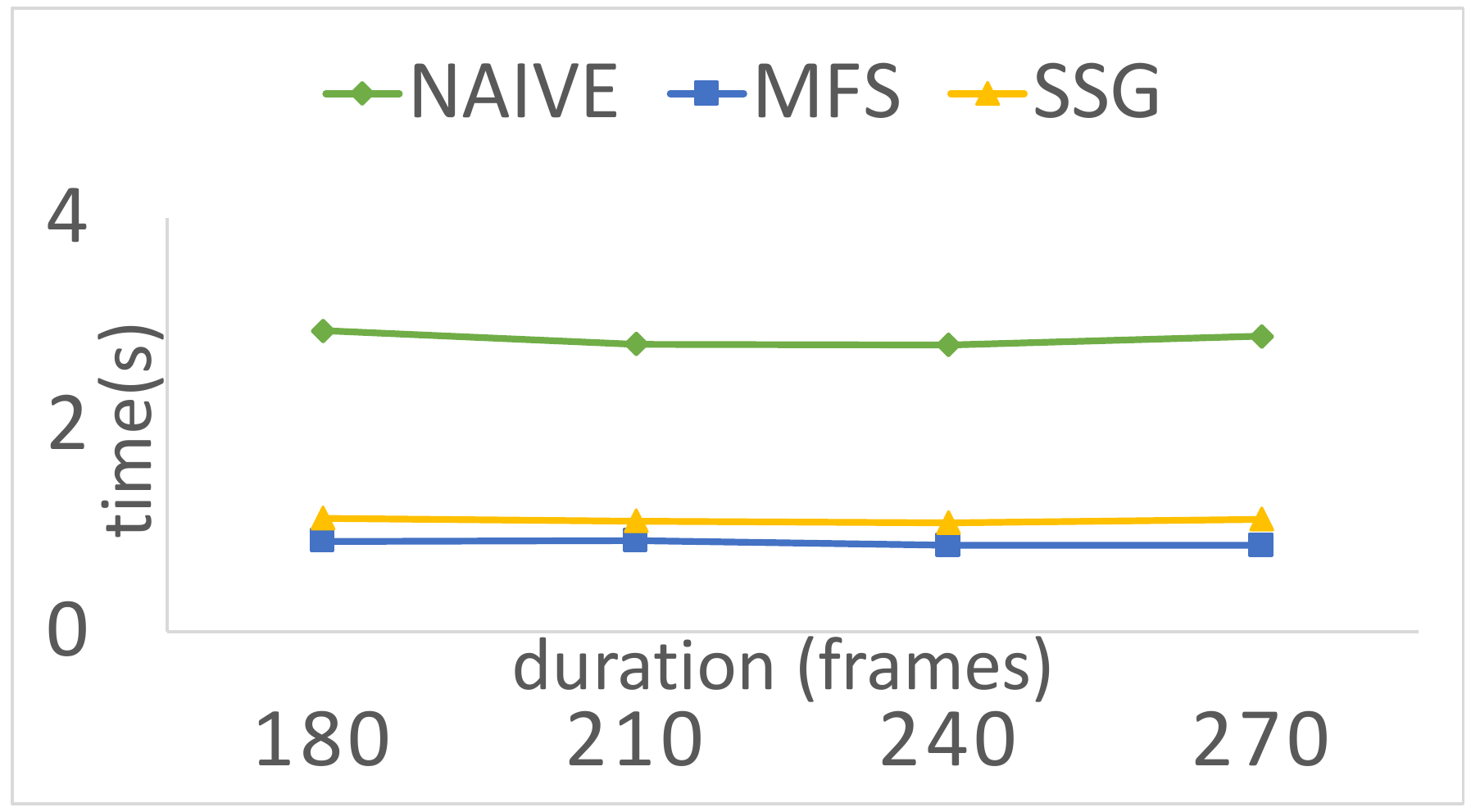}
		\caption{V1}
		\label{fig:state-duration-v1}
	\end{subfigure}
	\hfill
	\begin{subfigure}[b]{0.21\textwidth}
		\centering
		\includegraphics[width=\textwidth]{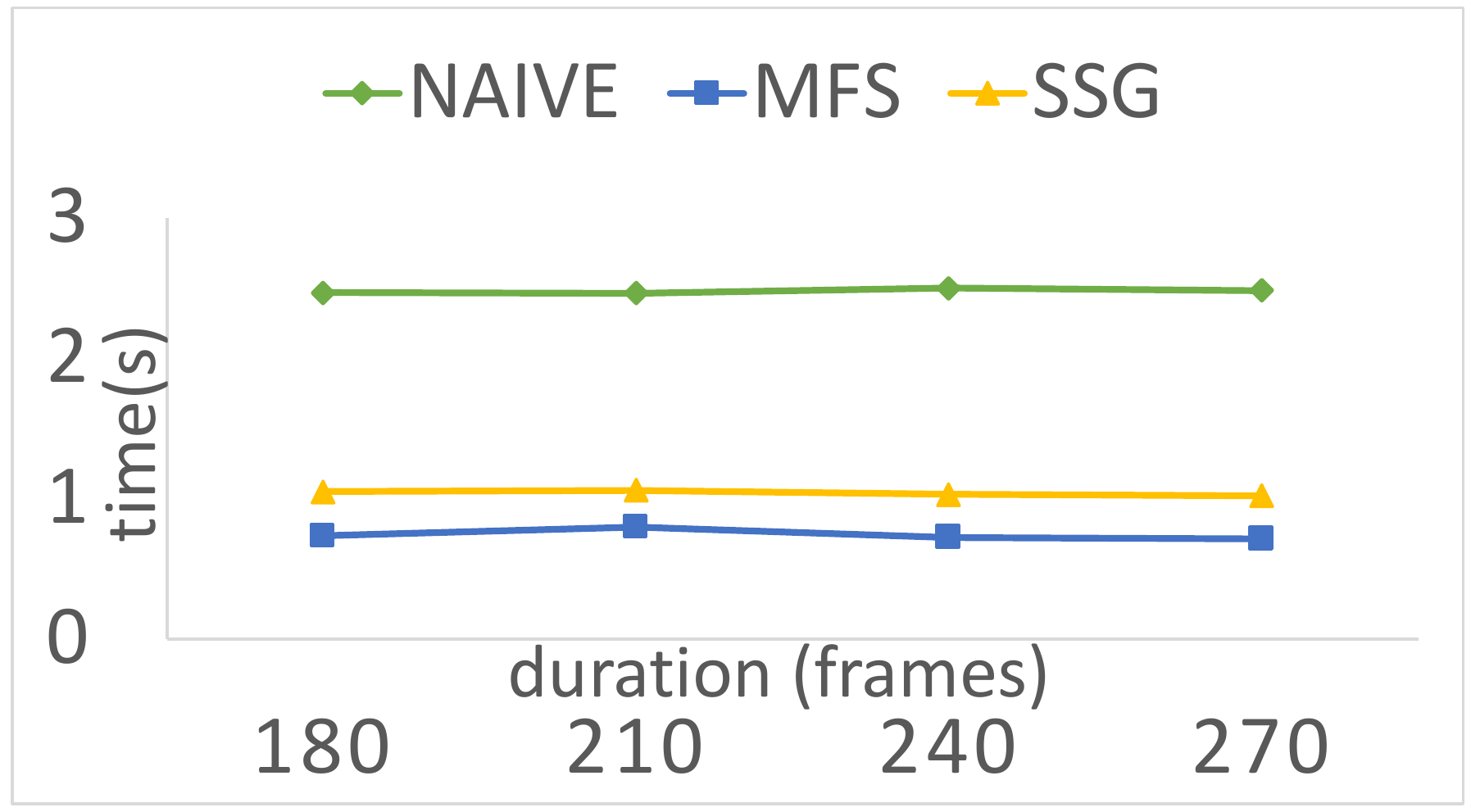}
		\caption{V2}
		\label{fig:state-duration-v2}
	\end{subfigure}
	\begin{subfigure}[b]{0.21\textwidth}
		\centering
		\includegraphics[width=\textwidth]{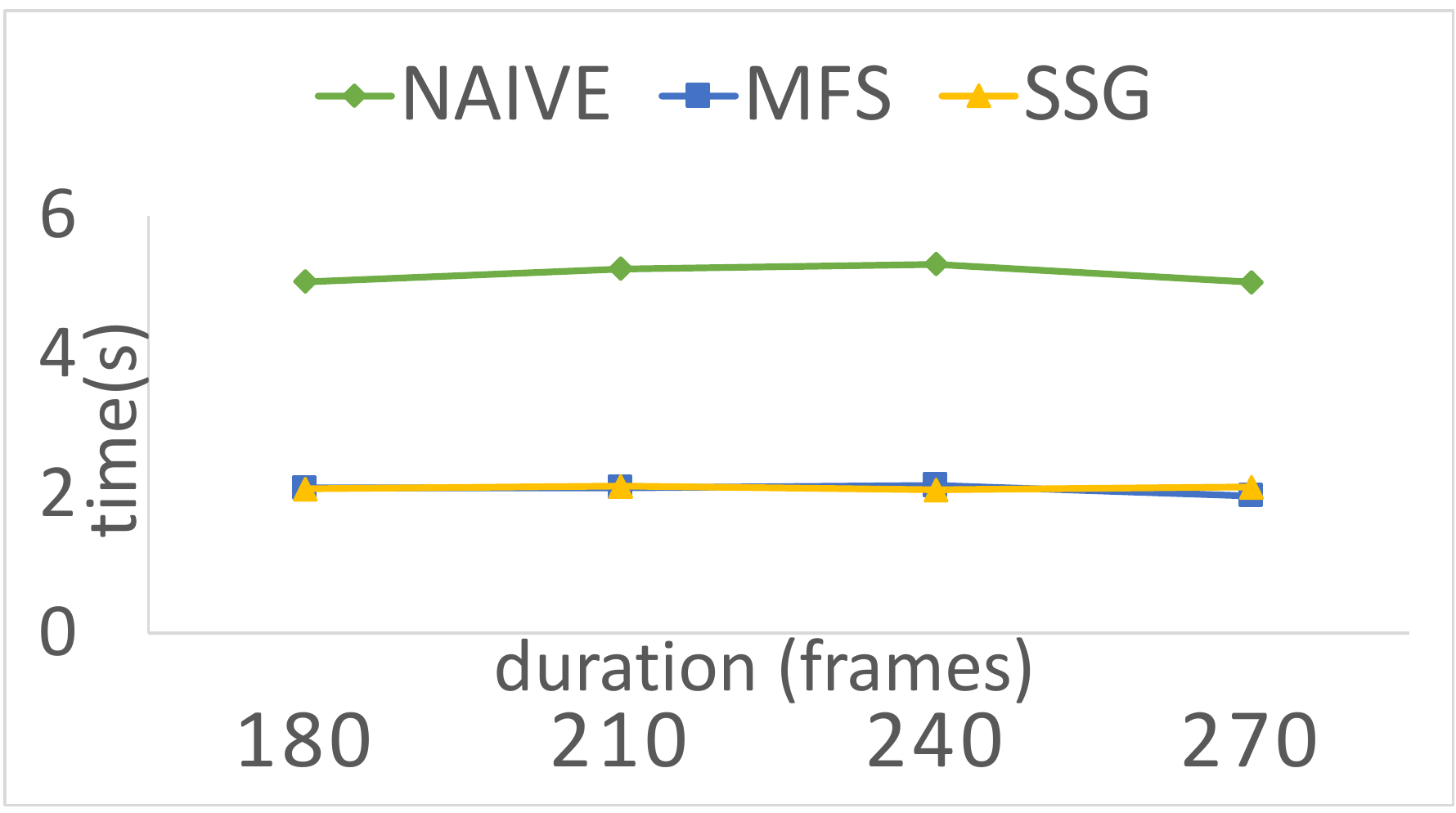}
		\caption{D1}
		\label{fig:state-duration-d1}
	\end{subfigure}
	\hfill
	\begin{subfigure}[b]{0.21\textwidth}
	\centering
	\includegraphics[width=\textwidth]{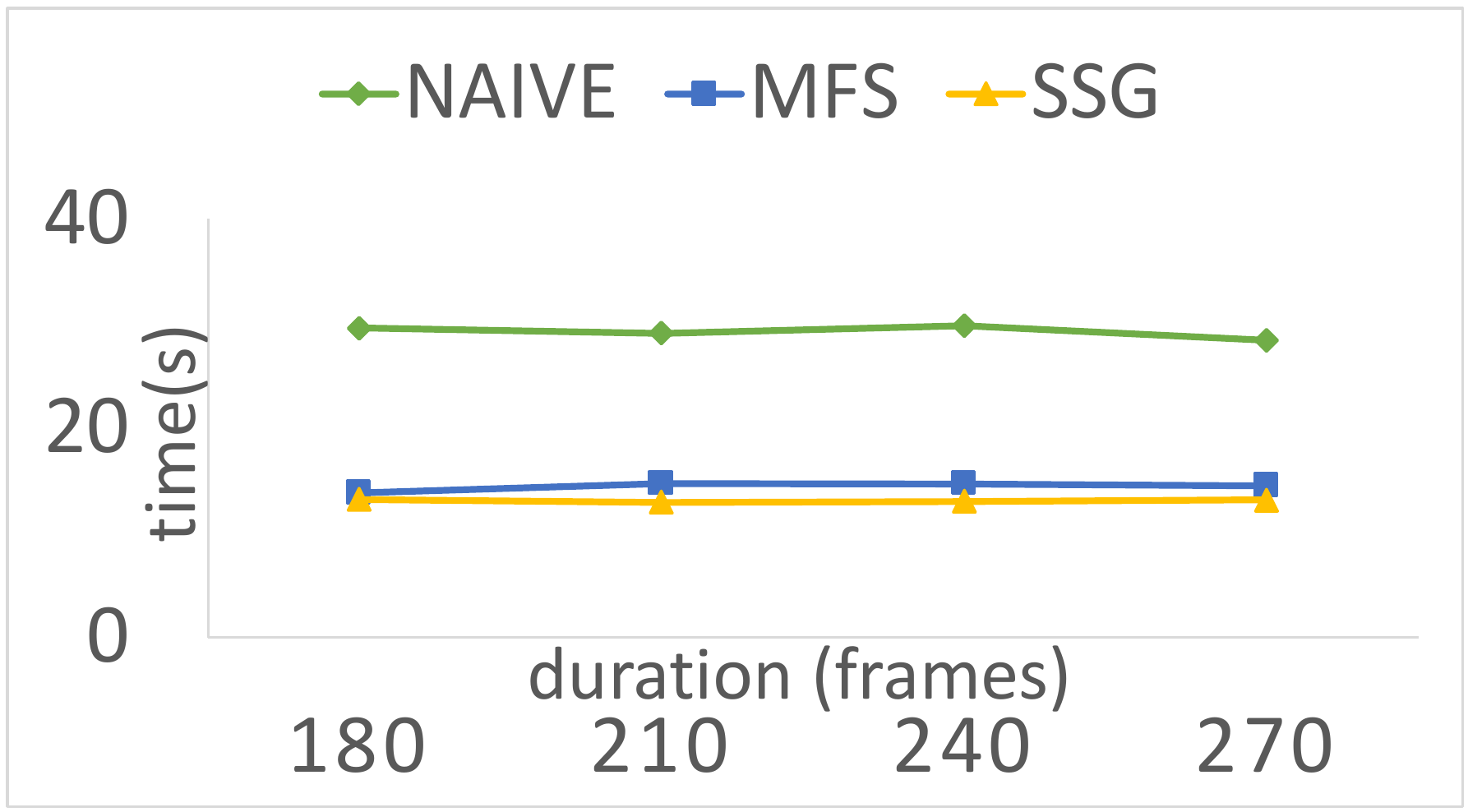}
	\caption{D2}
		\label{fig:state-duration-d2}
	\end{subfigure}
	\begin{subfigure}[b]{0.21\textwidth}
		\centering
		\includegraphics[width=\textwidth]{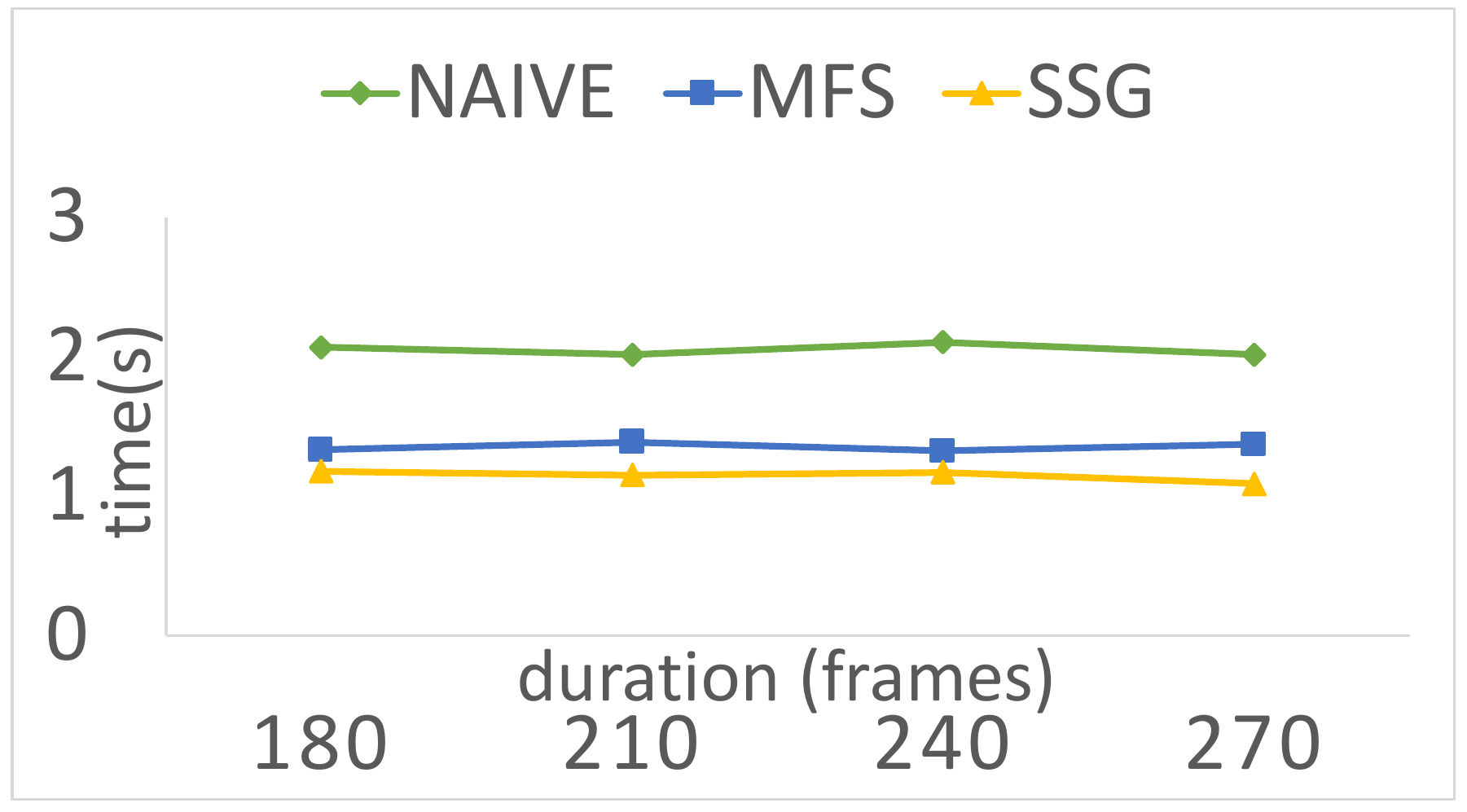}
		\caption{M1}
		\label{fig:state-duration-m1}
	\end{subfigure}
	\hfill
	\begin{subfigure}[b]{0.21\textwidth}
	\centering
	\includegraphics[width=\textwidth]{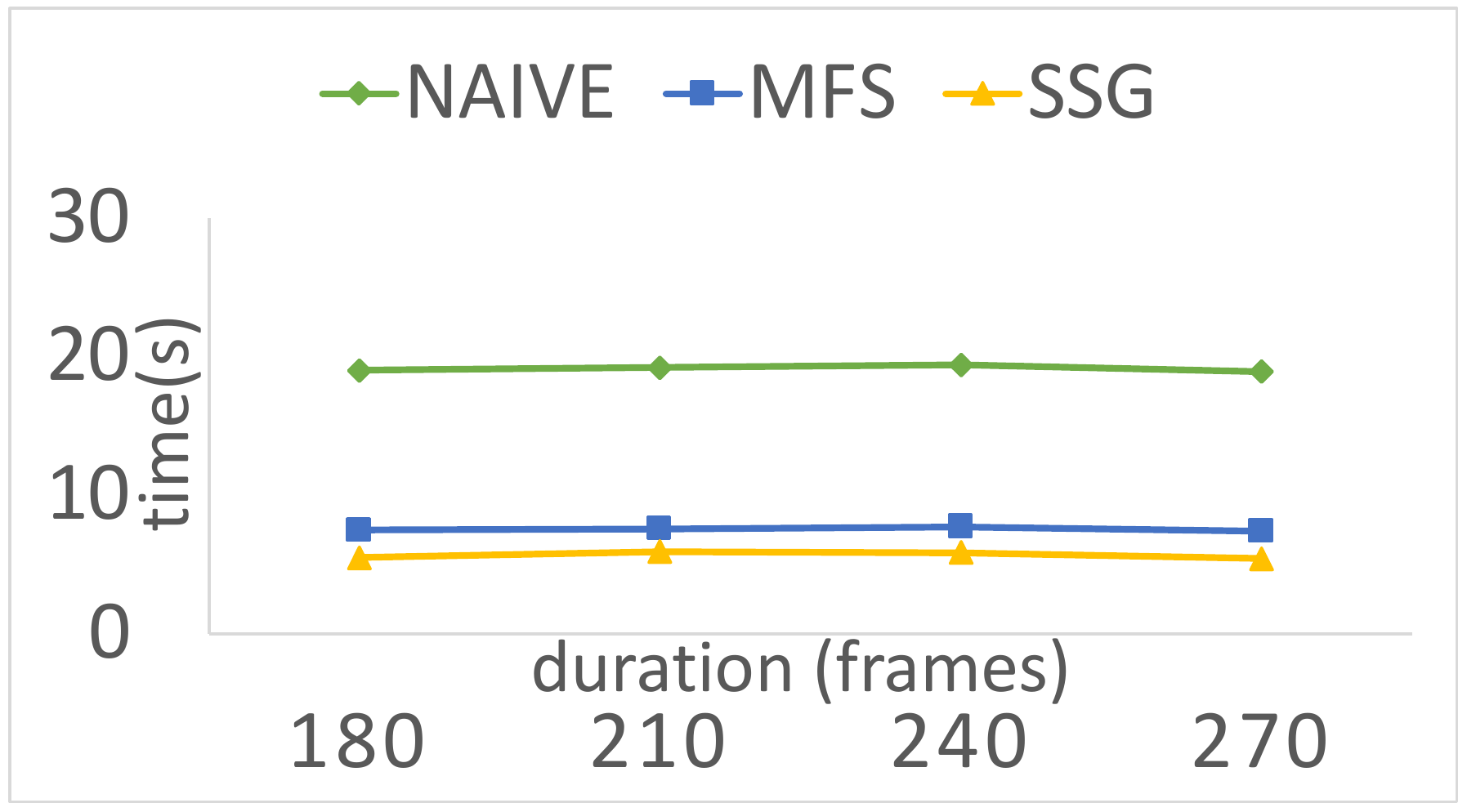}
	\caption{M2}
		\label{fig:state-duration-m2}
	\end{subfigure}
	\caption{Varying Duration $d$}
	\label{fig:state-duration}
\end{figure}

We next vary the window size $w$ and fix the duration parameter $d=240$ showing the results as Figure \ref{fig:state-window}. All methods require more time to compute as the window size increases since more states have to be maintained and computed. The NAIVE and MFS methods are penalized more, since they both have to compute object set intersection between each existing state and the new arriving frame. 
As we increase the window size, the trends remain overall the same. 
It is interesting to observe how data set characteristics affect performance.
For example, Dataset M1, M2 are captured by moving cameras; this means the duration of each object in the visible screen (Obj/F in Table \ref{tab:dataset-statistics}) is smaller and new objects are introduced frequently to the visible screen, leading to more unique states.
By increasing the window size, SSG benefits most due to its pruning power, which can be observed from Figure \ref{fig:state-window-m1} (where SSG is 40\% faster than MFS) and \ref{fig:state-window-m2} (where SSG requires almost half the time of MFS to execute).
Such trends do not exist on other datasets generated or captured by static cameras (Figure \ref{fig:state-window-v1} to \ref{fig:state-window-d2}).

\begin{figure}[h!]
	\centering
	\begin{subfigure}[b]{0.21\textwidth}
		\centering
		\includegraphics[width=\textwidth]{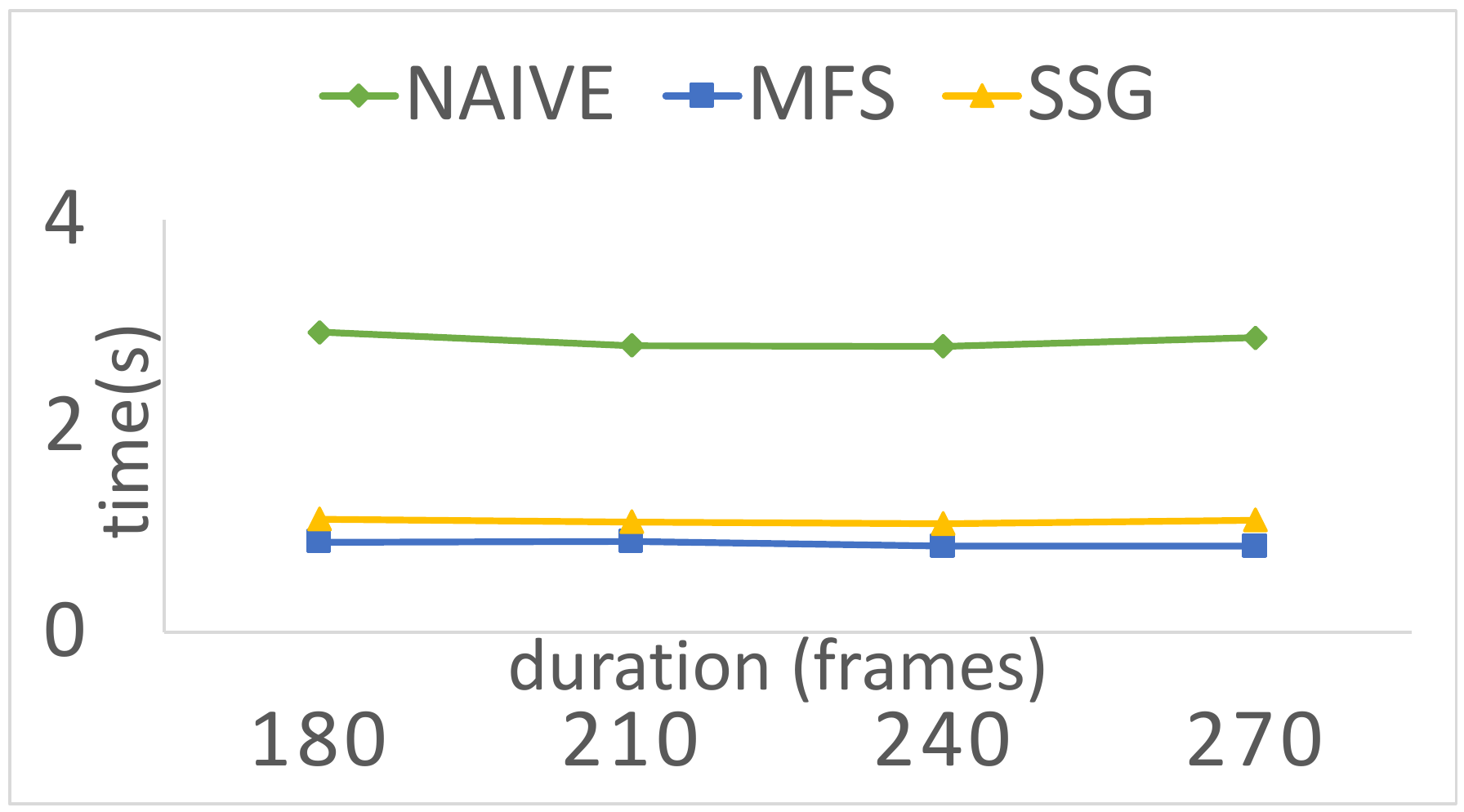}
		\caption{V1}
		\label{fig:state-window-v1}
	\end{subfigure}
	\hfill
	\begin{subfigure}[b]{0.21\textwidth}
		\centering
		\includegraphics[width=\textwidth]{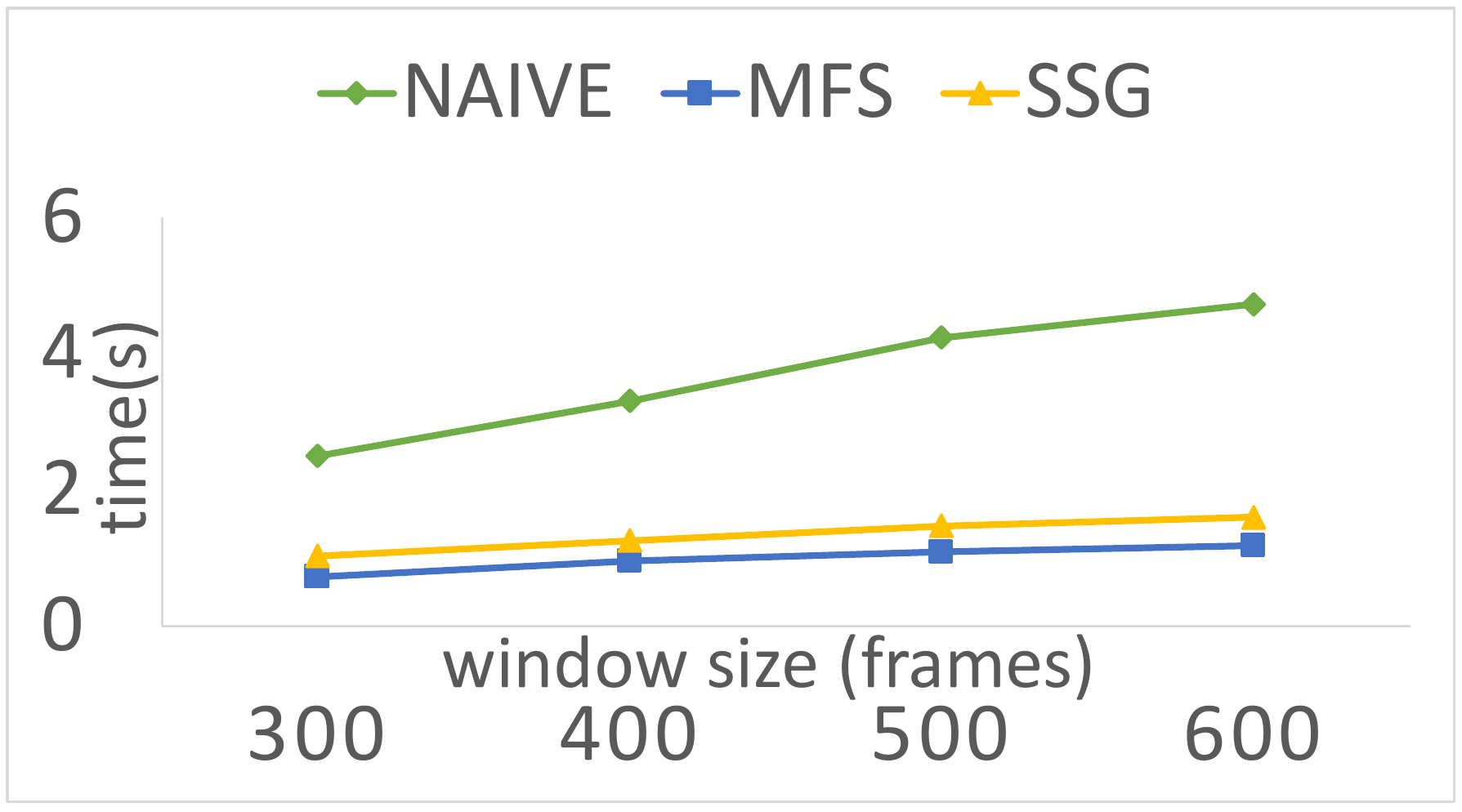}
		\caption{V2}
		\label{fig:state-window-v2}
	\end{subfigure}
	\begin{subfigure}[b]{0.21\textwidth}
		\centering
		\includegraphics[width=\textwidth]{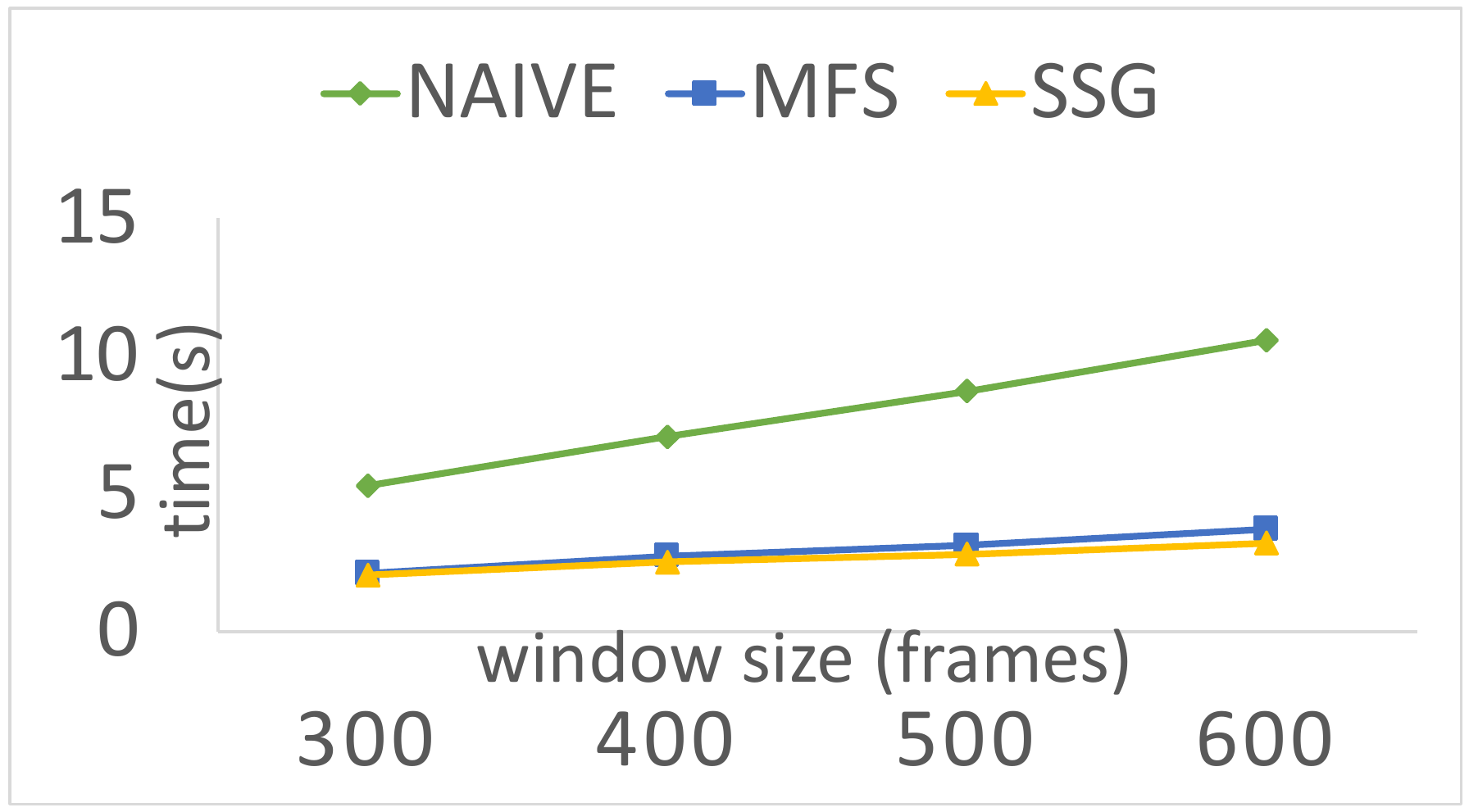}
		\caption{D1}
		\label{fig:state-window-d1}
	\end{subfigure}
	\hfill
	\begin{subfigure}[b]{0.21\textwidth}
	\centering
	\includegraphics[width=\textwidth]{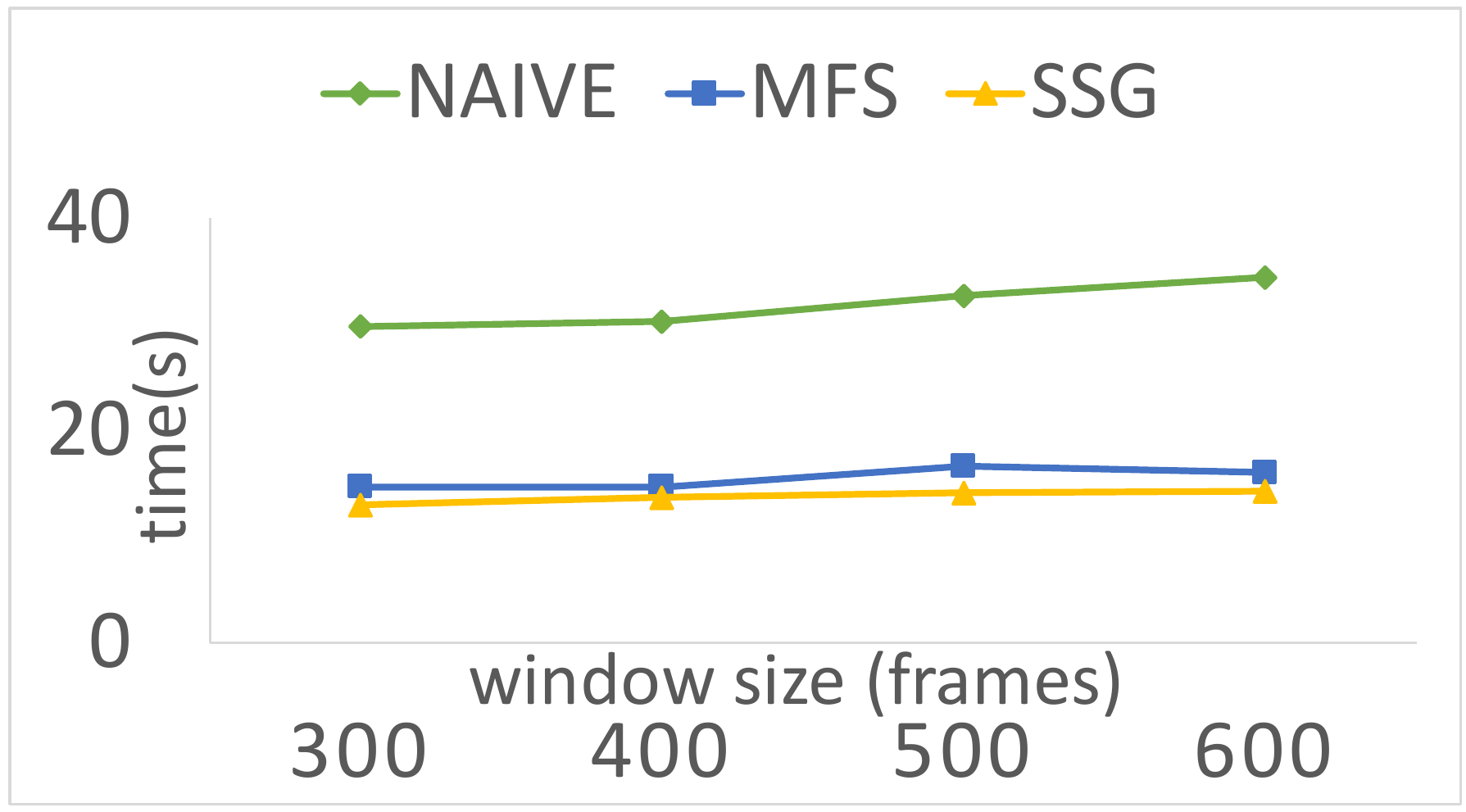}
	\caption{D2}
		\label{fig:state-window-d2}
	\end{subfigure}
	\begin{subfigure}[b]{0.21\textwidth}
		\centering
		\includegraphics[width=\textwidth]{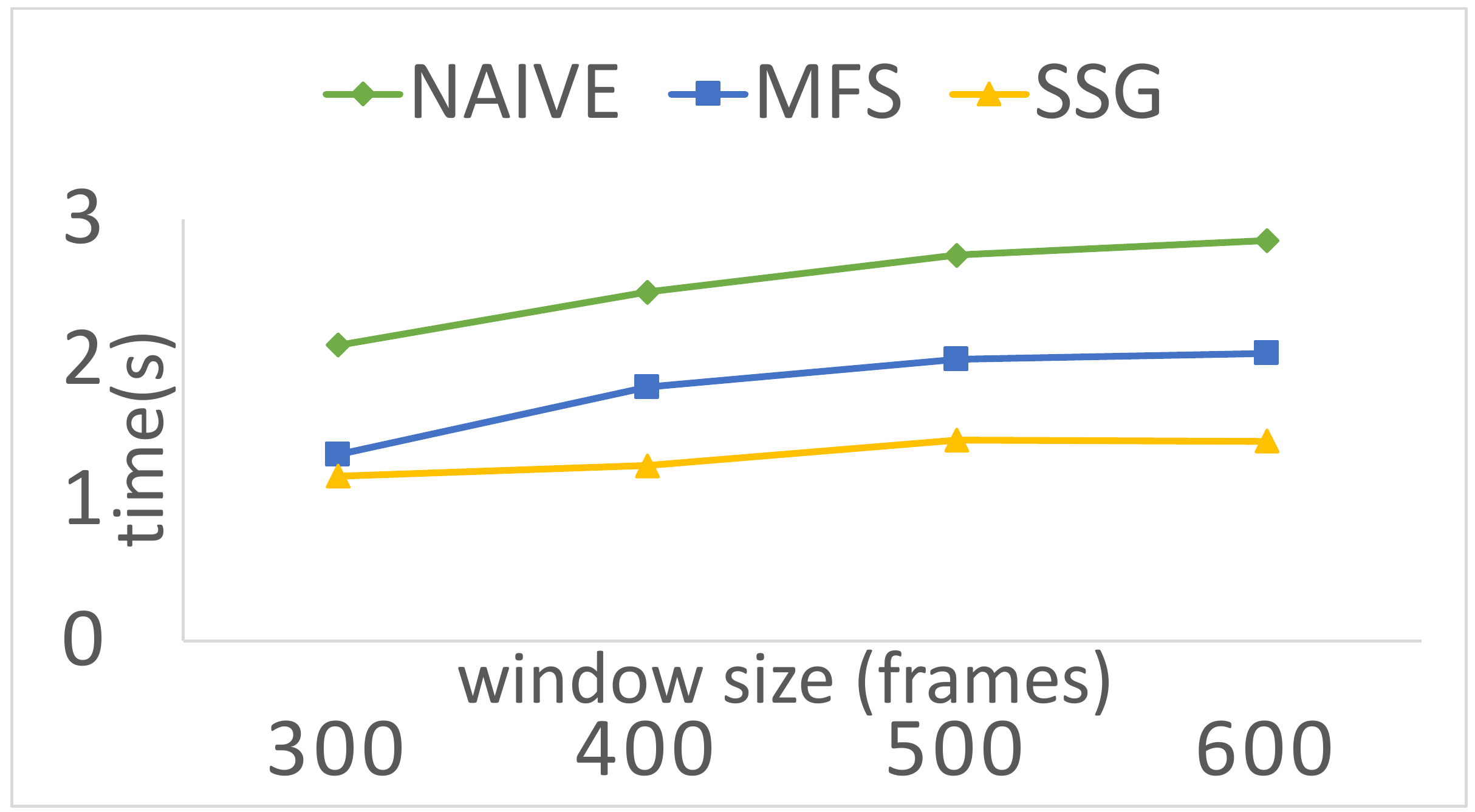}
		\caption{M1}
		\label{fig:state-window-m1}
	\end{subfigure}
	\hfill
	\begin{subfigure}[b]{0.21\textwidth}
	\centering
	\includegraphics[width=\textwidth]{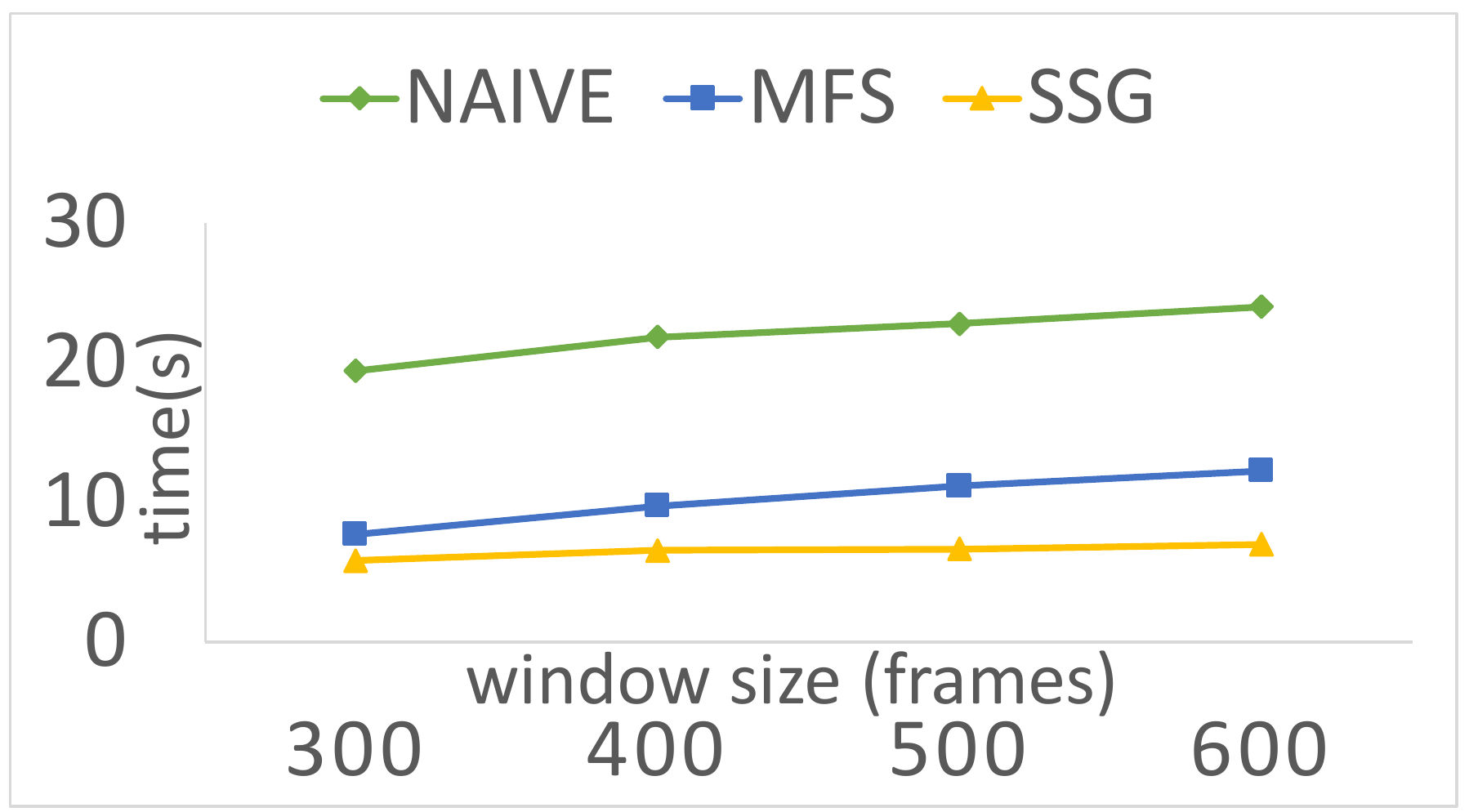}
	\caption{M2}
		\label{fig:state-window-m2}
	\end{subfigure}
	\caption{Varying Window Size $w$}
	\label{fig:state-window}
\end{figure}

So far both real and synthetic datasets utilized the intrinsic number of object occlusions present in them.
To introduce more occlusions (and be able to vary them) we reuse the same object id after an object disappears from the video.
Thus, to simulate more occlusions, we introduce an occlusion parameter, $p_o$.
Each object id will be reused at most $p_o$ times. Figure \ref{fig:state-occ} depicts the results, where $p_o$ varies from 0 to 3.
In general, more occlusions increase the chances that the intersection of object sets between states is non-empty. This penalizes the overall performance as larger number of states have to be maintained. Such trends can be observed for both synthetic and real data sets.
Since both MFS and SSG provide the ability to remove invalid states early, both methods can benefit from more occlusions. 
For example, in Figure \ref{fig:state-occ-v1}, MFS is more than 3.8 times faster than NAIVE when $p_o=3$, while SSG is more than 2.8 times faster.
As $p_o$ increases, more states may need to be explored on the graph, thus the pruning power is reduced. We can observe MFS can perform slightly better than SSG when $p_o=3$ (Figure \ref{fig:state-occ-m1}).

\begin{figure}[h!]
	\centering
	\begin{subfigure}[b]{0.21\textwidth}
		\centering
		\includegraphics[width=\textwidth]{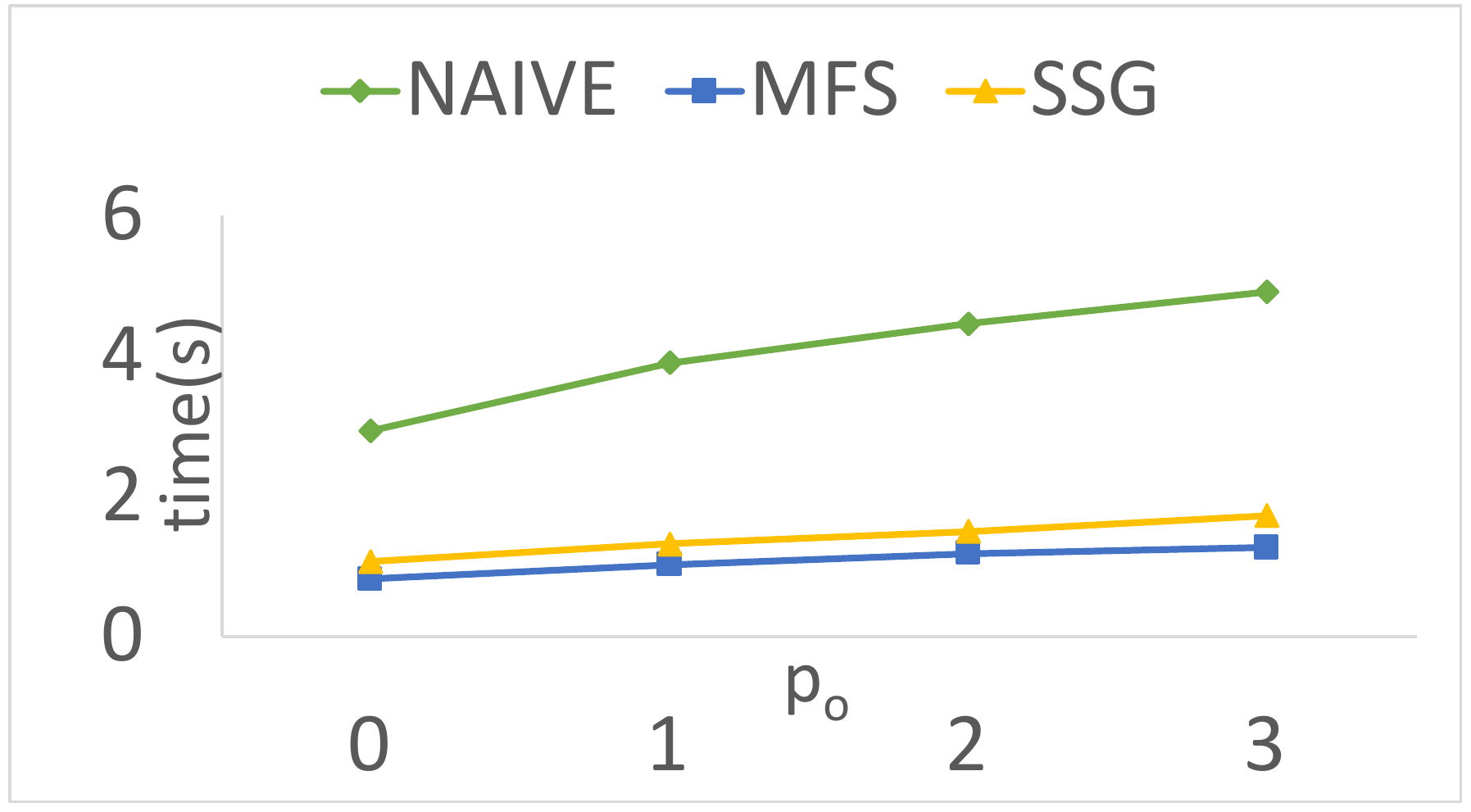}
		\caption{V1}
		\label{fig:state-occ-v1}
	\end{subfigure}
	\hfill
	\begin{subfigure}[b]{0.21\textwidth}
		\centering
		\includegraphics[width=\textwidth]{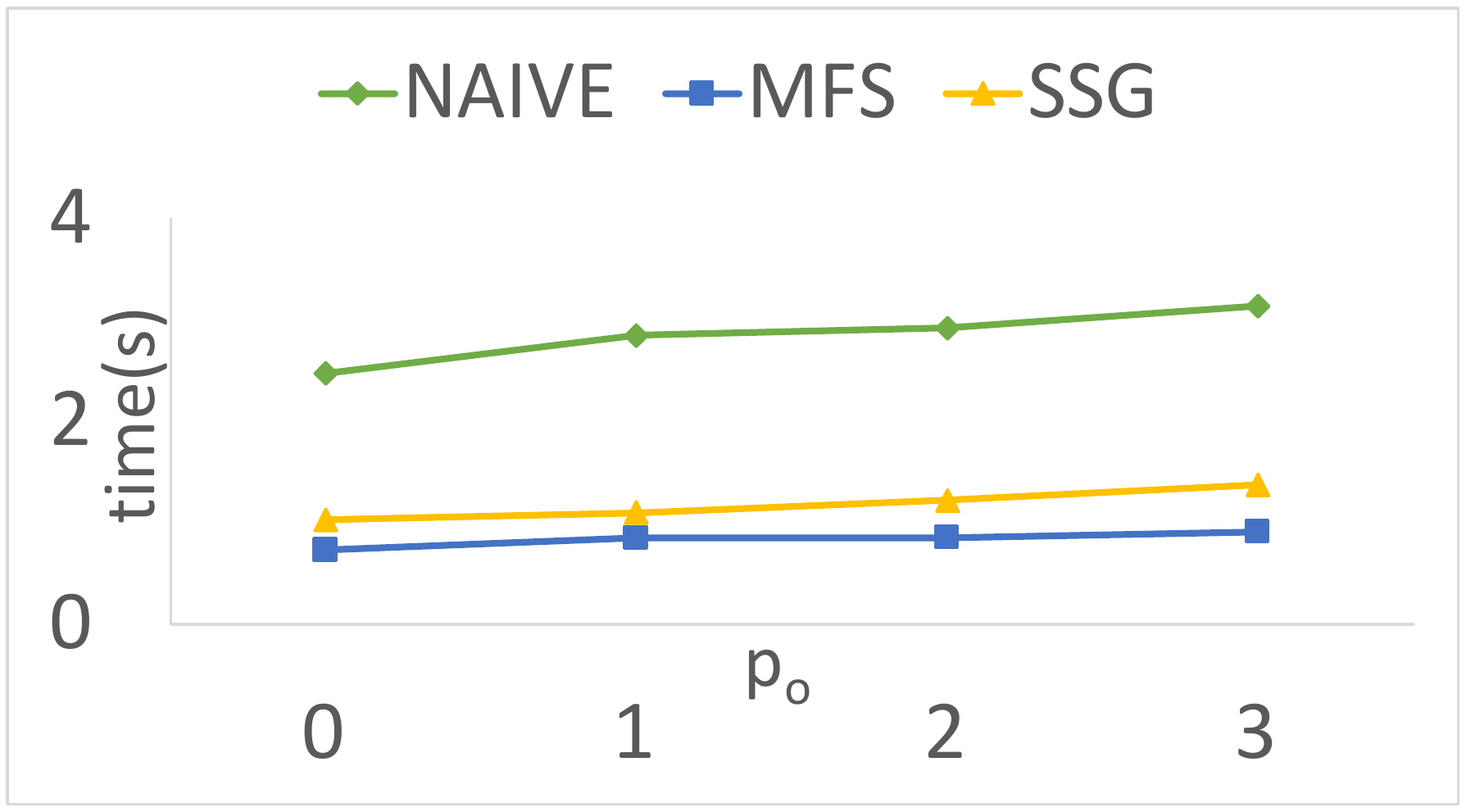}
		\caption{V2}
		\label{fig:state-occ-v2}
	\end{subfigure}
	\begin{subfigure}[b]{0.21\textwidth}
		\centering
		\includegraphics[width=\textwidth]{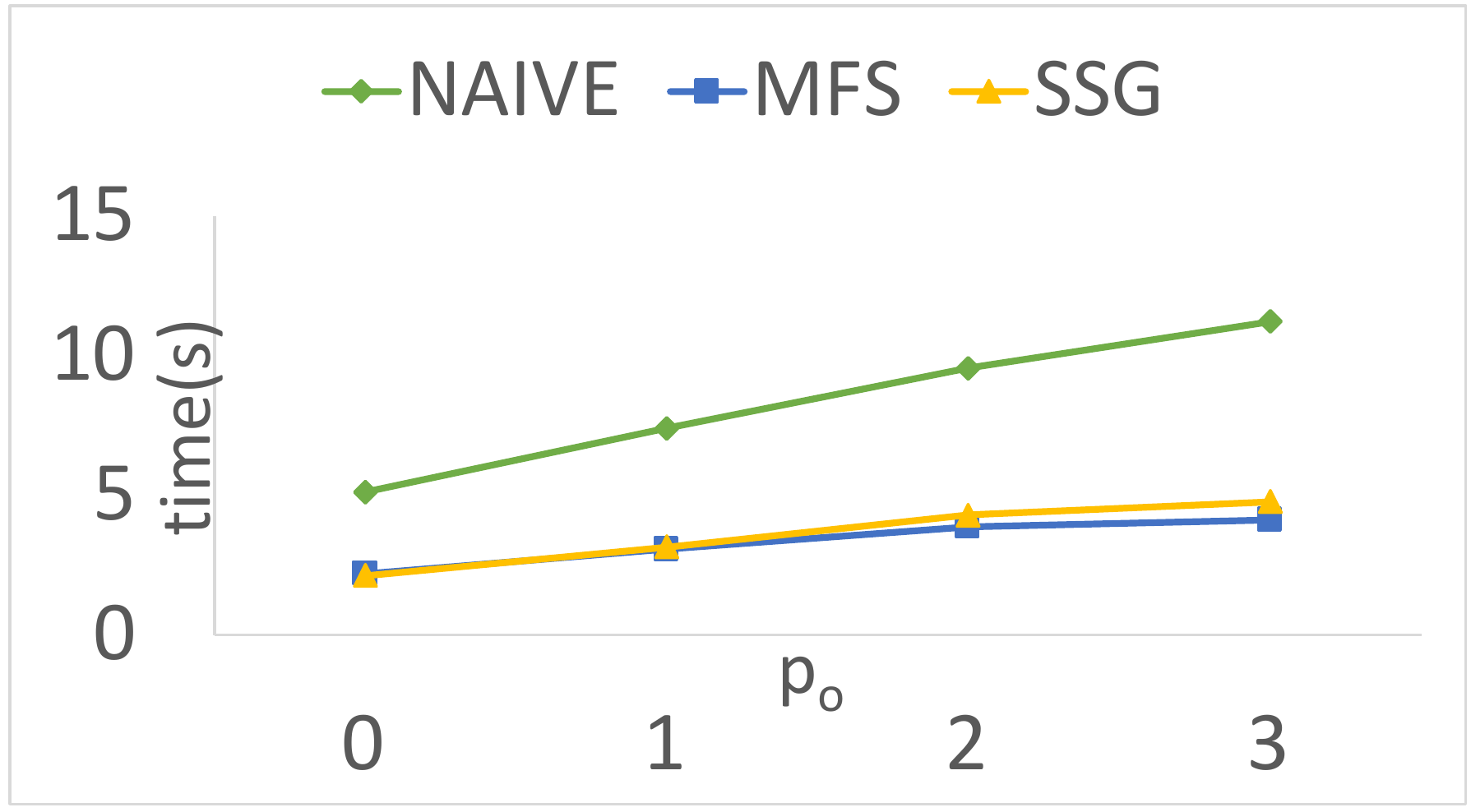}
		\caption{D1}
		\label{fig:state-occ-d1}
	\end{subfigure}
	\hfill
	\begin{subfigure}[b]{0.21\textwidth}
	\centering
	\includegraphics[width=\textwidth]{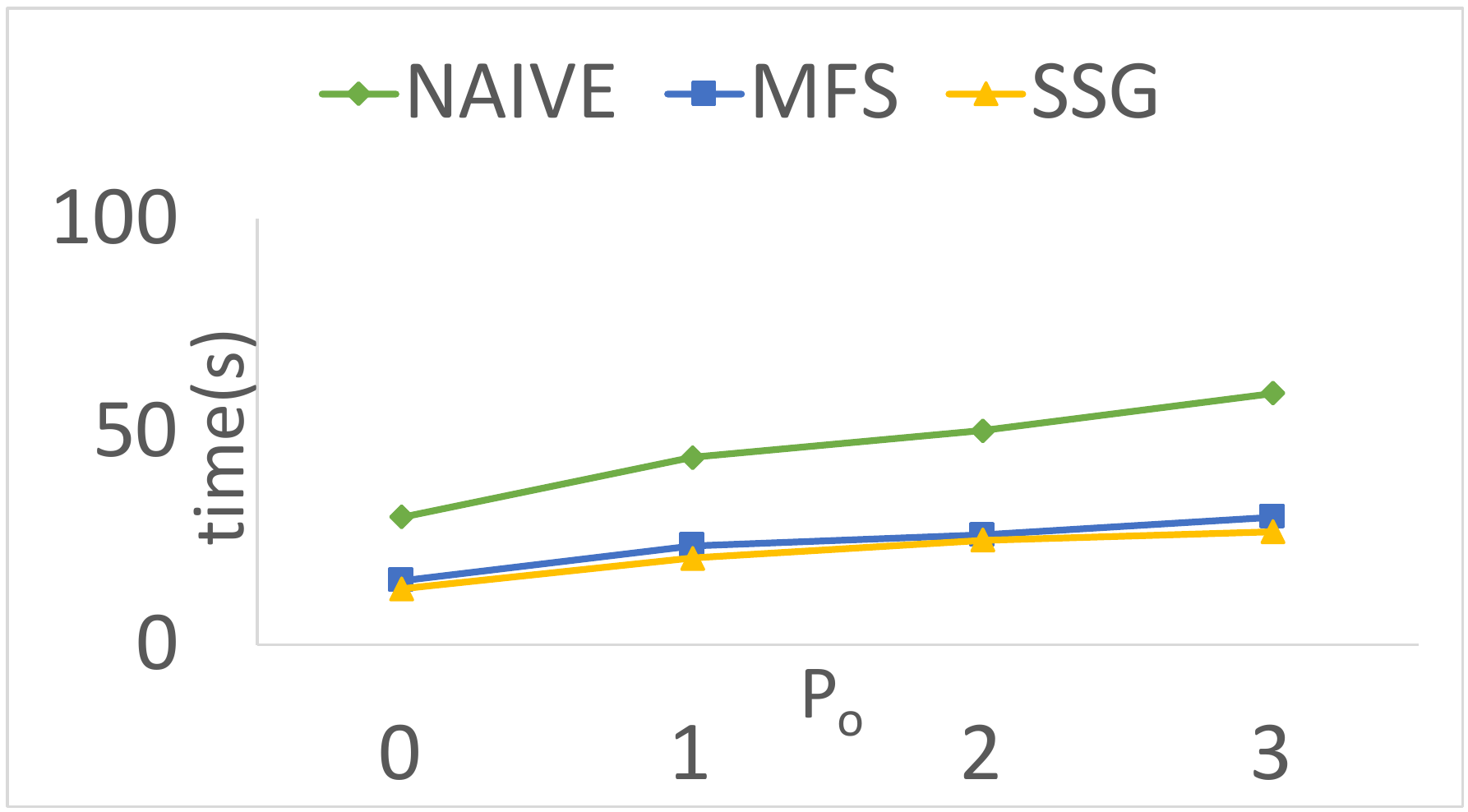}
	\caption{D2}
		\label{fig:state-occ-d2}
	\end{subfigure}
	\begin{subfigure}[b]{0.21\textwidth}
		\centering
		\includegraphics[width=\textwidth]{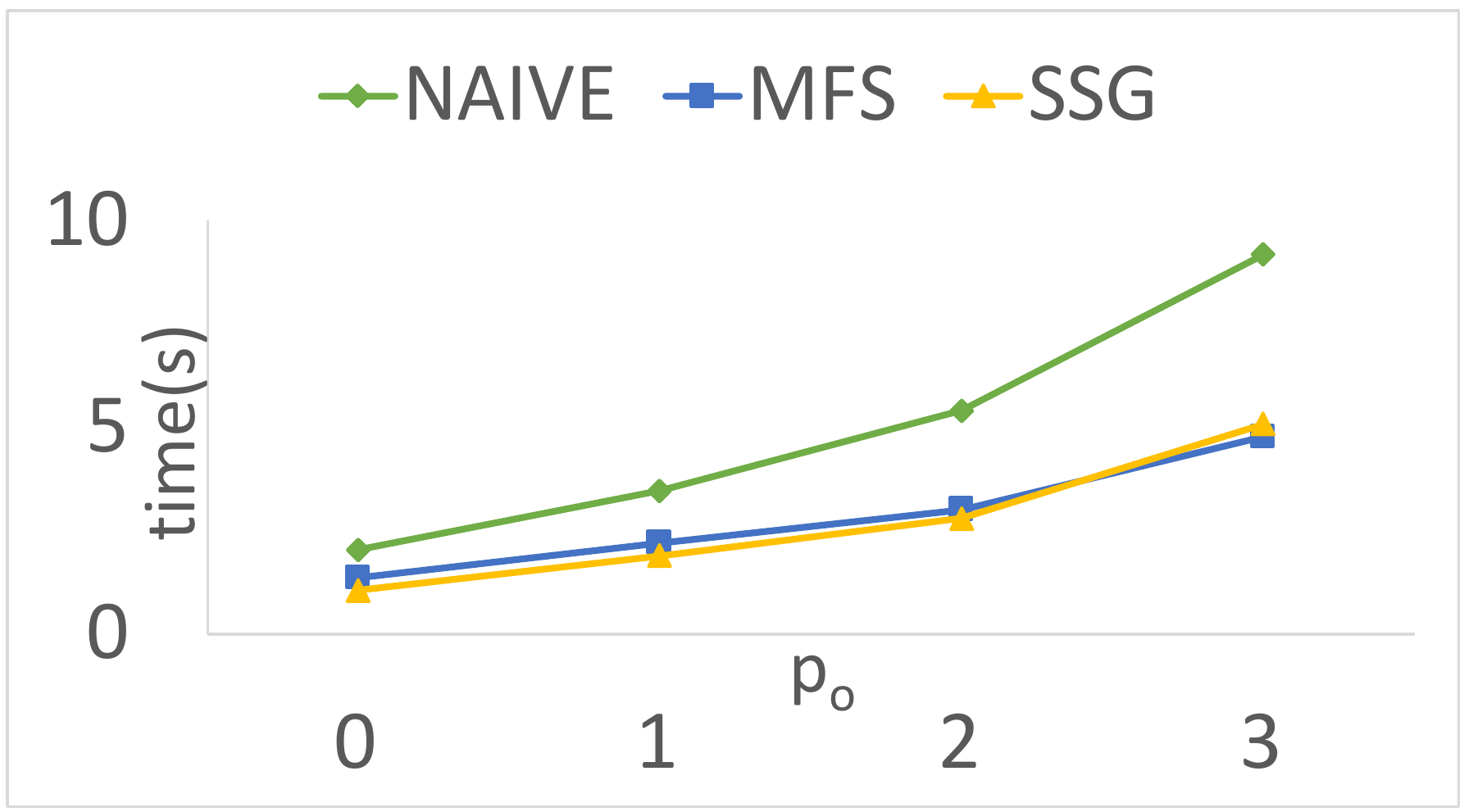}
		\caption{M1}
		\label{fig:state-occ-m1}
	\end{subfigure}
	\hfill
	\begin{subfigure}[b]{0.21\textwidth}
	\centering
	\includegraphics[width=\textwidth]{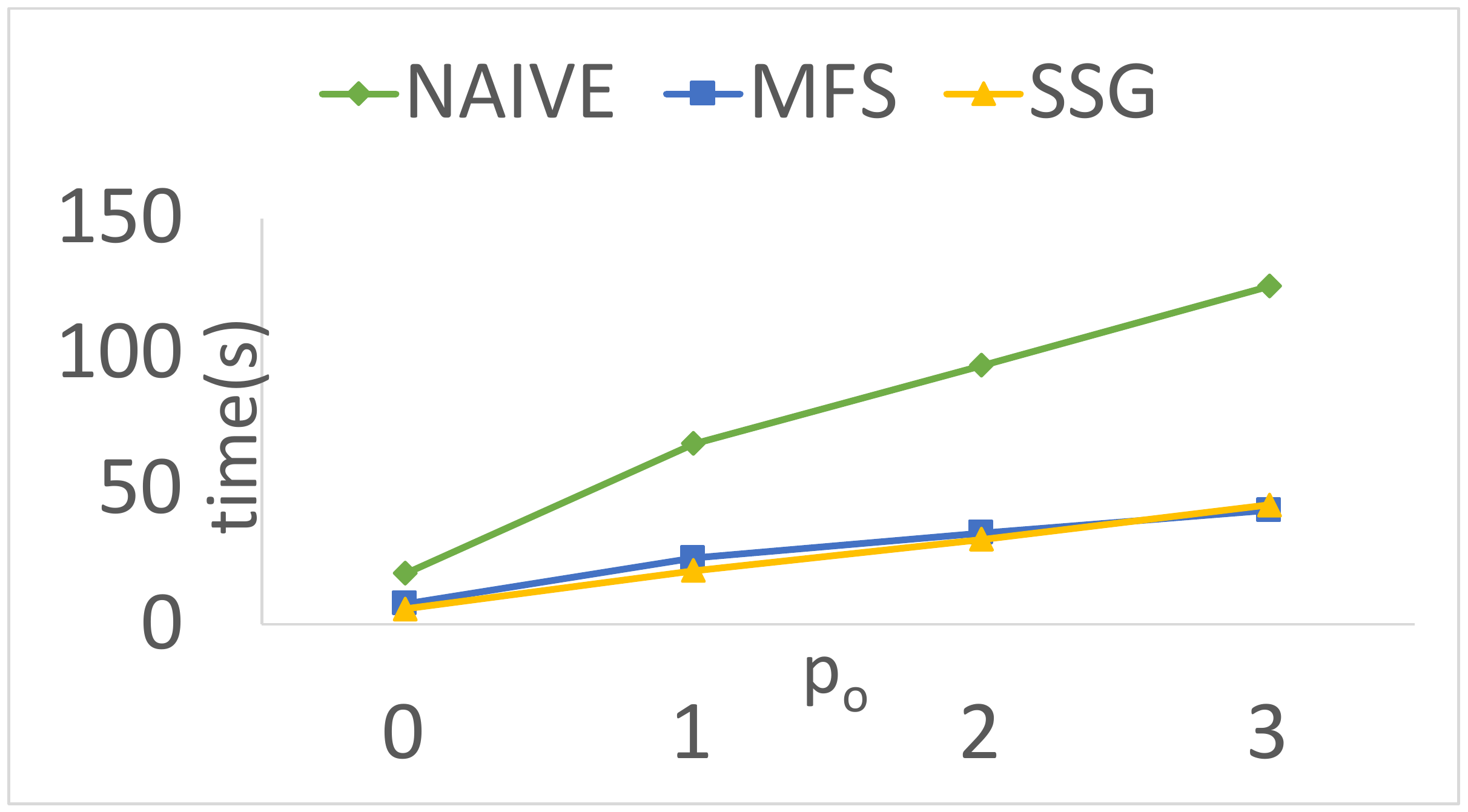}
	\caption{M2}
		\label{fig:state-occ-m2}
	\end{subfigure}
	\caption{Varying \# Occlusions $p_o$}
	\label{fig:state-occ}
\end{figure}

In summary, it is evident that both MFS and SSG offer large performance improvements compared to the NAIVE method.
On videos with small number of objects per frame, MFS could perform better than SSG.
When the number of objects per frame is relatively larger and/or the duration of objects in the visible screen is relatively smaller, SSG has the best performance, up to more than 3 times faster than the NAIVE method and can be up to one time faster than MFS. 

In general, the performance trade-off between MFS and SSG depends on the characteristics of different datasets.
On datasets with more objects per frame and/or data sets captured by moving cameras (yielding shorter duration of objects in the visible screen and increased number of new objects entering the visible screen, thus more unique states), SSG achieves higher speedups, as is observed from experimental results.

\subsection{Query Evaluation} \label{sect:exp:eval}
Based on MCOS Generation methods, we benchmark our Query Evaluation module (Section \ref{sect:eval:eval}). 
We vary the number of queries from 10 to 50, as shown in Figure \ref{fig:state-query}.
The y-axis depicts overall performance including both MCOS generation and query evaluation. The performance of each method remains almost the same even when the number of queries increases.
Both MFS and SSG are more than 2 times faster than NAIVE in Figure \ref{fig:state-query-s}, while SSG is even better than MFS in Figure \ref{fig:state-query-r}, with an overall speedup of more than 3. 
Thus, it is evident that the overheads of query evaluation on MCOS are negligible compared to that of state maintenance.

\begin{figure}
	\centering
	\begin{subfigure}[b]{0.21\textwidth}
		\centering
		\includegraphics[width=\textwidth]{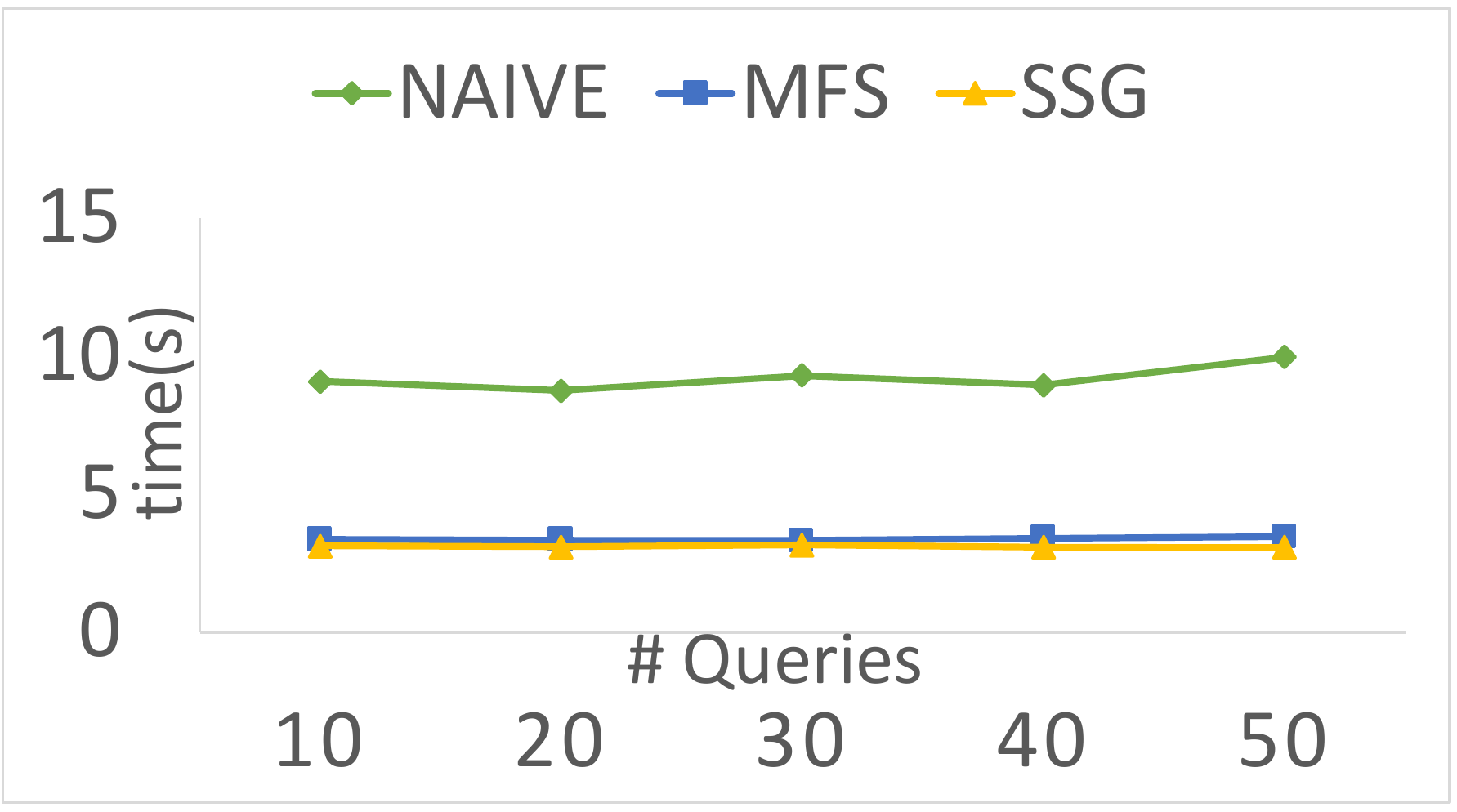}
		\caption{Synthetic (V1)}
		\label{fig:state-query-s}
	\end{subfigure}
	\hfill
	\begin{subfigure}[b]{0.21\textwidth}
		\centering
		\includegraphics[width=\textwidth]{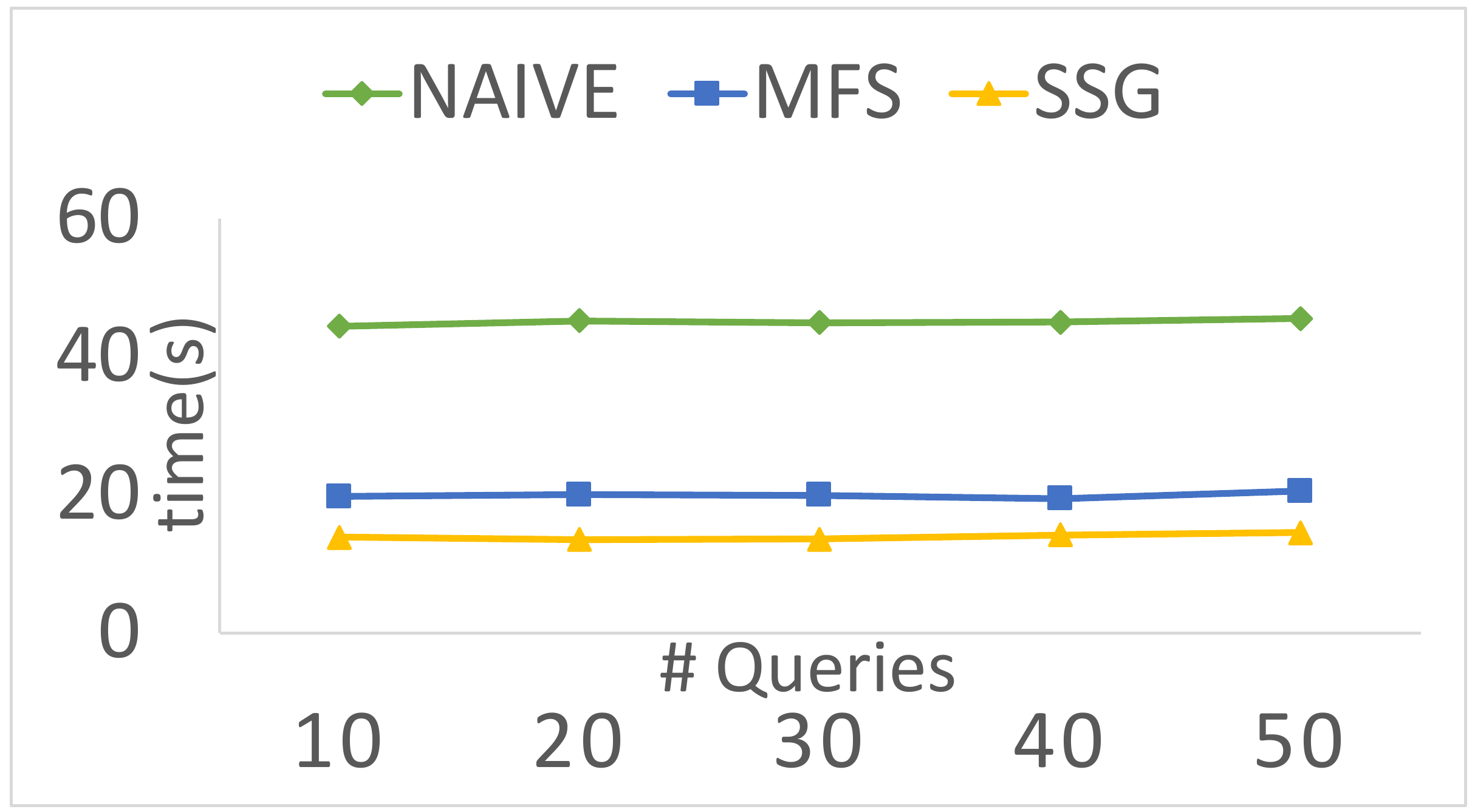}
		\caption{Real data (M2)}
		\label{fig:state-query-r}
	\end{subfigure}
	\caption{Varying \# Queries}
	\label{fig:state-query}
\end{figure}

To evaluate the pruning strategy introduced in Section \ref{sect:eval:prune}, we implemented five different methods: MFS\_O and SSG\_O, which implement the pruning strategy based on MFS and SSG methods; NAIVE\_E, MFS\_E and SSG\_E, which only follow the CNFEvalE method proposed in Section \ref{sect:eval:eval}. 
We generate 100 queries containing $\ge$ conditions only.
Let $C$ be the set containing all conditions from the given queries; we use $n_{min}= \min_{\forall c \in C} n_c$ to denote the minimum value in all conditions, where $n_c$ is the threshold value in condition $c$ (conditions are of the form $x \geq n_c$, where $x$ is the object class).
We vary $n_{min}$ from 1 to 9, the evaluation result on real datasets is shown as Figure \ref{fig:state-oeval-real}.
As $n_{min}$ increases, methods adopting SSG perform better than those adopting MFS in general. For example, in Figure \ref{fig:state-oeval-m1}, when $n_min=3$, SSG\_O is 3 times faster than NAIVE, while MFS\_O is only 2 times faster. Similarly, in Figure \ref{fig:state-oeval-d2}, when $n_min=5$, SSG\_O is also around one time faster than MFS\_O compared to NAIVE, where the speedup is five times.
When $n_{min}$ is large enough, we observe a significant performance improvement on all methods that adopt the pruning strategy (MFS\_O and SSG\_O). Among them, the SSG\_O has the best performance in general.
When $n_{min}=1$, methods adopting the pruning strategy are slightly more expensive on some datasets (Figure \ref{fig:state-oeval-d1}, \ref{fig:state-oeval-d2}) since more states in the graph are evaluated on the given queries.
We observe that methods based on SSG are always the best compared to other approaches, while the pruning strategy offers significant improvements on both MFS and SSG methods. 
When the pruning strategy is enabled, the optimized approach provides more than 100 times speedup on real datasets, as is evident in Figure \ref{fig:state-oeval-real}, where the $n_{min}$ is set to 9.

\begin{figure}
	\centering
	\begin{subfigure}[b]{0.22\textwidth}
		\centering
		\includegraphics[width=\textwidth]{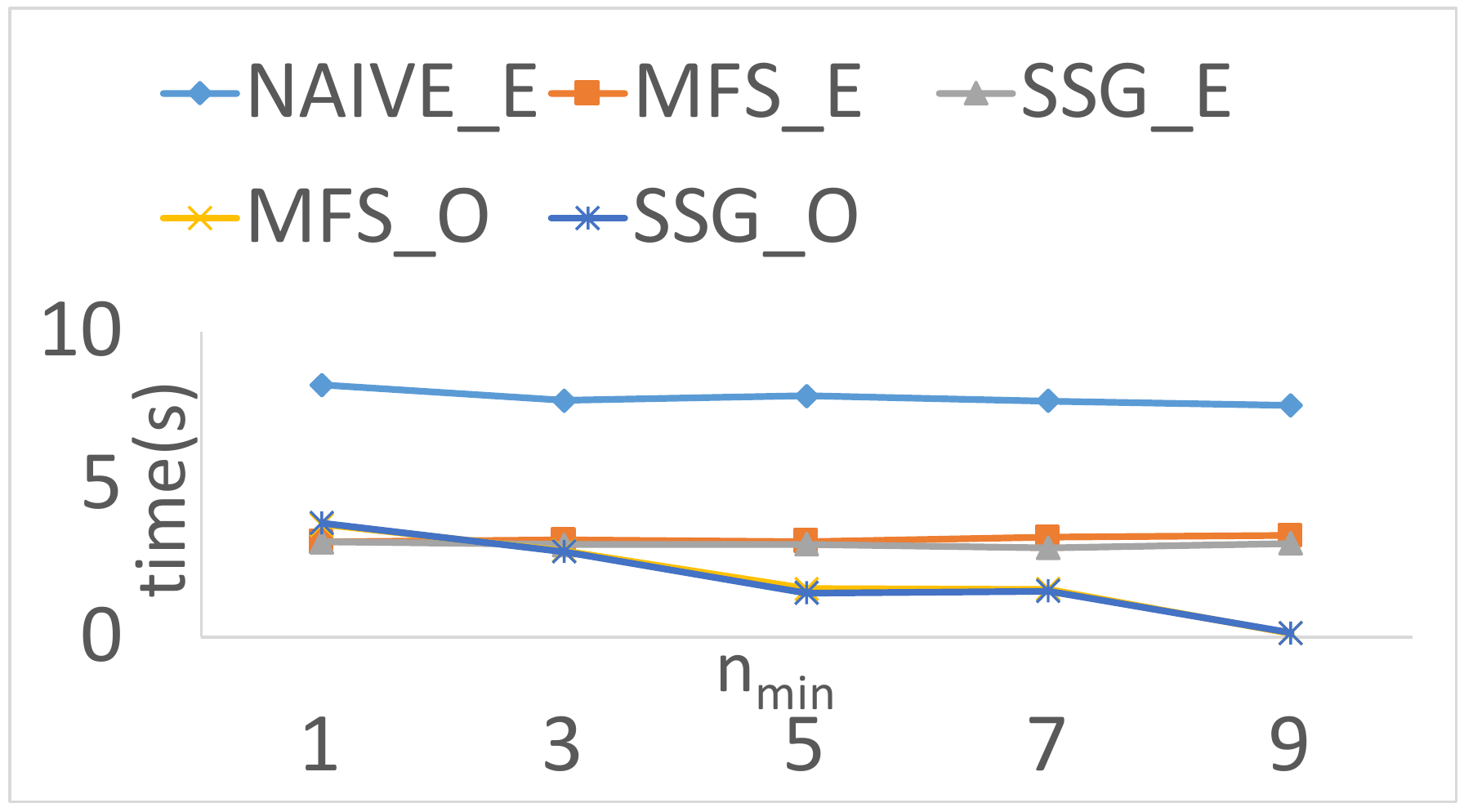}
		\caption{D1}
		\label{fig:state-oeval-d1}
	\end{subfigure}
	\hfill
	\begin{subfigure}[b]{0.22\textwidth}
		\centering
		\includegraphics[width=\textwidth]{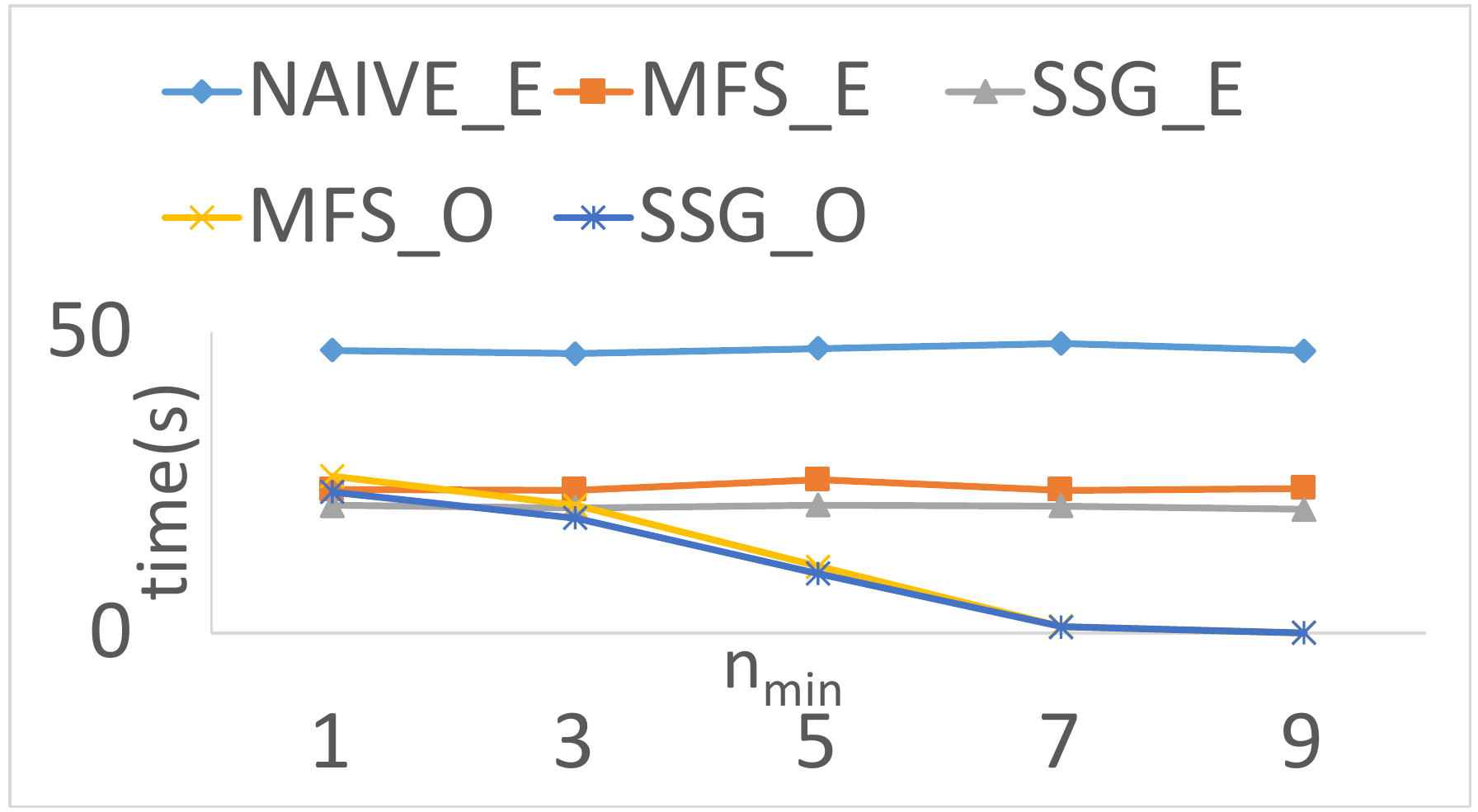}
		\caption{D2}
		\label{fig:state-oeval-d2}
	\end{subfigure}
	\begin{subfigure}[b]{0.22\textwidth}
		\centering
		\includegraphics[width=\textwidth]{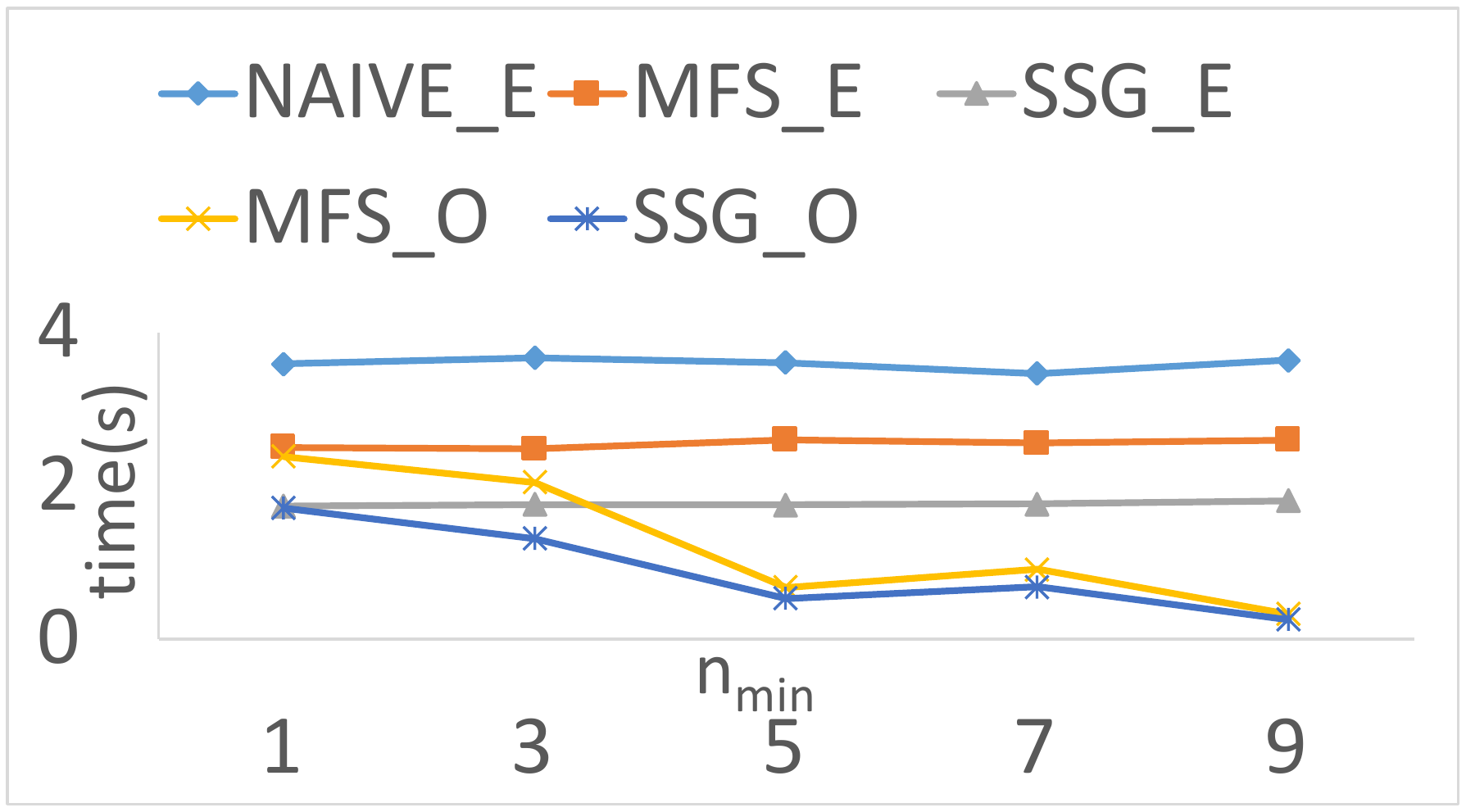}
		\caption{M1}
		\label{fig:state-oeval-m1}
	\end{subfigure}
	\hfill
	\begin{subfigure}[b]{0.22\textwidth}
	\centering
	\includegraphics[width=\textwidth]{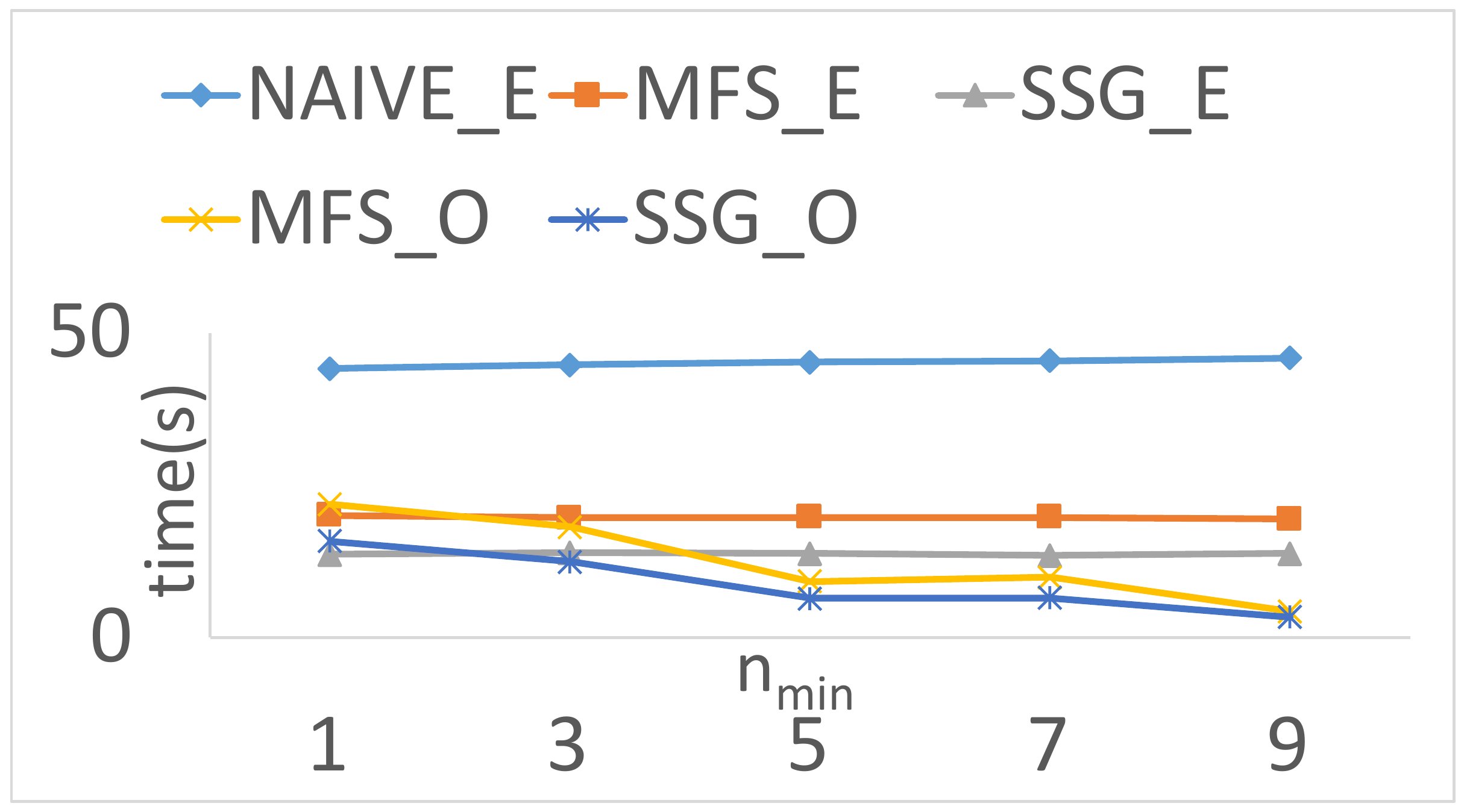}
	\caption{M2}
		\label{fig:state-oeval-m2}
	\end{subfigure}
	\caption{Varying $n_{min}$ in $\ge$ Queries}
	\label{fig:state-oeval-real}
\end{figure}

We finally present the performance of the three methods in an end-to-end manner, as shown in Figure \ref{fig:e2e-all}.
We measure the average time per query (issuing 50 queries and computing the average) in seconds for each dataset (the lower the better), including the object detection and tracking time. As can be observed, both MFS and SSG have leading performance, among which SSG has the best performance overall.

\begin{figure}[]
	\centering
		\centering
		\includegraphics[width=0.33\textwidth]{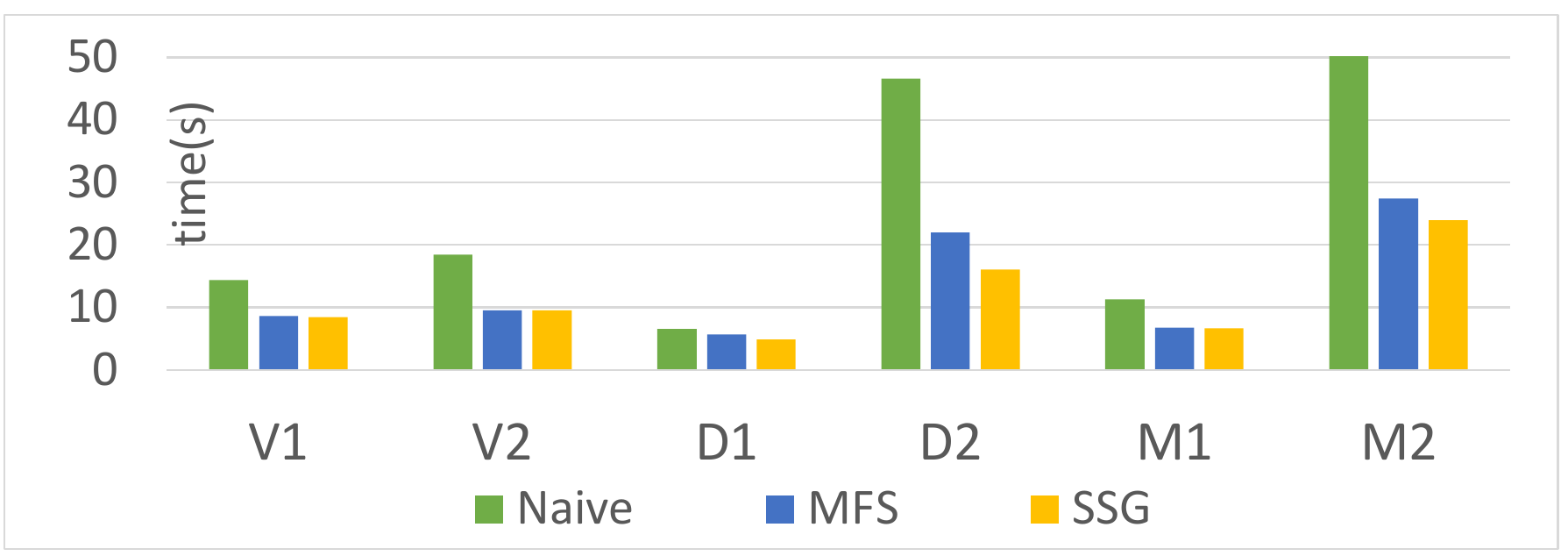}
		\caption{End-to-end Evaluation Time}
		\label{fig:e2e-all}
\end{figure}
\section{Related Work} \label{sect:related}

Several pieces of recent work focus on different aspects of declarative query processing over video feeds \cite{Kang:2017:NON:3137628.3137664,DBLP:conf/cidr/KangBZ19,blazeit}. Kang et. al. \cite{Kang:2017:NON:3137628.3137664} present an approach to filter frames via inexpensive filters for query processing purposes. Subsequently they present query processing techniques for specific types of aggregates over the video feed  \cite{blazeit,DBLP:conf/cidr/KangBZ19}. On a related thread, recent declarative query processing frameworks support queries with spatial constraints among objects on video feeds \cite{xarchakos,xarchakos1} as well as interactions among objects in the visible screen \cite{chao20}. 
Other proposals, \cite{Lu:2016:ORP:2987550.2987564,Hsieh18} present system approaches to query processing across video feeds and are concerned mainly with scalability aspects.
In the vision community, several Deep Learning approaches for object detection and classification have been proposed with impressive accuracy \cite{Krizhevsky:2017:ICD:3098997.3065386, DBLP:conf/iccv/Girshick15, DBLP:conf/cvpr/GirshickDDM14, DBLP:journals/corr/SimonyanZ14a, DBLP:conf/iccv/HeGDG17}.  
Although not directly related, some work has been conducted in the context of mining frequent itemsets over sliding windows. The main focus of these works is to produce itemsets  given a support/confidence threshold \cite{ chi2004moment, jiang2006cfi, nori2013sliding}. 

\section{Conclusions} \label{sect:conclusion}

We considered the problem of evaluating CNF temporal queries over video feeds. Utilizing  DL approaches, it is possible to extract valuable context from frames and enable advanced query processing. We introduced the Marked Frame Set and Strict State Graph approaches to maintain context that aims to reduce the overheads of processing new frames and facilitate CNF query evaluation. 
We demonstrate via experiments that our approaches yield vast performance improvements during query execution over alternate approaches.

\bibliographystyle{abbrv}
\bibliography{reference}  

\newpage
\appendix
\section{Table of Notations}

\begin{table}[!htbp]
    \centering
    \begin{tabular}{|c|c|}
        \hline
        Notation &  Definition \\
        \hline
        $V$ & A video feed, $V = \langle f_0, \ldots,  f_{i}, \ldots, f_{N-1}\rangle$  \\
        \hline
        $VR$ & Structured relation \\
        & which is obtained from video feed,$V$ \\
        \hline
        
        $\mathbb{L}$ & A set of class labels for all objects \\
        \hline
        $\mathbb{ID}$ & A set of objects  \\
        \hline
        $\mathbb{F}$ & A set of frames  \\
        \hline
        $\mathbb{KF}$ & A set of key frames \\
        \hline
        $\mathbb{S}$ & A set of states  \\
        \hline
        $\mathbb{PS}$ & A set of principal states  \\
        \hline
        $\mathbb{E}$ & A set of edges  \\
        \hline
        $w$ & Window size  \\
        \hline
        $d$ & Duration threshold  \\
        \hline
    \end{tabular}
    \caption{Table of Notations}
    \label{tab:my_label}
\end{table}

\section{Proofs}

\setcounter{theorem}{0}

\begin{theorem}
The set of marked frames in each state is a Key Frame Set.
\end{theorem}


\begin{proof}
    For each $s \in \mathbb{S}$, we use $f_m$ to denote the maximum frame id, i.e. $f_m=max(\mathbb{F}_{s})$. 
    $f_m$ must be one of the following cases:
        \begin{enumerate}
            \item $f_m$ is marked. This means the object set can be directly generated by frame $f_m$. Assume the frame set is marked properly before frame $f_m$ arrives. We use $\mathbb{KF'}_{s}$ to denote the key frame set. Since $f_m$ can generate the object set at frame $f_m$ directly, $\mathbb{KF'}_{s}$ is no longer a valid key frame set. According to the definition, $f_m$ has to be added to it. Since $f_m$ is marked, the new marked frame set is a valid key frame set.
            \item $f_m$ is not marked. 
            Then $s$ must can be generated by two or more states.
            We use $\mathbb{S}_s \subseteq \mathbb{S}$ denote the set of states that generates $s$. Assume all the marked frame set in states in $\mathbb{S}_s$ are key frame sets. We use $\mathbb{KF}_{s}$ to denote the key frame set of state $s$.
            Then we have $\mathbb{KF}_s = \bigcup\limits_{ss \in \mathbb{S}_s} {\mathbb{KF}_{ss}}$ , $\mathbb{ID}_s = \bigcap\limits_{ss \in \mathbb{S}_s} {\mathbb{ID}_{ss}}$.
            
            Suppose the marked frame set of the new state is not a key frame set. Then $\exists kf_s \in \mathbb{KF}_s$ which is not a key frame. According to the definition, we can still obtain $\mathbb{ID}_s$ if we remove $kf_s$ from the key frame set. However, $\exists ss' \in \mathbb{S}_s$, $kf_s \in \mathbb{KF}_{ss'}$, which means $kf_s$ can be removed from the marked frame set of $ss'$ which is a  contradiction.

            Suppose a key frame $kf_s$ is missing in the generated state, Then $\mathbb{ID}_s$ can still be generated if we remove all of the marked frames.
            Assume $\forall ss \in \mathbb{S}_{s}$, we remove the marked frames from its frame set and use $\mathbb{ID'}_{ss}$ to denote the object set generated by the rest of the frames. Then we have $\mathbb{ID}_{s} =  \bigcap\limits_{ss \in \mathbb{S}_s} {\mathbb{ID'}_{ss}}$. However, for each $ss$, we have $\mathbb{ID}_{ss} \subseteq \mathbb{ID'}_{ss}$. Which means, $\exists ss'$, $\mathbb{ID'}_{ss'} = \mathbb{ID}_{ss}$. Thus the marked frame set is not a key frame set for some state $ss'$ which is a contradiction.
        \end{enumerate}

\end{proof}

\begin{theorem}
Let $\mathbb{PS}$ be the set of principal states and $ns$ be a new principal state. In an SSG, a state $s \in \mathbb{S}$ could be adjacent to $ns$ only if $\exists ps \in \mathbb{PS}$, such that $\mathbb{ID}_{ns} \cap \mathbb{ID}_{ps} = \mathbb{ID}_{s}$.
\end{theorem}

\begin{proof}
    Assume $s$ is adjacent to $ns$, and $\forall PS \in \mathbb{PS}$, $\mathbb{ID}_{PS} \cap \mathbb{ID}_{ns} \ne \mathbb{ID}_{s}$. 
    Then there must be some principal state $PS' \in \mathbb{PS}$, such that $\mathbb{ID}_{s} \subset \mathbb{ID}_{PS'} \cap \mathbb{ID}_{ns}$. 
    For a State Graph, $s$ must can be reached from some principal state $PS'$. Since $s$ is adjacent to $ns$, $s$ can be generated by intersecting objects between $ns$ and other existing states.
    Assume there are some state $s'$ with object set $\mathbb{ID}_{ns} \cap \mathbb{ID}_{PS'}$. Then $s'$ must also be adjacent to $ns$.
    However, since $s'$ and $s$ both are adjacent to $ns$, while $\mathbb{ID}_{s} \subset \mathbb{ID}_{s'}$, which violates the property \ref{prop:ssg}.
\end{proof}

\begin{theorem}
The State Selection Procedure will not violate Property \ref{prop:ssg} while all newly generated states can be reached from the new principal state.
\end{theorem}

\begin{proof}
    Assume it breaks the property \ref{prop:ssg}. $\exists s_i, s_j \in \mathbb{S}$ that are connected in the graph through one or more edges. Suppose $s_j$ can be reached from $s_i$. Then we have $\mathbb{ID}_{s_j} \subset \mathbb{ID}_{s_i}$. Thus, $|\mathbb{ID}_{s_j}| < |\mathbb{ID}_{s_i}|$. Since we sort $C$ before selection, the only case is $s_j \notin CS_{s_i}$. However, $CS_{s_i}$ is obtained using DPS, which means, $s_j$ can not be reached from $s_i$, which leads to a contradiction.

    Assume some edge $(ns, s_i)$ is missing, where $s_i \in \mathbb{S}$. 
    Then $\exists s_j$, such that $s_i \in CS_{s_j}$.
    Since $s_i \in CS_{s_j}$, it is reachable from state $s_j$, which is a contradiction.
    
\end{proof}

\begin{theorem}
    Using the State Marking Procedure, a state is valid in SSG iff at least one frame in its frame set is marked. 
\end{theorem}

\begin{proof}
    When there is only one principal state, the marked frame set is always correct. 
    Assume after processing frame $i$ (which principal state is $ns$), all valid states have at least one frame marked.
    According to the State Marking Procedure, after processing a new arriving frame $ns$, all states that are principal states or generated by some principal states directly will be marked.
    
    Suppose at frame $i'$ with principal state $ns'$, for some valid state $s$, no frame is marked. Then $s$ must be generated by principal states indirectly.
    We use $S_{s}$ to denote the set of states that are adjacent to $s$.
    Since $s$ is valid, $s$ can be generated by computing intersections between some existing state, $s'$ and $ns'$. By applying this rule repeatedly, eventually, we have some state $s''$, which is generated by some existing principal state $ps''$. Since at least one frame is marked in $s''$, so does $s'$ and $s$, which is a contradiction to the assumption. 
    
    If a state, $s$, is invalid. Then starting from all the principal states from the current window on the SSG, state $s$ can not be generated. State $s$ must be generated from some previous window with some marked frame, $i'$, expired. Since $s$ can no longer be generated in the current window, the marked frames of $s$ will not be updated. Therefore, no frame is marked in state, $s$, which is invalid and will be pruned.

\end{proof}

\begin{theorem}
    In SSG, there will be at most $2^{x}$ number of states and $x2^{x}$ number of edges. Thus, the complexity of algorithm ST is $O(x2^{x})$, where $x=\frac{w}{\lambda}$.
\end{theorem}

\begin{proof}
Suppose the number of frames that share the same object set is $\lambda$ on average.
Given window size $w$, $\frac{w}{\lambda}$ number of unique principal states will be generated. 
We use $x=\frac{w}{\lambda}$ to denote the number of unique principal states.
We use $N_O$ to denote the number of objects in each frame on average.

We first consider the case where all principal states share some objects in common. In this case, if we revert the direction of each edge, we can obtain a tree, where the root is the state with the common object set among all frames, and the leaves are principal states.

In the worst case, starting from $N_O$, every possible state should be generated with different object set sizes.
Suppose a new state can be generated by any two states, then ${x}\choose{2}$ number of states can be generated with object size $N_O-1$. The maximum number of states would be ${x}\choose{1}$ $+$ ${x}\choose{2}$ $+ ... +1=2^{x}=2^{\frac{w}{\lambda}}$.
For each principal state with object set size of $size$, there will be at most $size -1$ number of edges. 
Thus in the worst case, the total number of edges would be ${x}\choose{1}$ $*(N_O-1) + $ ${x}\choose{2}$ $*(N_O-2) +... +$ ${x}\choose{x-1}$ $ * 1= x!(\frac{N_O-1}{x-1} + \frac{N_O-2}{(x-2)!2!} +...+\frac{1}{(x-1)!})$.
Normally, $\frac{w}{\lambda} > N_O$. Thus, we have the upper bound $x2^{x}=\frac{w}{\lambda}2^{\frac{w}{\lambda}}$.

Since we are using the DFS algorithm as the graph traversal algorithm, the complexity of the algorithm is
$O(V+E)$=$O(\frac{w}{\lambda}2^\frac{w}{\lambda})$.
    
\end{proof}

\end{document}